\documentclass{lmcs}
\pdfoutput=1

\usepackage{lastpage}
\lmcsdoi{21}{2}{10}
\lmcsheading{}{\pageref{LastPage}}{}{}%
{May~28,~2024}{Apr.~29,~2025}{}


\usepackage{hyperref}
\usepackage{booktabs}
\usepackage[utf8]{inputenc}
\usepackage{definitions}

\begin{document}

\title[Left-Linear Completion with AC Axioms]{Left-Linear Completion with
AC Axioms}
\thanks{The research described in this article is partly supported by JSPS
KAKENHI Grant Number JP22K11900.}

\author[J.~Niederhauser]{Johannes Niederhauser%
\lmcsorcid{0000-0002-8662-6834}}[a]
\author[N.~Hirokawa]{Nao Hirokawa\lmcsorcid{0000-0002-8499-0501}}[b]
\author[A.~Middeldorp]{Aart Middeldorp\lmcsorcid{0000-0001-7366-8464}}[a]

\address{Department of Computer Science, University of Innsbruck,
Innsbruck, Austria}
\email{johannes.niederhauser@uibk.ac.at, aart.middeldorp@uibk.ac.at}

\address{School of Information Science, JAIST, Nomi, Japan}
\email{hirokawa@jaist.ac.jp}

\begin{abstract}
\noindent We revisit completion modulo equational theories for
left-linear term rewrite systems where unification modulo the theory
is avoided and the normal rewrite relation can be used in order to
decide validity questions. To that end, we give a new correctness
proof for finite runs and establish a simulation result between the
two inference systems known from the literature. Given a concrete
reduction order, novel canonicity results show that the resulting
complete systems are unique up to the representation of their rules'
right-hand sides. Furthermore, we show how left-linear AC
completion can be simulated by general AC completion. In particular,
this result allows us to switch from the former to the latter at any
point during a completion process.
\end{abstract}

\maketitle

\section{Introduction}

Completion has been extensively studied since its introduction in the
seminal paper by Knuth and Bendix \cite{KB70}. One of the main
limitations of the original formulation is its inability to deal with
equations which cannot be oriented into a terminating rule such as the
commutativity axiom. This shortcoming can be resolved by completion
modulo an equational theory $\xE$. In the literature, there are two
different approaches of achieving this. The general approach
\cite{JK86,B91} requires $\xE$-unification and allows us to decide
validity problems using the rewrite relation $\Rb{\xR/\xE}$ which is
defined as $\LRab{\xE}{*} \cdot \Rab{\xR}{} \cdot \LRab{\xE}{*}$. For
left-linear term rewrite systems, however, there is Huet's approach
\cite{H80} which avoids $\xE$-unification. In particular, Huet's
approach does not consider local peaks modulo $\xE$
($\LbR \cdot \sim_\xE \cdot \RbR$) but works with ordinary
local peaks ($\LbR \cdot \RbR$) as well as local cliffs of
the form $\LRb{\xE} \cdot \RbR$. Hence, instead of
$\xE$-critical pairs we can use normal critical pairs when we also
take overlaps between $\xR$ and $\xE$ into account. We call this
approach \emph{left-linear completion modulo an equational theory}.
If we have a complete TRS $\xR$ in the general sense, we can decide
validity problems $s \approx t$ by rewriting both terms to normal
form with $\Rb{\xR/\xE}$ and then checking the result for
$\xE$-equivalence. For complete systems stemming from left-linear
$\xE$-completion, however, it suffices to rewrite both $s$ and $t$
to normal form using the normal rewrite relation $\RbR$ and then
perform just one $\xE$-equivalence check on these normal forms.

In their respective books, Avenhaus \cite{A95} and Bachmair \cite{B91}
present inference systems for left-linear completion modulo an
equational theory. This article gives a detailed account of the
nature of and relation between these two systems for finite
runs. For the concrete case of AC (associative and commutative
function symbols), we compare left-linear completion modulo
equational theories with the general approach by presenting an
implementation of left-linear AC completion and comparing it with
the state of the art concerning general AC completion.
After setting the stage in \secref{preliminaries}, we present a new
criterion for the Church--Rosser modulo property of left-linear TRSs
based on prime critical pairs in \secref{cc}. Slightly modified
versions (\ia and \ib) of the inference systems due to Avenhaus and
Bachmair are discussed in Sections~\ref{sec:avenhaus} and
\ref{sec:bachmair}, respectively. Both sections include a new
correctness proof of the given inference system for finite runs. For
\ia, this is done from scratch by using the criterion from \secref{cc}
in the spirit of \cite{HMSW19}. Correctness of \ib is then reduced to
the correctness of \ia by establishing a simulation result between
finite runs in these systems. Furthermore, \secref{canonicity}
reports on novel results on canonicity for this setting. For the
concrete equational theory of associative and commutative (AC)
function symbols, we also show the connection between the inference
system \ia and general AC completion by means of another simulation
result (\secref{ac-completion}). Finally, we describe our
implementation of \ia in the tool \accompll and present experimental
results which show that the avoidance of AC unification and AC
matching can result in significant performance improvements over
general AC completion (Sections~\ref{sec:implementation} and
\ref{sec:experimental-results}).

This article extends our previous papers \cite{NHM23a} and
\cite{NHM23b} by including full proof details, more examples
and experimental data as well as novel results on canonicity.

\section{Preliminaries}
\label{sec:preliminaries}

We assume familiarity with term rewriting and completion as described
e.g.\ in \cite{BN98,T03} but recall some central notions
in \secref{rewriting}. The concept of (prime) critical peaks which
leads to the definition of (prime) critical pairs is introduced in
\secref{cpeaks}. Finally, \secref{rewritingmodulo} provides the
necessary background for rewriting modulo equational theories.

\subsection{Rewrite Systems}
\label{sec:rewriting}  

An \emph{abstract rewrite system} or \emph{abstract reduction system}
(ARS) $\xA = \langle A,{\R} \rangle$ is a set $A$ together with a
binary relation $\R$ on $A$. If $a \R b$ for no $b$, then $a$ is
called a \emph{normal form} of $\xA$. Otherwise, we say that $a$
is \emph{reducible}. The set of normal forms is denoted by
$\nf{\xA}$. Given an arbitrary binary relation $\R$, we write $\L$,
$\LR$, $\R^=$, $\R^+$ and $\R^*$ to denote its \emph{inverse}, its
\emph{symmetric closure}, its \emph{reflexive closure}, its
\emph{transitive closure} and its \emph{symmetric and transitive closure},
respectively. Hence, the relation $\LR^*$ denotes the
symmetric, reflexive and transitive closure of $\R$ and is called
\emph{conversion}. The relation $a \R^! b$ is defined as
$a \R^* b$ and $b \in \nf{\xA}$ and is used to denote rewriting
to a normal form. Finally, $\D$ abbreviates
the \emph{joinability relation} $\R^* \cdot \La{*}$ where $\cdot$
denotes the composition of binary relations
which is defined as follows: $R_1 \cdot R_2 =
\SET{(a,c) \mid \text{$(a,b) \in R_1$ and $(b,c) \in R_2$}}$.

Given a signature $\xF$ and a set of variables $\xV$, we consider the
set of terms $\xT(\xF,\xV)$ which is defined as usual. The set of
variables in a term $t$ is written as $\var{t}$. Terms which do not
contain the same variable more than once are referred to as
\emph{linear} terms. Each subterm of a term $t$ has a unique
\emph{position} which is a finite sequence of positive integers where
the empty sequence representing the root position is written as
$\epsilon$. The set of positions in a term $t$ is denoted by $\pos{t}$
and further divided into the subset $\posf{t}$ of positions
which address function symbols and the subset
$\posv{t} = \pos{t} \setminus \posf{t}$ of variable positions.
If a position $p$ is a
prefix of the position $q$ we write $p \leq q$. Positions $p$ and $q$
are parallel, denoted by $p \parallel q$, if neither $p \leq q$
nor $q \leq p$. If $p \leq q$ then $q \setminus p$ denotes the unique
position $r$ such that $pr = q$. By $s|_p$ we denote the subterm
of $s$ at position $p$. Mappings $\sigma$ from variables to
terms with a finite domain ($\SET{x \in \xV \mid x \neq \sigma(x)}$) are
called \emph{substitutions}. The application of a substitution $\sigma$
to a term $t$ is denoted by $t\sigma$. A \emph{renaming} is a
bijective substitution from $\xV$ to $\xV$. A term $s$ is a
\emph{variant} of a term $t$ ($s \doteq t$) if there exists a renaming
$\sigma$ such that $s = t\sigma$. A context $C$ is a term with
exactly one occurrence of the special symbol $\Box \notin \xF \cup \xV$
called \emph{hole} which acts as a placeholder for concrete
terms. Replacing the hole with a term $t$ results in a term which we
denote by $C[t]$. A term $s$ \emph{encompasses} a term $t$
($s \enceq t$) if $s = C[t\sigma]$ for some context $C$ and
substitution $\sigma$. 
It is known that $\enceq$ is a quasi-order on terms and its strict
part $\enc$ is a well-founded order with
${\enc} = {\enceq} \setminus {\doteq}$. We write $s[t]_p$ for the term
which is created from $s$ by replacing its subterm at position $p$ by $t$.

A pair of terms $(s,t)$ can be viewed as an \emph{equation} ($s \approx t$)
or a \emph{rule} ($s \R t$). In the latter case we assume that $s$ is not
a variable and $\var{s} \supseteq \var{t}$. \emph{Equational systems}
(ESs) are sets of equations while \emph{term rewrite systems} (TRSs)
are sets of rules. Given a set of pairs of terms $\xE$, we define a
\emph{rewrite relation} $\R_{\xE}$ as the closure of its pairs under
substitutions and contexts. More formally, $s \Rb{\xE} t$ if there is
a pair $(\ell,r) \in \xE$, a position $p$ and a substitution $\sigma$
such that $s|_p = \ell\sigma$ and $t = s[r\sigma]_p$.
We sometimes make the position $p$ explicit by writing
$s \Rab{\xE}{p} t$ and define $\LRab{\xE}{p}$ as
$\Lab{\xE}{p} \cup \Rab{\xE}{p}$. The
\emph{equational theory} induced by $\xE$ consists of all pairs of
terms $(s,t)$ such that $s \LRab{\xE}{*} t$. A TRS $\xR$
\emph{represents} an ES $\xE$ if ${\LRab{\xE}{*}} = {\LRab{\xR}{*}}$. Two
rules $\ell \to r$ and
$\ell' \to r'$ are \emph{variants} if there exists a renaming $\sigma$
such that $\ell\sigma = \ell'$ and $r\sigma = r'$. Two TRSs $\xR_1$
and $\xR_2$ over the same signature $\xF$ are \emph{literally similar}
($\xR_1 \doteq \xR_2$) if every rule in $\xR_1$ has a variant in $\xR_2$
and vice versa.
A TRS is
\emph{left-linear} if $\ell$ is a linear term for every rule
${\ell \R r} \in \xR$. Given an ES $\xE$, $\xE^\pm$ denotes
$\xE \cup \SET{t \approx s \mid {s \approx t} \in \xE}$.

A TRS $\xR$ is \emph{terminating} if there is no infinite rewrite
sequence $s_1 \RbR s_2 \RbR \cdots$. If
${\Lab{\xR}{*} \cdot \Rab{\xR}{*}} \subseteq {\Rab{\xR}{*} \cdot
\Lab{\xR}{*}}$ then $\xR$ is \emph{confluent}. A TRS is
\emph{complete} if it is terminating and confluent. For standard
rewriting, the \emph{Church--Rosser property}
(${\LRab{\xR}{*}} \subseteq {\Db{\xR}}$) coincides with confluence.
Hence, complete presentations $\xR$ of an ES $\xE$ can be used to
decide the validity problem for $\xE$: $s \LRab{\xE}{*} t$ if and only
if $s \Rab{\xR}{!} \cdot \Lab{\xR}{!} t$.

\subsection{Critical Peaks}
\label{sec:cpeaks}

Confluence of terminating TRSs is characterized by joinability of
critical pairs. Critical pairs are computed from overlaps which
are the potentially dangerous cases of nondeterminism in rewriting
as long as termination holds.

\begin{defi}
Let $\xR$ be a TRS. An \emph{overlap} is a triple
$\langle {\ell_1 \R r_1}, p, {\ell_2 \R r_2} \rangle$
satisfying the following properties:
\begin{itemize}[label=$\triangleright$]
\item
$\ell_1 \R r_1$ and $\ell_2 \R r_2$ are variants of rewrite rules of
$\xR$ without common variables,
\smallskip
\item
$p \in \posf{\ell_2}$,
\smallskip
\item
$\ell_1$ and $\ell_2|_p$ are unifiable and
\smallskip
\item
if $p = \epsilon$ then $\ell_1 \R r_1$ and $\ell_2 \R r_2$ are not
variants.
\end{itemize}
\end{defi}
\noindent 
Overlaps give rise to \emph{critical peaks} from which the critical
pairs can then be extracted.

\begin{defi}
Let $\langle {\ell_1 \R r_1}, p, {\ell_2 \R r_2} \rangle$ be an
overlap of a TRS $\xR$ and $\sigma$ a most general unifier of
$\ell_1$ and $\ell_2|_p$. The term
$\ell_2\sigma[\ell_1\sigma]_p = \ell_2\sigma$ can be rewritten in
two different ways which gives rise to a critical peak, here
illustrated graphically:
\begin{center}
\begin{tikzpicture}
\node (1) {$\ell_2\sigma[\ell_1\sigma]_p = \ell_2\sigma$};
\node (2) [below left=of 1] {$\ell_2\sigma[r_1\sigma]_p$};
\node (3) [below right=of 1] {$r_2\sigma$};
\ArrowOS{(1)}{(2)}{$p$}{$r_1 \L \ell_1$};
\Arrow{(1)}{(3)}{$\epsilon$}{$\ell_2 \R r_2$};
\end{tikzpicture}
\end{center}
More formally, a critical peak is a quadruple
$\langle \ell_2\sigma[r_1\sigma]_p, p, \ell_2\sigma, r_2\sigma \rangle$
and the equation $\ell_2\sigma[r_1\sigma]_p \approx r_2\sigma$ is a
critical pair of $\xR$ obtained from the original overlap.
\end{defi}

Usually, we denote a critical peak $\langle t, p, s, u \rangle$ more
illustratively by \smash{$t \Lo{}{p} s \Ro{}{\epsilon} u$}. The set of all
critical pairs of a TRS $\xR$ (obtained from all possible overlaps) is
denoted by $\cp{\xR}$.
Knuth and Bendix' criterion~\cite{KB70} states that a terminating TRS
$\xR$ is confluent if and only if $\cp{\xR} \subseteq {\Db{\xR}}$.
Kapur et al.~\cite{KMN88} showed that joinability of prime
critical pairs still guarantees confluence for terminating TRSs.

\begin{defi}
A critical peak \smash{$t \Lo{}{p} s \Ro{}{\epsilon} u$} is \emph{prime}
if all proper subterms of $s|_p$ are in normal form. Critical pairs
derived from prime critical peaks are called prime. The set of all prime
critical pairs of a TRS $\xR$ is denoted by $\pcp{\xR}$.
\end{defi}

\begin{thmC}[\cite{KMN88}]
\label{thm:KMN88}
Let $\xR$ be a terminating \textup{TRS}. The \textup{TRS} $\xR$ is
confluent if and only if $\pcp{\xR} \subseteq {\Db{\xR}}$. \qed
\end{thmC}

\begin{exa}
\label{exa:pcp}
Using \thmref{KMN88}, we show confluence of the following TRS $\xR$:
\begin{align*}
\m{f}(\m{a} + x) &\cR x &
\m{f}(x + \m{a}) &\cR x &
\m{f}(\m{b} + x) &\cR x &
\m{f}(x + \m{b}) &\cR x &
\m{a} &\cR \m{b}
\end{align*}
Since every rewrite step reduces the total number of $\m{f}$ and $\m{a}$,
termination of $\xR$ follows. Observe that $\xR$ admits six (modulo
symmetry) critical peaks of the form 
$t \Lab{\xR}{p} s \Rab{\xR}{\epsilon} u$:
\begin{align*}
&  \peak{\underline{\m{f}(\m{a} + \m{a})}}{\m{a}}{\m{a}}
&& \peak{\underline{\m{f}(\m{a} + \m{b})}}{\m{a}}{\m{b}}
&& \peak{\underline{\m{f}(\m{b} + \m{a})}}{\m{a}}{\m{b}}
&& \peak{\underline{\m{f}(\m{b} + \m{b})}}{\m{b}}{\m{b}}
&& \peak{\m{f}(\underline{\m{a}} + x)}{\m{f}(\m{b} + x)}{x}
&& \peak{\m{f}(x + \underline{\m{a}})}{\m{f}(x + \m{b})}{x}
\end{align*}
Here the positions $p$ in $s$ are indicated by underlining. The first
three critical peaks are not prime due to the reducible proper subterm
$\m{a}$ in $s|_p$, while the others are prime. Therefore,
$\pcp{\xR} = \SET{\m{b} \approx \m{b},\, \m{f}(\m{b} + x) \approx x,\,
\m{f}(x + \m{b}) \approx x}$. Since $\pcp{\xR} \subseteq {\Db{\xR}}$
holds, $\xR$ is confluent.
\end{exa}

\subsection{Rewriting Modulo}
\label{sec:rewritingmodulo}

We now turn our attention to rewriting modulo an equational theory. To
that end, we start by giving general definitions for abstract rewrite
systems (ARSs). Let $\xA = \langle A,{\R} \rangle$ be an ARS and
$\sim$ an equivalence relation on $A$. We write 
$\bigLR$ for $\L \cup \R \cup \sim$ (conversion modulo $\sim$),
${\R}/{\sim}$ for $\sim \cdot \R \cdot \sim$ (rewriting modulo $\sim$)
and $\Da{\sim}$ for $\Ra{*} \cdot \sim \cdot \La{*}$
(valley modulo $\sim$). In order to simplify our terminology,
we sometimes refer to conversions modulo $\sim$ as conversions. Hence,
whether conversions include $\sim$-steps or not depends on the type of
rewriting we consider. Given $\xA$, we
denote $\langle A,{{\R}/{\sim}} \rangle$ by ${\xA}/{\sim}$. The ARS
$\xA$ is \emph{terminating modulo $\sim$} if there are no infinite
rewrite sequences with ${\R}/{\sim}$ and
\emph{Church--Rosser modulo $\sim$} if
${\bigLRa{*}} \subseteq {\Da{\sim}}$.

Confluence and the
Church--Rosser property do not coincide for rewriting modulo
equational theories \cite{O98}. Therefore, a notion of completeness
for this setting has to depend on the Church--Rosser property instead
of confluence in order to facilitate a decision procedure for validity
problems of equational theories with complete presentations.
Hence, we define an ARS $\xA$ to be \emph{complete modulo $\sim$}
if it is terminating modulo
$\sim$ and Church--Rosser modulo $\sim$. While there is no distinction
for termination modulo $\sim$ between $\xA$ and ${\xA}/{\sim}$
(${\sim \cdot \sim} = {\sim}$ by transitivity), it makes a
considerable difference whether we talk about the Church--Rosser
modulo $\sim$ property and therefore completeness modulo $\sim$ of
$\xA$ or ${\xA}/{\sim}$.

\begin{exa}
\label{exa:cr-a-vs-asim}
Consider the ARS $\xA$ together with the equivalence relation $\sim$
defined as follows:
\[
\m{a} \sim \m{b} \R \m{c}
\]
Since $\m{a} \bigLRa{*} \m{c}$ and
$\m{a} \mathrel{{\R}/{\sim}} \m{c}$, the ARS ${\xA}/{\sim}$ is
Church--Rosser modulo $\sim$. However, $\xA$ is not
Church--Rosser modulo $\sim$ as $\m{a}$ and $\m{c}$
are different normal forms of $\R$ which are not
equivalent in $\sim$.
\end{exa}

The following lemma is taken from
\cite[Lemma~4.1.12]{A95}. It establishes an important connection
between the Church--Rosser modulo $\sim$ property of an ARS $\xA$ and
${\xA}/{\sim}$. 

\begin{lemC}[\cite{A95}]
\label{lem:crconnection}
Let $\xA = \langle A,{\R} \rangle$ and
$\xA' = \langle A,{\rightharpoonup} \rangle$ be \textup{ARS}s and
$\sim$ an equivalence relation on $A$ such that
${\R} \subseteq {\rightharpoonup} \subseteq {{\R}/{\sim}}$. If
$\xA'$ is Church--Rosser modulo $\sim$ then ${\xA}/{\sim}$ is
Church--Rosser modulo $\sim$. \qed
\end{lemC}
\noindent 
Note that the implication cannot be strengthened to an equivalence due
to \exaref{cr-a-vs-asim}.  The following example illustrates a successful
use of \lemref{crconnection}.

\begin{exa}
Consider the ARS $\xA = \langle A,{\R} \rangle$ from \exaref{cr-a-vs-asim}.
For $\xA' = \langle A, {\rightharpoonup} \rangle$ with
${\rightharpoonup} = {\sim \cdot \R}$ we can easily establish
the Church--Rosser modulo $\sim$ property. Hence, by
\lemref{crconnection} also ${\xA}/{\sim}$ is Church--Rosser
modulo $\sim$.
\end{exa}

The definitions and results for ARSs carry over to TRSs by replacing
the equivalence relation $\sim$ by the equational theory
$\LRab{\xB}{*}$ of an ES $\xB$. Most theoretical results of this
article are not specific to AC but hold for an arbitrary base theory
$\xB$ of which we only demand that $\var{\ell} = \var{r}$ for all
${\ell \approx r} \in \xB$. We abbreviate $\LRab{\xB}{*}$ by $\bsim$
and the rewrite relation $\RmB$ is defined as
$\bsim \cdot \RbR \cdot \bsim$. Furthermore, we write
$\Dab{\xR}{\sim}$ for the relation
$\Rab{\xR}{*} \cdot \bsim^{} \cdot \Lab{\xR}{*}$. Note
that the omission of $\xB$ in the notation $\Dab{\xR}{\sim}$ poses no
problem as $\xB$ is usually fixed and can be inferred from
the context in all other cases.
Termination modulo
$\xB$ is shown by \emph{$\xB$-compatible reduction orders} $>$, i.e.,
$>$ is well-founded, closed under contexts and substitutions and
${\bsim \cdot > \cdot \bsim} \subseteq {>}$. Note that $\RmB$
is not a very practical rewrite relation: It is undecidable in
general, and even if the equational theory of $\xB$ is decidable,
rewriting a term $t$ requires to check every member of its
$\xB$-equivalence class. A more practical alternative due to Peterson
and Stickel \cite{PS81} is the relation
$\Rb{\xR,\xB}$ defined as follows:
$s \Rb{\xR,\xB} t$ if there exist a rule $\ell \R r \in \xR$,
a substitution $\sigma$ and a position $p$ such that
$s|_p \bsim \ell\sigma$ and $t = s[r\sigma]_p$. This definition is
very similar to the definition of standard rewriting but with
$\xB$-matching instead of normal matching. It is immediate from the
respective definitions that the inclusions
${\RbR} \subseteq {\Rb{\xR,\xB}} \subseteq {\RmB}$ hold.

\begin{exa}
Consider the TRS
$\xR$ consisting of the rules
\begin{align*}
\m{0} + y &\R y &
\m{s}(x) + y &\R \m{s}(x + y) 
\end{align*}
where $+$ is an AC symbol. We have $y + \m{s}(x) \RcAC \m{s}(x + y)$ as
$y + \m{s}(x) \acsim \m{s}(x) + y$ but $y + \m{s}(x)$ is a normal form with
respect to $\RbR$. Furthermore, $(y + z) + \m{s}(x) \RmAC \m{s}(x + y) + z$
since $(y + z) + \m{s}(x) \acsim (\m{s}(x) + y) + z$ but this step is not
possible with $\RcAC$ as the rewrite step takes place at position $1$
whereas we need to search for AC equivalent terms at the root position.
\end{exa}

\section{Church--Rosser Criterion}
\label{sec:cc}

In this section we present a new characterization of the Church--Rosser
property modulo an equational theory $\xB$ for left-linear TRSs which are
terminating modulo $\xB$. The original left-linear completion
procedure~\cite{A95,B91} relies on the following 
theorem. We use critical pairs between two different TRSs $\xR_1$ and
$\xR_2$ with the same signature. We write $\cp{\xR_1,\xR_2}$ for the set
of all critical pairs originating from overlaps
$\langle \rho_1, p, \rho_2 \rangle$ where $\rho_1 \in \xR_1$ and
$\rho_2 \in \xR_2$.  The union of $\cp{\xR_1,\xR_2}$ and
$\cp{\xR_2,\xR_1}$ is denoted by $\cppm{\xR_1,\xR_2}$.

\begin{thmC}[\cite{H80}]
\label{thm:accp}
A left-linear \textup{TRS} $\xR$ which is terminating modulo $\xB$
is Church--Rosser modulo $\xB$ if and only if
$\cp{\xR} \cup \cppm{\xR,\xB^\pm} \subseteq {\Dab{\xR}{\sim}}$. \qed
\end{thmC}

\begin{exa}[continued from \exaref{pcp}]
\label{exa:accp}
Let $\xB = \SET{x + y \approx y + x}$. We show that the left-linear
TRS $\xR$ of \exaref{pcp} is Church--Rosser modulo $\xB$. As before, the
termination of $\xR/\xB$ follows from the fact that every rewrite step
reduces the total number of $\m{f}$ and $\m{a}$. The six critical peaks
result in $\cp{\xR} = \SET{
\m{a} \approx \m{a},\,
\m{a} \approx \m{b},\,
\m{b} \approx \m{a},\,
\m{b} \approx \m{b},\,
\m{f}(\m{b} + x) \approx x,\,
\m{f}(x + \m{b}) \approx x}$.
Observe that $\xR$ and $\xB$ admit four critical peaks of the forms 
$t \Lab{\xR}{p} s \LRab{\xB}{\epsilon} u$ or
$t \LRab{\xB}{p} s \Rab{\xR}{\epsilon} u$:
\begin{align*}
&  \peak{\m{f}(\underline{\m{a} + x})}{\m{f}(x + \m{a})}{x}
&& \peak{\m{f}(\underline{x + \m{a}})}{\m{f}(\m{a} + x)}{x}
&& \peak{\m{f}(\underline{\m{b} + x})}{\m{f}(x + \m{b})}{x}
&& \peak{\m{f}(\underline{x + \m{b}})}{\m{f}(\m{b} + x)}{x}
\end{align*}
Thus, $\cppm{\xR,\xB^\pm} = \SET{
\m{f}(x + \m{a}) \approx x,\,
\m{f}(\m{a} + x) \approx x,\,
\m{f}(x + \m{b}) \approx x,\,
\m{f}(\m{b} + x) \approx x
}$. It is easy to see that $t \Dab{\xR}{\sim} u$ for all critical
pairs $t \approx u$ in $\cp{\xR} \cup \cppm{\xR,\xB^\pm}$. 
Hence, by \thmref{accp}, the TRS $\xR$ is
Church--Rosser modulo $\xB$.
\end{exa}

\thmref{accp} allows us to use ordinary critical pairs
instead of $\xB$-critical pairs, i.e., critical pairs stemming from
local peaks modulo $\xB$ ($\LbR \cdot \bsim \cdot \RbR$) \cite{JK86}.
In particular, equational unification
modulo $\xB$ can be replaced by syntactic unification which improves
efficiency. Furthermore, the form of the joining sequence
($\Dab{\xR}{\sim}$) is advantageous as it uses the normal rewrite
relation and just one $\xB$-equality check in the end as opposed to
rewrite steps modulo the theory ($\bsim \cdot \RbR \cdot \bsim$).
However, left-linearity is necessary in \thmref{accp}
as the following example illustrates.

\begin{exa}
Consider the AC-terminating TRS $\xR$ consisting of the single rule
$\m{f}(x,x) \R x$ with $+$ as an additional AC function symbol as well as
the following conversion:
\[
x + y \,\LbR\, \m{f}(x + y,x + y) \,\acsim\, \m{f}(x + y,y + x)
\]
There are no critical pairs in $\xR$ and between $\xR$ and
$\AC^\pm$, so $\cp{\xR} = \cppm{\xR,\AC^\pm} = \varnothing$.
However, $x + y \,\Dab{\xR}{\sim}\, \m{f}(x + y,y + x)$
does not hold because $x + y$ and $\m{f}(x + y,y + x)$ are $\xR$-normal
forms which are not AC equivalent. Thus, $\xR$ is not
Church--Rosser modulo AC.
\end{exa}

Even though most results in this article only apply to
left-linear TRSs, we do not demand that $\xB$ is linear (i.e.,
having equations only consisting of linear terms). The reason
for that can be found in the
proof of \thmref{accp} where left-linearity is crucial but only needed for
rewrite rules: If $\xR$ is not left-linear, the proof reveals that
we also have to consider variable overlaps (i.e., overlaps at variable
positions) between $\xR$ and $\xB$. This would make a confluence analysis
solely based on critical pairs impossible.
In the remainder of this section we show that joinability of
prime critical pairs suffices for the characterization of
\thmref{accp}.

\subsection{Peak-and-Cliff Decreasingness}

We present a new Church--Rosser modulo criterion, dubbed
\emph{peak-and-cliff decreasingness}. This is an extension of
peak decreasingness \cite{HMSW19} which is a simple confluence criterion
for ARSs designed to replace complicated proof orderings in the
correctness proofs of completion procedures. As such,
peak-and-cliff decreasingness will be a crucial ingredient
in the correctness proof we give for Avenhaus' inference system
in \secref{avenhaus}.

In the following, we assume that equivalence relations $\sim$ are
defined as the reflexive and transitive closure of a symmetric
relation $\eqstep$, so ${\sim} = {\eqstep^*}$. We refer to conversions
of the form $\L \cdot \eqstep$ or $\eqstep \cdot \R$ as
\emph{local cliffs} and conversions of the form $\L \cdot \R$ as
\emph{local peaks}. Furthermore, we assume that rewrite and equality
steps are labeled with labels from the same set $I$, so let
$\xA = \langle A, \SET{\Rb{\alpha} \mid \alpha \in I} \rangle$ be an
ARS
and ${\sim} = (\,\bigcup\,\SET{{\eqstep_\alpha} \mid \alpha \in I})^*$
an equivalence relation on $A$. Note that several different steps
can have the same label. Furthermore, for the sake of better readability,
we allow ourselves some freedom of where to annotate our arrow relations
with labels, closure operators (reflexive, transitive, \dots) and the like.
This will cause no confusion.

\begin{defi}
An ARS $\xA$ is \emph{peak-and-cliff decreasing} if there
is a well-founded order $>$ on $I$ such that for all $\alpha, \beta \in I$
the inclusions
\begin{align*}
{\Lb{\alpha} \cdot \Rb{\beta}} &\,\subseteq\,
{\bigLRo{\vee\alpha\beta}{*}} &
{\Lb{\alpha} \cdot \eqstep_\beta} &\,\subseteq\,
{\bigLRo{\vee\alpha}{*} \cdot \Lo{\beta}{=}}
\end{align*}
hold. Here $\Vee\alpha\beta$ denotes the set
$\SET{\gamma \in I \mid \text{$\alpha > \gamma$ or $\beta > \gamma$}}$
and if $J \subseteq I$ then $\Rb{J}$ denotes
$\bigcup\,\SET{{\Rb{\gamma}} \mid \gamma \in J}$. We abbreviate
$\Vee\alpha\alpha$ to $\Vee\alpha$.
\end{defi}

\begin{exa}
Let $I = \SET{1,2}$ and 
consider the following $I$-labeled ARS
$\xA = \langle A, \SET{\Rb{1},\Rb{2}} \rangle$
equipped with the equivalence relation
${\sim} = ({\eqstep_1} \cup {\eqstep_2})^*$:
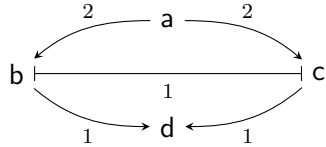
\begin{center}
\begin{tikzpicture}[>=stealth]
\node (a) at (2,1.4) {$\m{a}$};
\node (b) at (0,0.7) {$\m{b}$};
\node (c) at (4,0.7) {$\m{c}$};
\node (d) at (2,0.0) {$\m{d}$};
\draw[->]
 (a) edge[out=180,in=40] node[above] {$\scriptstyle 2$} (b)
 (a) edge[out=0,in=140]  node[above] {$\scriptstyle 2$} (c)
 (c) edge[out=-140,in=0] node[below] {$\scriptstyle 1$} (d)
 (b) edge[out=-40,in=180] node[below] {$\scriptstyle 1$} (d)
;
\draw[|-|]
 (b) edge node[below] {$\scriptstyle 1$} (c)
;
\end{tikzpicture}
\end{center}
Using the well-founded order $2 > 1$ it is easily established that $\xA$
is peak-and-cliff decreasing: For the only local peak 
$\m{b} \Lb{2} \m{a} \Rb{2} \m{c}$
we have $\m{b} \eqstep_1 \m{c}$.
The two local cliffs
$\m{d} \Lb{1} \m{b} \eqstep_1 \m{c}$
and
$\m{d} \Lb{1} \m{c} \eqstep_1 \m{b}$
are handled by 
$\m{d} \Lb{1} \m{c}$ and
$\m{d} \Lb{1} \m{b}$, respectively.
\end{exa}

We show that peak-and-cliff decreasingness is a sufficient condition
for the Church--Rosser modulo property.

\begin{lem}
\label{lem:acconvpeakcliff}
Every conversion modulo $\sim$ is
a valley modulo $\sim$ or contains a local peak or cliff:
\[
{\bigLRa{*}} \:\subseteq\: {\Da{\sim}} \cup
{\bigLRa{*} \cdot \L \cdot \R \cdot \bigLRa{*}} \cup
{\bigLRa{*} \cdot \eqstep \cdot \R \cdot \bigLRa{*}} \cup
{\bigLRa{*} \cdot \L \cdot \eqstep \cdot \bigLRa{*}}
\]
\end{lem}

\begin{proof}
We abbreviate
${\bigLRa{*} \cdot \L \cdot \R \cdot \bigLRa{*}} \cup
{\bigLRa{*} \cdot \eqstep \cdot \R \cdot \bigLRa{*}} \cup
{\bigLRa{*} \cdot \L \cdot \eqstep \cdot \bigLRa{*}}$
to $\lrhup$.
Suppose $a \bigLRa{n} b$. We show $a \Da{\sim} b$ or $a \lrhup b$ by
induction on $n$. If $n = 0$ then $a = b$ and therefore also
$a \Da{\sim} b$. If $n > 0$ then
$a \bigLR c \bigLRa{n-1} b$ for some $c$. The induction hypothesis
yields $c \Da{\sim} b$ or $c \lrhup b$. In the latter case we are
already done because ${\bigLR \cdot \lrhup} \subseteq {\lrhup}$. In
the former case, note that there exists a $k$ such that
$c \Ra{k} \cdot \sim \cdot \La{*} b$. We continue with a
case analysis on $k$:
\begin{itemize}[label=$\triangleright$]
\item $k = 0$:
From $a \bigLR c$ we obtain $a \R c$, $a \L c$ or $a \sim c$.
If $a \R c$ we immediately obtain $a \Da{\sim} b$.
If $a \L c$ we have $a \L c \sim c' \La{*} b$ for some $c'$. Now
$c = c'$ and hence $a \Da{\sim} b$
or $c \eqstep \cdot \sim c'$ and therefore
$a \lrhup b$. If $a \sim c$ we have $a \Da{\sim} b$ because $\sim$ is
transitive.
\smallskip
\item $k > 0$:
From $a \bigLR c$ we obtain $a \R c$, $a \L c$ or $a \sim c$.
If $a \R c$ we immediately obtain $a \Da{\sim} b$.
If $a \L c$ then there exists a $c'$ such that $a \L c \R c' \bigLRa{*} b$
and therefore $a \lrhup b$. Finally, if $a \sim c$ then
$a \sim c \R c' \bigLRa{*} b$ for some $c'$. If $a = c$ then we obtain
$a \Da{\sim} b$ from the induction hypothesis as there is a conversion
between $a$ and $b$ of length $n-1$. Otherwise, $a \sim \cdot \eqstep c$
and therefore $a \lrhup b$. \qedhere
\end{itemize}
\end{proof}
\noindent 
The proof of the following theorem relies on the fact that the
well-founded order on an index set obtained from peak-and-cliff
decreasingness
can be extended to a well-founded order on multisets of labels.  
Here, the multiset extension of an order $>$ is defined as follows:
$M_1 \gtmul M_2$ if $M_2 = (M_1 \setminus X) \uplus Y$ where
$\varnothing \neq X \subseteq M_1$ and for all $y \in Y$ there
exists an $x \in X$ such that $x > y$. It is well-known that the multiset
extension of a well-founded order is also well-founded \cite{DM79}.

\begin{thm}
\label{thm:pcd}
If $\xA$ is a peak-and-cliff decreasing \textup{ARS} then $\xA$
is Church--Rosser modulo $\sim$.
\end{thm}

\begin{proof}
With every conversion $C$ we associate a multiset $M_C$ consisting
of labels of its rewrite and equivalence relation steps. Since $\xA$
is peak-and-cliff decreasing, there is a well-founded order $>$ on
$I$ which allows us to replace conversions $C$ of the forms
$\Lb{\alpha} \cdot \Rb{\beta}$, $\Lb{\alpha} \cdot \eqstep_\beta$
and $\eqstep_\beta \cdot \Rb{\alpha}$ by conversions $C'$ where
$M_C \gtmul M_{C'}$. Hence, we prove that $\xA$ is Church--Rosser
modulo $\sim$, i.e., ${\bigLRa{*}} \subseteq {\Da{\sim}}$, by
well-founded induction on $\gtmul$. Consider a conversion
$a \bigLRa{*} b$ which we call $C$. By \lemref{acconvpeakcliff} we
have $a \Da{\sim} b$ (which includes the case that $C$ is
empty) or one of the following cases holds:
\begin{align*}
a &\bigLRa{*} \cdot \L \cdot \R \cdot \bigLRa{*} b &
a &\bigLRa{*} \cdot \L \cdot \eqstep \cdot \bigLRa{*} b &
a &\bigLRa{*} \cdot \eqstep \cdot \R \cdot \bigLRa{*} b
\end{align*}
If $a \Da{\sim} b$ we are immediately done. In the remaining cases,
we have a local peak or cliff with concrete labels $\alpha$ and
$\beta$, so $M_C = \Gamma_1 \uplus \SET{\alpha, \beta} \uplus \Gamma_2$.
Since $\xA$ is peak-and-cliff decreasing, there is a conversion $C'$ with
$M_{C'} = \Gamma_1 \uplus \Gamma \uplus \Gamma_3$ where
$\SET{\alpha, \beta} \gtmul \Gamma$. Hence, $M_C \gtmul M_{C'}$ and
we finish the proof by applying the induction hypothesis.
\end{proof}
\noindent 
The above theorem can also be shown by verifying that it is a special
case of the Church--Rosser modulo criterion known as decreasing
diagrams~\cite[Theorem~31]{FvO13}. Note, however, that it is not as
obvious as the fact that peak decreasingness~\cite{HMSW19} is an instance
of decreasing diagrams for confluence~\cite{vO94}.\footnote{Pointed out
by Vincent van Oostrom (personal communication).}

For the main result of this section, a simpler version of
peak-and-cliff decreasingness suffices. The full power
of peak-and-cliff decreasingness will be needed in the
correctness proof of Avenhaus' inference system in
\secref{avenhaus}.

\begin{defi}
Let $\xA = \langle A, {\R} \rangle$ be an ARS equipped with
a $\sim$-compatible well-founded order $>$ on $A$ and
${\sim} = {\eqstep^*}$ an equivalence relation on $A$. We write
\smash{$b \Ro{}{a} c$ $(b \eqstepover{a} c)$} if $b \R c$ $(b \eqstep c)$
and $b \sim a$, i.e., steps are labeled with elements of $A$ as
indices.
We say that $\xA$ is
\emph{source decreasing modulo $\sim$} if the inclusions
\begin{align*}
{\L a \R} &\,\subseteq\, {\bigLRo{\vee a}{*}} &
{\L a \eqstep} &\,\subseteq\, {\bigLRo{\vee a}{*} \cdot \Lo{a}{=}}
\end{align*}
hold for all $a \in A$. Here $\L a \R$ $(\L a \eqstep)$ denotes the
binary relation consisting of all pairs $(b,c)$ such that $a \R b$
and $a \R c$ $(a \eqstep c)$. Furthermore, \smash{$\bigLRo{\vee a}{*}$}
denotes the binary relation consisting of all pairs of elements
which are connected by a conversion where each step is labeled by an
element smaller than $a$.
\end{defi}

\begin{cor}
\label{cor:sd}
Every \textup{ARS} that is source decreasing modulo $\sim$ is
Church--Rosser modulo $\sim$.
\end{cor}

\begin{proof}
In the definition of peak-and-cliff decreasingness we set $I = A$. Note that
this implies $\alpha = \beta$ for all local peaks and cliffs. Hence,
the ARS is peak-and-cliff decreasing and we can conclude by \thmref{pcd}.
\end{proof}

\subsection{Prime Critical Pairs}

We show that joinability of prime critical pairs is enough for
characterizing the Church--Rosser modulo property.
In the following, $\pcppm{\xR,\xB^\pm}$ denotes the restriction of
$\cppm{\xR,\xB^\pm}$ to prime critical pairs but where
irreducibility is always checked with respect to $\xR$, i.e., the
critical peaks $t \Lab{\xR}{p} s \LRab{\xB}{\epsilon} u$ and
$t' \LRab{\xB}{p} s \Rab{\xR}{\epsilon} u'$ are both prime if 
proper subterms of $s|_p$ are irreducible with respect to $\xR$.

\begin{exa}[continued from \exaref{accp}]
\label{exa:acpcp}
Recall the four critical peaks between $\xR$ and $\xB$. The first two
peaks
$\m{f}(x + \m{a}) \Lab{\xB}{1} \m{f}(\underline{\m{a} + x}) \Rb{\xR} x$
and
$\m{f}(\m{a} + x) \Lab{\xB}{1} \m{f}(\underline{x + \m{a}}) \Rb{\xR} x$
are not prime due to the reducible proper subterm $\m{a}$. The other two
are prime. Hence, $\pcppm{\xR,\xB^\pm} =
\SET{\m{f}(\m{b} + x) \approx x,\, \m{f}(x + \m{b}) \approx x}$.
\end{exa}

Correctness of \thmref{accp} can be shown by the combination of
\corref{sd} with the following lemma.

\begin{lemC}[\cite{H80}]
\label{lem:accplemma}
For left-linear \textup{TRS}s $\xR$, the following inclusion holds:
\[
{\LbR \cdot \LRb{\xB}} \,\subseteq\,
{\Dab{\xR}{\sim} \cup \LRb{\cppm{\xR,\,\xB^\pm}}} \eqno\qed
\]
\end{lemC}
\noindent 
In order to integrate the refinement by prime critical pairs some more
observations are required. 
Note that for the original refinement by Kapur et al.~\cite{KMN88},
correctness is shown in the context of general AC rewriting by flattening
terms with AC symbols.  We employ our novel notion of peak-and-cliff
decreasingness instead.  Our proof can be seen as an extension of the
corresponding proof given for the ordinary Church--Rosser property in
\cite{HMSW19}.

\begin{defi}
Given a TRS $\xR$ and terms $s$, $t$ and $u$, we write
$t \trid{s}{} u$ if $s \Rab{\xR}{+} t$, $s \Rab{\xR}{+} u$ and
$t \Db{\xR} u$ or $t \LRb{\pcp{\xR}} u$. We write $t \tridt{s} u$ if
$s \Rab{\xR}{+} t$, $s \sim u$ and $t \Dab{\xR}{\sim}u$ or
$t \LRb{\pcppm{\xR,\xB^\pm}} u$. Furthermore,
${\ttrid{s}} = \SET{(u,t) \mid t \tridt{s} u}$.
\end{defi}

Note that the joinability of ordinary critical peaks is not affected
by incorporating $\xB$ into conversions. Hence, the following result
is taken from \cite[Lemma 2.15]{HMSW19} and therefore stated without
proof. Here, $t \trid{s}{2} u$ means that there is a term $v$
with $t \trid{s}{} v$ and $v \trid{s}{} u$.

\begin{lemC}[\cite{HMSW19}]
\label{lem:cpeak}
If $t \Lo{\xR}{p} s \Ro{\xR}{\epsilon} u$ is
a critical peak of a \textup{TRS} $\xR$ then $t \trid{s}{2} u$. \qed
\end{lemC}

\begin{lem}
\label{lem:ccliff}
Let $\xR$ be a left-linear \textup{TRS}.
\begin{enumerate}
\item
If $t \Lo{\xR}{p} s \LRo{\xB}{\epsilon} u$ is a critical peak then
$t \trid{s}{} \cdot \tridt{s} u$.
\item
If $t \LRo{\xB}{p} s \Ro{\xR}{\epsilon} u$ is a critical peak then
$t \ttrid{s} \cdot \trid{s}{} u$.
\end{enumerate}
\end{lem}

\begin{proof}
We only prove (1) as the other statement is symmetrical. If all
proper subterms of $s|_p$ are in normal form with respect to $\RbR$,
$t \approx u \in \pcp{\xR,\xB^\pm}$ which establishes
$t \tridt{s} u$. Since also $t \trid{s}{} t$, we obtain the desired
result. Otherwise, there are a position $q > p$ and a term $v$ such
that $s \Ro{\xR}{q} v$ and all proper subterms of $s|_q$ are in
normal form with respect to $\RbR$. Together with
\lemref{accplemma} we obtain $v \Dab{\xR}{\sim} u$ or
$v \LRb{\pcppm{\xR,\xB^\pm}} u$. In both cases $v \tridt{s} u$
holds. A similar case analysis applies to the local peak
$t \Lo{\xR}{p} s \Ro{\xR}{q} v$: By the Critical Pair Lemma, either
$t \Db{\xR} v$ or $t \LRb{\cp{\xR}} v$. In the latter case
\[
v|_p \Lo{\xR}{q \setminus p} s|_p \Ro{\xR}{\epsilon} t|_p
\]
is an instance of a prime critical peak as $q > p$ and all proper
subterms of $s|_q$ are in normal form with respect to
$\RbR$. Closure of rewriting under contexts and substitutions
yields $t \LRb{\pcp{\xR}} v$. Therefore, we have $t \trid{s}{} v$ in
both cases, concluding the proof.
\end{proof}
\noindent 
The following lemma generalizes the previous results of this section
to arbitrary local peaks and cliffs.

\begin{lem}
\label{lem:pcplemma}
Let $\xR$ be a left-linear \textup{TRS}.
\begin{enumerate}
\item
If $t \LbR s \RbR u$ then $t \trid{s}{2} u$.
\item
If $t \LbR s \LRb{\xB} u$ then $t \trid{s}{} \cdot \tridt{s} u$.
\end{enumerate}
\end{lem}

\begin{proof}
We only prove (2) as the proof of (1) (which depends on the Critical
Pair Lemma) can be found in \cite[Lemma 2.16]{HMSW19}. Let
$t \LbR s \LRb{\xB} u$. From \lemref{accplemma} we obtain
$t \Dab{\xR}{\sim} u$ or $t \LRb{\cppm{\xR,\,\xB^\pm}} u$. In the
former case we are done as $t \trid{s}{} u \trid{s}{} u$. For the
latter case we further distinguish between the two subcases
$t \Rb{\cp{\xR,\,\xB^\pm}} u$ and $u \Rb{\cp{\xB^\pm,\,\xR}} t$.
Note that this list of subcases is exhaustive due to the
direction of the local cliff. If $t \Rb{\cp{\xR,\,\xB^\pm}} u$,
$t \trid{s}{} \cdot \tridt{s} u$ follows from \lemref{ccliff}(1) and
closure of rewriting under contexts and substitutions. If
$u \Rb{\cp{\xB^\pm,\,\xR}} t$, $u \ttrid{s} \cdot \trid{s}{} t$ and
therefore $t \trid{s}{} \cdot \tridt{s} u$ follows from
\lemref{ccliff}(2) as well as closure of rewriting under contexts
and substitutions.
\end{proof}

Now, we are able to prove the main result of this section, a novel
necessary and sufficient condition for the Church--Rosser property
modulo an equational theory $\xB$ which strengthens the original
result from \cite{H80} to prime critical pairs.

\begin{thm}
\label{thm:acpcp}
A left-linear \textup{TRS} $\xR$ which is terminating modulo $\xB$
is Church--Rosser modulo $\xB$ if and only if
$\pcp{\xR} \cup \pcppm{\xR,\xB^\pm} \subseteq {\Dab{\xR}{\sim}}$.
\end{thm}

\begin{proof}
The only-if direction is trivial. For a proof of the if
direction, we show that $\xR$ is source decreasing modulo $\xB$; the
Church--Rosser property modulo $\xB$ is then an immediate
consequence of \corref{sd}. From the termination of $\xR$ modulo
$\xB$ we obtain the well-founded order ${>} = {\RmB[+]}$.

Consider an arbitrary local peak $t \LbR s \RbR u$.
\lemref{pcplemma}(1) yields a term $v$ such that
$t \trid{s}{} v \trid{s}{} u$. Together with
$\pcp{\xR} \subseteq {\Dab{\xR}{\sim}}$ we obtain
$t \Dab{\xR}{\sim} v \Dab{\xR}{\sim} u$. By definition,
$s > t, v, u$ so the corresponding condition required by source
decreasingness modulo $\xB$ is fulfilled.

Now consider an arbitrary local cliff $t \LbR s \LRb{\xB} u$.
\lemref{pcplemma}(2) yields a term $v$ such that
$t \trid{s}{} v \tridt{s} u$. Together with
$\pcp{\xR} \cup \pcppm{\xR,\xB^\pm} \subseteq {\Dab{\xR}{\sim}}$ we
obtain $t \Dab{\xR}{\sim} v \Dab{\xR}{\sim} u$. By definition,
$s > t, v$ and $s \sim u$. The conversion between $v$ and $u$ is of
the form $v \Rab{\xR}{*} \cdot \sim \cdot \Lab{\xR}{k} u$ for some
$k$. If $k = 0$ then all steps between $v$ and $u$ can be labeled with
terms which are smaller than $s$. If $k > 0$ then there exists a
$w < s$ such that
$v \Rab{\xR}{*} \cdot \sim \cdot \Lab{\xR}{k-1} w \LbR u$. In
this case all steps of the conversion are labeled with terms which
are smaller than $s$ except for the rightmost step which we may
label with $s$. Hence, the corresponding condition required by
source decreasingness modulo $\xB$ is fulfilled in all cases.
\end{proof}

\begin{exa}[continued from \exaref{acpcp}]
One can verify the termination of $\xR/\xB$ and the inclusion
$\pcp{\xR} \cup \pcppm{\xR,\xB^\pm} \subseteq {\Dab{\xR}{\sim}}$.
By \thmref{acpcp} the Church--Rosser modulo property holds.
\end{exa}

Finally, we show that the previous result does not hold if we just
demand termination of $\xR$. The counterexample shows this for the
concrete case of AC and is based on Example 4.1.8 from \cite{A95}
which uses an ARS. Note that the usage of prime critical pairs instead
of critical pairs has no effect.

\begin{exa}
\label{exa:acterm}
Consider the TRS $\xR$ consisting of the rules
\begin{align*}
(\m{b}+\m{a})+\m{a} &\cR \m{a}+(\m{a}+\m{b}) &
(\m{a}+\m{b})+\m{a} &\cR \m{a}+(\m{a}+\m{b}) &
(\m{a}+\m{a})+\m{b} &\cR \m{a}+(\m{a}+\m{b}) \\
\m{a}+(\m{a}+\m{b}) &\cR \m{b}+(\m{a}+\m{a}) &
\m{b}+(\m{a}+\m{a}) &\cR \m{c} & & \\ 
\m{a}+(\m{a}+\m{b}) &\cR \m{a}+(\m{b}+\m{a}) &
\m{a}+(\m{b}+\m{a}) &\cR \m{d} & &
\end{align*}
where $+$ is an AC function symbol. Clearly, the (prime) critical
pairs of $\xR$ are joinable modulo AC because
$\m{b} + (\m{a} + \m{a}) \acsim \m{a} + (\m{b} + \m{a})$. For
$\pcppm{\xR,\AC^\pm}$ we only have to consider the rules which
rewrite to $\m{c}$ and $\m{d}$ respectively since all other rules
only involve AC equivalent terms. Modulo symmetry, these (prime)
critical pairs are:
\begin{align*}
\m{c} &\cE \m{b}+(\m{a}+\m{a}) &
\m{c} &\cE (\m{a}+\m{a})+\m{b} &
\m{c} &\cE (\m{b}+\m{a})+\m{a} \\
\m{c}+x &\cE \m{b}+((\m{a}+\m{a})+x) &
x+\m{c} &\cE (x+\m{b})+(\m{a}+\m{a}) \\
\m{d} &\cE \m{a}+(\m{a}+\m{b}) &
\m{d} &\cE (\m{b}+\m{a})+\m{a} &
\m{d} &\cE (\m{a}+\m{b})+\m{a} \\
\m{d}+x &\cE \m{a}+((\m{b}+\m{a})+x) &
x+\m{d} &\cE (x+\m{a})+(\m{b}+\m{a})
\end{align*}
Removing the joinable (prime) critical pairs leaves us with
\begin{align*}
\m{c}+x &\cE \m{b}+((\m{a}+\m{a})+x) &
x+\m{c} &\cE (x+\m{b})+(\m{a}+\m{a}) \\
\m{d}+x &\cE \m{a}+((\m{b}+\m{a})+x) &
x+\m{d} &\cE (x+\m{a})+(\m{b}+\m{a})
\end{align*}
which are not joinable at the moment. However, we can extend $\xR$
by the rewrite rules
\begin{align*}
\m{b} +((\m{a}+\m{a})+x) &\cR (\m{b}+(\m{a}+\m{a}))+x &
(x+\m{b})+(\m{a}+\m{a}) &\cR x+(\m{b}+(\m{a}+\m{a})) \\
\m{a} +((\m{b}+\m{a})+x) &\cR (\m{a}+(\m{b}+\m{a}))+x &
(x+\m{a})+(\m{b}+\m{a}) &\cR x+(\m{a}+(\m{b}+\m{a}))
\end{align*}
in order to make them joinable. Note that the additional (prime)
critical pairs in $\pcp{\xR} \cup \pcppm{\xR,\AC^\pm}$ caused by
adding the new rules are trivially joinable modulo AC as all of
these new critical pairs are AC equivalent. To sum up,
$\pcp{\xR} \cup \pcppm{\xR,\AC^\pm} \subseteq {\Dab{\xR}{\sim}}$.
Termination of $\xR$ can be checked by
e.g.~the termination tool
\href{\termexamplelink}{\TTTT}~\cite{KSZM09}, but the loop
\[
\m{a} + (\m{a} + \m{b}) \,\RbR\, \m{a} + (\m{b} + \m{a}) \,\acsim\,
\m{a} + (\m{a} + \m{b})
\]
shows that $\xR$ is not AC terminating. We have
$\m{c} \bigLRa{*} \m{d}$ but not $\m{c} \Dab{\xR}{\sim} \m{d}$ as
the terms are normal forms and not AC equivalent. Hence, $\xR$ is
not Church--Rosser modulo AC.
\end{exa}

\section{Avenhaus' Inference System}
\label{sec:avenhaus}

The idea of completion modulo an equational theory $\xB$ for
left-linear systems where the normal rewrite relation can be used to
decide validity problems has been put forward by Huet \cite{H80}. To
the best of our knowledge, inference systems for this approach are
only presented in the books by Avenhaus \cite{A95} and Bachmair
\cite{B91}. This section presents a new correctness proof of a version
of Avenhaus' inference system for finite runs in the spirit of
\cite{HMSW19} which does not rely on proof orderings. Correctness of
Bachmair's system is established by a simulation result in
\secref{bachmair}.

\subsection{Inference System}

\begin{defi}
\label{def:ia}
The inference system \ia is parameterized by a fixed
$\xB$-compatible reduction order $>$. It transforms pairs
consisting of an ES $\xE$ and a TRS $\xR$ over the common signature
$\xF$ according to the following inference rules where
$s \approxpm t$ denotes either
$s \approx t$ or $t \approx s$:
\[
\begin{array}{r@{\quad}c@{\quad}l@{\qquad}r@{\quad}c@{\quad}l@{}}
\ded & \ds \frac{\xE,\xR}{\xE \cup \SET{s \approx t},\xR}
& \text{if $s \LbR \cdot \RbR t$} & \ori
& \ds \frac{\xE \uplus \SET{s \approxpm t},\xR}{\xE,\xR \cup \SET{s \R t}}
& \text{if $s > t$} \\[0.6cm]
\ded & \ds \frac{\xE,\xR}{\xE,\xR \cup \SET{t \R s}}
& \text{if $s \LbR \cdot \LRb{\xB} t$} & \del &
\ds \frac{\xE \uplus \SET{s \approx t},\xR}{\xE,\xR}
& \text{if $s \bsim t$} \\[0.6cm]
\sip & \ds
\frac{\xE \uplus \SET{s \approxpm t},\xR}{\xE \cup \SET{u \approx t},\xR}
& \text{if $s \RmB u$} & \col &
\ds \frac{\xE,\xR \uplus \SET{s \to t}}{\xE \cup \SET{u \approx t},\xR}
& \text{if $s \RbR u$} \\[0.6cm]
& & & \com &
\ds \frac{\xE,\xR \uplus \SET{s \R t}}{\xE,\xR \cup \SET{s \R u}}
& \text{if $t \RmB u$}
\end{array}
\]
\end{defi}

A step in an inference system \infsys from an ES $\xE$ and a TRS $\xR$
to an ES $\xE'$ and a TRS $\xR'$ is denoted by
\smash{$(\xE,\xR) \seq{\infsys} (\xE',\xR')$}. The parentheses of the
pairs are only used when the expression is surrounded by text in order to
increase readability.

\begin{defi}
\label{def:fairness}
Let $\xE$ be an ES. A finite sequence
\[
\xE_0,\xR_0 \,\iA\, \xE_1,\xR_1 \,\iA\, \cdots \,\iA\, \xE_n,\xR_n
\]
with $\xE_0 = \xE$ and $\xR_0 = \varnothing$ is a \emph{run} for
$\xE$. If $\xE_n \neq \varnothing$, the run \emph{fails}. The run is
\emph{fair} if $\xR_n$ is left-linear and the following inclusions hold:
\begin{align*}
\pcp{\xR_n} &\,\subseteq\,
{{\Dab{\xR_n}{\sim}} \cup {\bigcup_{i=0}^n {\LRb{\xE_i \cup \xR_i}}}} &
{\pcppm{\xR_n, \xB^\pm}} &\,\subseteq\,
{{\Dab{\xR_n}{\sim}} \cup {\bigcup_{i=0}^n {\LRb{\xR_i}}}}
\end{align*}
\end{defi}

Intuitively, fair and non-failing runs yield a $\xB$-complete
presentation $\xR_n$ of the initial set of equations $\xE$, i.e.,
${\LRab{\xE \cUp \xB}{*}} = {\LRab{\xR_n \cUp \xB}{*}} \subseteq
{\Dab{\xR_n}{\sim}}$. In particular, the inference rules are designed
to preserve the equational theory augmented by $\xB$.

\begin{exa}
\label{exa:a95_ex_4_2_15b}
In this example we illustrate a successful run for the ES $\xE$
consisting of the equations
\begin{align*}
\m{f}(x + y) &\nE{1} \m{f}(x) + \m{f}(y) &
\m{f}(\m{0}) &\nE{2} \m{0} &
x + \m{0} &\nE{3} x
\end{align*}
where $+$ is an AC function symbol. This example is taken from
\cite[Example 4.2.15(b)]{A95}. As suggested by \defref{fairness} we
only consider prime critical pairs. As AC-compatible reduction
order we use the polynomial interpretation \cite{BL87}
\begin{align*}
+_{\mathbb{N}}(x,y) &\,=\, x+y+1 &
\m{f}_{\mathbb{N}}(x) &\,=\, x^2 + x &
\m{0}_{\mathbb{N}} &\,=\, 1
\end{align*}
and start by orienting equations
2 and 3
into the rules
\begin{align*}
\m{f}(\m{0}) &\nR{2'} \m{0} &
x + \m{0} &\nR{3'} x
\intertext{which only leads to prime critical pairs between rule 3$'$
and $\AC^\pm$. We add these prime critical pairs as rules by
applying \ded:}
\m{0} + x &\nR{4} x &
x + (y + \m{0}) &\nR{6} x + y \\
x + (\m{0} + y) &\snR{5} x + y &
(x + y) + \m{0} &\snR{7} x + y
\end{align*}
Keeping rule 4
enables us to collapse the remaining three rules to the equations
\begin{align*}
x + y &\nE{5'} x + y &
x + y &\nE{6'} x + y &
x + y &\nE{7'} x + y
\intertext{
which can all be deleted. Next, we deduce the prime critical pair
stemming from rules 3$'$
and 4
which is just the
trivial equation $\m{0} \approx \m{0}$ and therefore can be
deleted. We continue by deducing the prime critical pairs between
rule 4
and $\AC^\pm$ which adds the new rules}
\m{0} + (x + y) &\nR{8} x + y &
(\m{0} + x) + y &\nR{9} x + y &
(x + \m{0}) + y &\nR{10} x + y
\intertext{which can all be collapsed to trivial equations}
x + y &\nE{8'} x + y &
x + y &\nE{9'} x + y &
x + y &\nE{10'} x + y
\intertext{and therefore deleted. Now we orient the only remaining
original equation 1 to}
&& \m{f}(x + y) &\nR{1'} \m{f}(x) + \m{f}(y)
\end{align*}
which gives rise to two prime critical pairs between rules 1$'$ and 3$'$
as well as rule 4
\begin{align*}
\m{f}(x) + \m{f}(\m{0}) &\nE{11} \m{f}(x) &
\m{f}(\m{0}) + \m{f}(x) &\nE{12} \m{f}(x)
\intertext{which can be simplified to the trivial equations}
\m{f}(x) &\nE{11'} \m{f}(x) &
\m{f}(x) &\nE{12'} \m{f}(x)
\end{align*}
and therefore deleted. Finally, we deduce rules corresponding to the
prime critical pairs between rule 1$'$
and $\AC^\pm$:
\begin{align*}
\m{f}(y + x) &\nR{13} \m{f}(x) + \m{f}(y) &
\m{f}(x + (y + z)) &\nR{14} \m{f}(x + y) + \m{f}(z) \\
&& \m{f}((x + y) + z) &\snR{15} \m{f}(x) + \m{f}(y + z)
\intertext{Applications of \col and \sip transform these rules to AC
equivalent equations}
\m{f}(y) + \m{f}(x) &\nE{13'} \m{f}(x) + \m{f}(y) &
\m{f}(x) + (\m{f}(y) + \m{f}(z)) &\nE{14'}
(\m{f}(x) + \m{f}(y)) + \m{f}(z) \\[.5ex]
&& (\m{f}(x) + \m{f}(y)) + \m{f}(z) &\snE{15'}
\m{f}(x) + (\m{f}(y) + \m{f}(z))
\end{align*}
which can be deleted. Thus, the TRS consisting of the rules
1$'$, 2$'$, 3$'$ and 4
is the result of a fair and non-failing run which is an AC complete
presentation of the original equations as we will show in the
correctness proof.
\end{exa}

The next example shows that deducing local cliffs as rules instead
of equations as well as the restriction to $\RbR$ in the \col rule are
crucial properties of the inference system.

\begin{exa}
\label{exa:cliffscol}
Consider the ES $\xE$ consisting of the single equation
\[
x + \m{0} \cE x
\]
where $+$ is an AC function symbol. We clearly have
$\m{0} + x \LRab{\xE \cUp \AC}{*} x$, so an AC complete system $\xC$
representing $\xE$ has to satisfy $\m{0} + x \Dab{\xC}{\sim}
x$. There is just one way to orient the only equation in $\xE$ which
results in the rule $x + \m{0} \R x$. Since we want our run to be
fair, we add the rules stemming from the prime critical pairs
between $x + \m{0} \R x$ and $\AC^\pm$:
\begin{align*}
\m{0} + x &\cR x &
x + (\m{0} + y) &\cR x + y &
x + (y + \m{0}) &\cR x + y &
(x + y) + \m{0} &\cR x + y
\end{align*}
If collapsing with $\RmAC$ is allowed, all these rules become
trivial equations and can therefore be deleted. Thus, the modified
inference system allows for a fair run which is not complete as
$\m{0} + x \Dab{\xR}{\sim} x$ does not hold for
$\xR = \SET{x + \m{0} \R x}$. Furthermore, if we add pairs of terms
stemming from local cliffs as equations, we get the same result by
applications of \sip.
\end{exa}

The inference system presented in \defref{ia} is almost the same
as the one presented by Avenhaus in \cite{A95}. However, since we only
consider finite runs, the encompassment condition for the \col rule
has been removed in the spirit of \cite{ST13}. (The original side
condition is $s \RbR u$ with $\ell \R r \in \xR$ and $s \enc l$.)
The following example shows that this can lead to smaller
$\xB$-complete systems.

\begin{exa}
\label{exa:encompassment}
Consider the ES $\xE$ consisting of the single equation
\[
\m{f}(x + y) \cE \m{f}(x) + \m{f}(y)
\]
where $+$ is an AC function symbol. The inference system presented
in \cite{A95} produces the AC complete system
\begin{align*}
\m{f}(x + y) &\cR \m{f}(x) + \m{f}(y) &
\m{f}(y + x) &\cR \m{f}(x) + \m{f}(y)
\end{align*}
in which either of the rules could be collapsed if it was allowed to
collapse with the other rule. In \cite{A95} this is prevented by an
encompassment condition which essentially forbids to collapse at the
root position with a rewrite rule whose left-hand side is a variant
of the left-hand side of the rule which should be
collapsed. However, this is possible with the system presented in
this article, so for an AC complete representation just one of the
two rules suffices.
\end{exa}

\subsection{Correctness Proof}

We show that every fair and non-failing finite run results in a
$\xB$-complete presentation. To this end, we first verify that
inference steps in \ia preserve convertibility. We abbreviate
$\xE \cup \xR \cup \xB$ to $\xERB$ and $\xE' \cup \xR' \cup \xB$ to
$\xERBp$.

\begin{lem}
\label{lem:singleacllstep}
If $(\xE,\xR) \iA (\xE',\xR')$ then the following inclusions hold:
\begin{enumerate}
\item
${\Ro{\xERB}{}} \,\subseteq\, {\Ro{{\xR'}/{\xB}}{=} \cdot
\mathrel{(\Ro{\xERp}{=} \cup \LRo{\xB\h}{*})} \cdot \Lo{{\xR'}/{\xB}}{=}}$
\item
${\Ro{\xERBp}{}} \,\subseteq\, {\LRo{\xERB\h}{*}}$
\end{enumerate}
\end{lem}

\begin{proof}
By inspecting the inference rules of \ia we obtain the following
inclusions: \bigskip \\
\ded
\begin{align*}
\xE \cup \xR &\,\subseteq\, \xE' \cup \xR' &
\MR[20]{$\xE' \cup \xR' \,\subseteq\,
\xE \cup \xR \cup {\Lo{\xR}{} \cdot \Ro{\xR}{}}$}
&\cup {\LRo{\xB\h}{} \cdot \Ro{\xR}{}} \cup
{\Lo{\xR}{} \cdot \LRo{\xB\h}{}}
\intertext{\ori}
\xE \cup \xR &\,\subseteq\, \xE' \cup \xR' \cup (\xR')^{-1} &
\xE' \cup \xR' &\,\subseteq\, \xE \cup (\xE)^{-1} \cup \xR
\intertext{\del}
\xE \cup \xR &\,\subseteq\, \xE' \cup \xR' \cup {\bsim} &
\xE' \cup \xR' &\,\subseteq\, \xE \cup \xR
\intertext{\com}
\xE \cup \xR &\,\subseteq\, \xE' \cup \xR' \cup
{\Ro{\xR'}{} \cdot \Lo{{\xR'}/{\xB}}{}} &
\xE' \cup \xR' &\,\subseteq\, \xE \cup \xR \cup
{\Ro{\xR}{} \cdot \Ro{{\xR}/{\xB}}{}}
\intertext{\col}
\xE \cup \xR &\,\subseteq\,
\xE' \cup \xR' \cup {\Ro{\xR'}{} \cdot \Ro{\xE'}{}} &
\xE' \cup \xR' &\,\subseteq\, \xE \cup \xR \cup
{\Lo{\xR}{} \cdot \Ro{\xR}{}}
\intertext{\sip}
\xE \cup \xR &\,\subseteq\, \ML{$\xE' \cup \xR' \cup
{\Ro{{\xR'}/\xB}{} \cdot \Ro{\xE'}{}} \cup
{\Ro{\xE'}{} \cdot \Lo{{\xR'}/{\xB}}{}}$} \\
\xE' \cup \xR' &\,\subseteq\, \ML{$\xE \cup \xR \cup
{\Lo{{\xR}/\xB}{} \cdot \Ro{\xE}{}} \cup
{\Ro{\xE}{} \cdot \Ro{{\xR}/\xB}{}}$}
\end{align*}
Then, inclusion (2) follows directly from the closure of rewrite
relations under contexts and substitutions. Statement (1) holds
since it is a generalization that all cases have in common.
\end{proof}

\begin{cor}
\label{cor:acllrepr}
If $(\xE,\xR) \iA^* (\xE',\xR')$ then
${\LRo{\xERB\h}{*}} = {\LRo{\xERBp}{*}}$.
\end{cor}

\begin{lem}
\label{lem:aclltermination}
If $(\xE,\xR) \iA^* (\xE',\xR')$ and $\xR \subseteq {>}$ then
$\xR' \subseteq {>}$.
\end{lem}

\begin{proof}
According to the assumption we have $(\xE,\xR) \iA^n (\xE',\xR')$
for some natural number $n$. We proceed by induction on $n$. If
$n = 0$, the statement is trivial since $\xR = \xR'$. Let $n > 0$
and consider $(\xE,\xR) \iA^{n-1} (\xE'',\xR'') \iA (\xE',\xR')$.
The induction hypothesis yields $\xR'' \subseteq {>}$. We continue
with a case analysis on the rule applied in the step
$(\xE'',\xR'') \iA (\xE',\xR')$. For the rules \del and \sip there
is nothing to show as the set of rewrite rules is not changed. If
\ded is applied to a local peak there is nothing to show. Otherwise,
we have $\xR' = \xR'' \cup \SET{s \R t}$ where
$s \LRb{\xB} \cdot \RbR t$. From the fact that $>$ is
$\xB$-compatible we obtain $s > t$ and therefore
$\xR' \subseteq {>}$. For \ori we have
$\xR' = \xR'' \cup \SET{s \R t}$ with $s > t$, so
$\xR' \subseteq {>}$. In the case of \com we have
$\xR' = (\xR'' \setminus \SET{s \R t}) \cup \SET{s \R u}$ with
$t \RmB u$. Since $>$ is a $\xB$-compatible reduction order,
$t \RmB u$ implies $t > u$. From the induction hypothesis we
obtain $s > t$. Now $\xR' \subseteq {>}$ follows by the transitivity
of $>$. Finally, for \col we have
$\xR' \subsetneq \xR'' \subseteq {>}$ which establishes
$\xR' \subseteq {>}$.
\end{proof}

In order to use peak-and-cliff decreasingness in the correctness proof,
we have to define an appropriate notion of labeled rewriting. Intuitively,
we want to annotate a step $s \RbR t$ or $s \bsim t$ with a collection of
terms in such a way that the collection contains terms which are
$\xB$-equivalent or rewrite in a positive number of steps using $\RmB$
to $s$ and $t$, respectively.

\begin{defi}
Let $\LR$ be a rewrite relation or equivalence relation, $M$ a
finite multiset of terms and $>$ a $\xB$-compatible reduction
order. We write $\smash{s \LRo{}{M} t}$ if $s \LR t$ and there exist
terms $s', t' \in M$ such that $s' \gtrsim s$ and $t' \gtrsim t$ for
${\gtrsim} = {{>} \cup {\bsim}}$.
\end{defi}

We follow the convention that if a conversion is labeled with $M$, all
single steps can be labeled with $M$.

\begin{exa}
Consider the TRS $\xR$ consisting of the rules
\begin{align*}
\m{0} + y &\R y & \m{s}(x) + y &\R \m{s}(x + y)
\end{align*}
and the equational theory $\xB = \SET{x + y \approx y + x}$. 
Let ${>} = {\RmB}$ which is a $\xB$-compatible reduction order by
definition as $\xR$ is $\xB$-terminating. We have
\[
\m{s}(\m{0}) + x \Ro{\xR}{\SET{x + \m{s}(\m{0})}} \m{s}(\m{0} + x)
\quad\text{and}\quad
\m{s}(\m{0} + x) \Ro{\xR}{\SET{\m{s}(\m{0} + x)}} \m{s}(x)
\]
as well as
\[
\m{s}(\m{0}) + x \Rao{\xR}{\SET{x + \m{s}(\m{0})}}{*} \m{s}(x)
\quad\text{but not}\quad
\m{s}(\m{0}) + x \Rao{\xR}{\SET{\m{s}(\m{0} + x)}}{*} \m{s}(x).
\]
\end{exa}

\begin{lem}
\label{lem:labelpreserve}
Let $(\xE,\xR) \iA (\xE',\xR')$ and $\xR' \subseteq {>}$.
\begin{enumerate}
\item
For any finite multiset $M$ we have
$\smash{{\LRao{\xERB\h}{M}{*}} \subseteq {\LRao{\xERBp}{M}{*}}}$.
\item
If $s \Ro{\xR}{M} t$ then $s \Rao{\xR'}{M}{=} \cdot \LRao{\xERBp}{N}{*} t$
with $\SET{s} \gtmul N$.
\end{enumerate}
\end{lem}

\begin{proof}
For (1) it suffices to show that
${\Ro{\xERB}{M}} \subseteq {\LRao{\xERBp}{M}{*}}$. Let
$s \Ro{\xERB}{M} t$. By definition, there exist terms $s', t' \in M$
with $s' \gtrsim s$ and $t' \gtrsim t$. According to
\lemref{singleacllstep} there exist terms $u$ and $v$ such that
\[
s \Ro{{\xR'}/{\xB}}{=} u \mathrel{(\Ro{\xE' \cUp \xR'}{=} \cup
\LRo{\xB\h}{*})} v \Lo{{\xR'}/{\xB}}{=} t
\]
Since $\xR' \subseteq {>}$ we have $s \gtrsim u$ and $t \gtrsim v$
and therefore $s' \gtrsim u$ and $t' \gtrsim v$. Hence, every
non-empty step can be labeled with $M$ and we obtain
$s \LRao{\xERBp}{M}{*} t$ as desired.

\smallskip

For a proof of (2), let $s \Rab{\xR}{M} t$.
By definition, there exists an $s' \in M$
such that $s' \gtrsim s$. We proceed by case analysis on the rule
applied in the inference step. For \ded, \ori, \del and \sip there
is nothing to show since $\xR \subseteq \xR'$.

\smallskip

Suppose the step is an application of \com. If the rule used in the
step $s \Rab{\xR}{M} t$ is not altered, we are done. Otherwise, the
step was performed with a rule ${\ell \to r} \in \xR$ which is
changed to ${\ell \to r'} \in \xR'$ with $r \Rb{{\xR'}/{\xB}}{} r'$.
There exist a substitution $\sigma$ and a position $p$ such
that $s = s[\ell\sigma]_p$ and $t = s[r\sigma]_p$. The new step is
$s \Rab{\xR'}{M} t'$ where $t' = s[r'\sigma]_p$. Since $>$ is a
$\xB$-compatible reduction order, we have $s' > t'$ and therefore the
label $M$ is still valid. From $t'$ we can reach $t$ with
\[
t' \LRao{\xB}{\SET{t}}{*} \cdot \Lo{\xR'}{\SET{t}} \cdot
\LRao{\xB}{\SET{t}}{*} t
\]
From $s > t$ we obtain $\SET{s} \gtmul \SET{t}$ which means that the new
conversion between $s$ and $t$ is of the desired form.

\smallskip

Finally, suppose the step is an application of \col. If the rule
used in the step $s \Rab{\xR}{M} t$ is not altered, we are done
immediately. Otherwise, the step was performed with a rule
${\ell \to r} \in \xR$ which is changed to an equation
${\ell' \approx r} \in \xE'$ with $\ell \Rb{\xR'} \ell'$. There
exist a substitution $\sigma$ and a position $p$ such that
$s = s[\ell\sigma]_p$ and $t = s[r\sigma]_p$. The new step is
$s \Rab{\xR'}{M} t'$ where $t' = s[\ell'\sigma]_p$. Since $>$ is a
$\xB$-compatible reduction order we have $s' > t'$ and therefore the
label $M$ is still valid. From $t'$ we can reach $t$ with
$t' \Rab{\xE'}{N} t$ where $N = \SET{t',t}$. From $s > t'$ and $s > t$
we obtain $\SET{s} \gtmul N$ which means that the new conversion
between $s$ and $t$ is of the desired form.
\end{proof}

Finally, we are able to prove the correctness result for \ia,
i.e., all finite fair and non-failing runs produce a $\xB$-complete
TRS which represents the original set of equations. In contrast to
\cite{A95} and \cite{B91}, the proof shows that it suffices to
consider prime critical pairs. This is achieved by
showing peak-and-cliff decreasingness and  using \thmref{pcd} instead
of directly using the main theorem of \secref{cc} (\thmref{acpcp}).
The usage of peak-and-cliff decreasingness makes the proof more
modular and easier to formalize than the proof in \cite{A95} because
it is split up into the preceding auxiliary lemmata of this
section which are mostly independent from each other and can use
different suitable proof methods. This is very different from
the approach in \cite{A95} where all the necessary information and
induction hypotheses are incorporated into one large proof ordering.
The results presented in \cite{HMSW19} for standard rewriting suggest
that our approach still has merits when compared to the original proofs
in \cite{A95} and \cite{B91} if infinite runs are considered.

\begin{thm}
\label{thm:correctness}
Let $\xE$ be an \textup{ES}. For every fair and non-failing run
\[
\xE_0,\xR_0 \,\iA\, \xE_1,\xR_1 \,\iA\, \cdots \,\iA\, \xE_n,\xR_n
\]
for $\xE$, the \textup{TRS} $\xR_n$ is a $\xB$-complete
representation of $\xE$.
\end{thm}

\begin{proof}
Let $>$ be the $\xB$-compatible reduction order used in the run.
From fairness we obtain $\xE_n = \varnothing$ as well as the fact
that $\xR_n$ is left-linear. \corref{acllrepr} establishes
${\LRab{\xE \cUp \xB}{*}} = {\LRab{\xR_n \cUp \xB}{*}}$ and
termination modulo $\xB$ of $\xR_n$ follows from
\lemref{aclltermination}. It remains to prove that $\xR_n$ is
Church--Rosser modulo $\xB$ which we do by showing peak-and-cliff
decreasingness. So consider a labeled local peak
$t \Lab{\xR_n}{M_1} s \Rab{\xR_n}{M_2} u$. \lemref{pcplemma}(1)
yields $t \trid{s}{2} u$. Let $v \trid{s}{} w$ appear in this
sequence (so $v = t$ or $w = u$). By definition, $v \Db{\xR_n} w$ or
$v \LRb{\pcp{\xR_n}} w$. Together with fairness, the fact that
$\bsim$ is reflexive as well as closure of rewriting under contexts
and substitutions we obtain $v \Dab{\xR_n}{\sim} w$ or
$(v,w) \in {\bigcup_{i=0}^n {\LRb{\xE_i \cUp \xR_i}}}$. In both
cases, it is possible to label all steps between $v$ and $w$ with
$\SET{v,w}$. Since $s > v$ and $s > w$ we have $M_1 \gtmul \SET{v,w}$
and $M_2 \gtmul \SET{v,w}$. Repeated applications of
\lemref{labelpreserve}(1) therefore yield a conversion in
$\xR_n \cup \xB$ between $v$ and $w$ where every step is labeled
with a multiset that is smaller than both $M_1$ and $M_2$. Hence,
the corresponding condition required by peak-and-cliff
decreasingness is fulfilled.

Next consider a labeled local cliff
$t \Lab{\xR_n}{M_1} s \LRab{\xB}{M_2} u$. From \lemref{pcplemma}(2)
we obtain a term $v$ such that $t \trid{s}{} v \tridt{s} u$. As in
the case for local peaks we obtain a conversion between $t$ and $v$
where each step can be labeled with $\SET{t,v} \ltmul M_1$. Together
with fairness, $v \tridt{s} u$ yields $v \Dab{\xR_n}{\sim} u$ or
$(v,u) \in {\bigcup_{i = 0}^n {\LRb{\xR_i}}}$. In the former case
there exists a $k$ such that
$v \Rab{\xR_n}{*} \cdot \bsim^{} \cdot \Lab{\xR_n}{k} u$. If $k = 0$
we can label all steps with $\SET{v}$. If $k > 0$ the conversion is of
the form
$v \Rab{\xR_n}{*} \cdot \bsim^{} \cdot \Lab{\xR_n}{k-1} w \Lb{\xR_n}
u$. We can label the rightmost step with $M_2$ and the remaining
steps with $\SET{v,w}$. Note that $s > v$. Since $>$ is a
$\xB$-compatible reduction order we also have $s > w$. Thus,
$M_1 \gtmul \SET{v,w}$ which establishes the corresponding condition
required by peak-and-cliff decreasingness for all $k$. In the
remaining case we have
$(v,u) \in {\bigcup_{i = 0}^n {\LRb{\xR_i}}}$, so there is some
$i \leq n$ such that $v \LRb{\xR_i} u$. Actually, we know that
$u \Rab{\xR_i}{M_2} v$ since otherwise we would have both $s > v$
and $v > s$ by the $\xB$-compatibility of $>$. Repeated applications
of \lemref{labelpreserve}(1,2) therefore yield a conversion between
$u$ and $v$ of the form
\[
u \Rao{\xR_n}{M_2}{=} \cdot \LRao{\xR_n \cUp \xB}{N}{*} v
\]
where $\SET{u} \gtmul N$. By definition, $s' \gtrsim u$ for some
$s' \in M_1$ and therefore $M_1 \gtmul N$, which means that the
corresponding condition required by peak-and-cliff decreasingness is
fulfilled. Overall, it follows that $\xR_n$ is peak-and-cliff
decreasing and therefore Church--Rosser modulo $\xB$.
\end{proof}
\noindent 
Note that the proofs of the previous theorem and \thmref{pcd} do not
require multiset orders induced by quasi-orders but use multiset
extensions of proper $\xB$-compatible reduction orders which are
easier to work with. This could be achieved by defining peak-and-cliff
decreasingness in such a way that well-founded orders suffice for the
abstract setting. However, the usage of multiset orders based on
$\xB$-compatible reduction orders as well as a notion of labeled
rewriting which allows us to label steps with $\xB$-equivalent terms
are crucial in order to establish peak-and-cliff decreasingness for
TRSs.

As we have established correctness of \ia, it is natural to
ask the question whether \ia is also \emph{complete}, i.e., 
can \ia generate a complete presentation whenever there exists one.
Contrary to ordered completion where there are known completeness
results (see e.g.\ \cite{BDP89,D91}),
this is not possible in our setting.
Consider the ES $\xE$ \cite{BD94} consisting of the three equations
\begin{align*}
\m{1} \cdot (-x + x) &\approx \m{0} &
\m{1} \cdot (x + -x) &\approx x + -x &
-x + x &\approx y + -y
\end{align*}
which admits the following complete presentation $\xR$:
\begin{align*}
\m{1} \cdot \m{0} &\R \m{0} & x + -x &\R \m{0} & -x + x &\R \m{0}
\end{align*}
In standard completion, only the first two equations can be
oriented from left to right but no further step is possible. In \ia,
we have the same situation if $\xB = \varnothing$, but we can choose
$\xB$ as some nonempty subset of $\xE$. Note that the first and
third equation are not eligible since they violate the assumption
$\var{\ell} = \var{r}$. Hence,
$\xB = \SET{\m{1} \cdot (x + -x) \approx x + -x}$ is the only other option.
After orienting the first equation, again no further step is possible
since there is also no overlap between $\m{1} \cdot (-x + x)$ and
$x + -x$.

\section{Bachmair's Inference System}
\label{sec:bachmair}

As already mentioned, the inference system proposed by Avenhaus
\cite{A95} is essentially the same as \ia. The only other inference
system for $\xB$-completion for left-linear TRSs is due to Bachmair
\cite{B91}. We investigate a slightly modified version of this
inference system where arbitrary local peaks are deducible and the
encompassment condition from the \col rule is removed as we only
consider finite runs. The resulting system will be called \ib.
Note that the purpose of the change in the \ded rule of Bachmair's
system is to eliminate this unimportant difference to \ia. The following
results would still hold if we limited \ded in \ib to (prime) critical
pairs but aligning \ia and \ib would be unnecessarily complicated.
Furthermore, deducing arbitrary local peaks offers a simpler definition
of \ded~which gives implementations more freedom. However, we are not
aware of any work that has investigated whether deducing non-critical
peaks can be beneficial in completion.

The main difference between \ia and \ib is that in \ib one may only
use the standard rewrite relation $\RbR$ for simplifying equations
and composing rules. This allows us to deduce local cliffs as
equations. The goal of this section is to establish correctness of
\ib via a simulation by \ia.

\begin{defi}
\label{def:ib}
The inference system \ib is the same as \ia but with
rewriting in \com and \sip restricted to $\RbR$ and the
following rule which replaces the two deduction rules of \ia:
\[
\ded \quad \ds \frac{\xE,\xR}{\xE \cup \SET{s \approx t},\xR} \quad
\text{if $s \LbR \cdot \Rb{\xR \cUp \xB^\pm} t$}
\]
\end{defi}

\begin{exa}
Recall the ES $\xE = \SET{x + \m{0} \approx x}$ of \exaref{cliffscol},
where $+$ is an AC function symbol.  AC completion based on \ib
proceeds as follows:
\[
(\xE,\varnothing)
\,\iB\, (\varnothing,\SET{x + \m{0} \R x}) 
\,\iB\, (\SET{\m{0} + x \approx x},\SET{x + \m{0} \R x})
\,\iB\, \cdots
\]
The second step obtained by \ded reveals a
main difference between \ia and \ib.  If \ded in \ia is employed, 
\(
(\varnothing,\SET{x + \m{0} \R x}) \iA
(\varnothing,\SET{\m{0} + x \to x,\; \m{0} + x \R x})
\)
is obtained.
\end{exa}

\begin{defi}
\label{def:fairnessb}
Let $\xE$ be an ES. A finite sequence
\[
\xE_0,\xR_0 \,\iB\, \xE_1,\xR_1 \,\iB\, \cdots \,\iB\, \xE_n,\xR_n
\]
with $\xE_0 = \xE$ and $\xR_0 = \varnothing$ is a \emph{run} for
$\xE$. If $\xE_n \neq \varnothing$, the run \emph{fails}. The run
is \emph{fair} if $\xR_n$ is left-linear and the following inclusion
holds:
\[
\pcp{\xR_n} \cup {\pcppm{\xR_n,\xB^\pm}} \,\subseteq\,
{{\Dab{\xR_n}{\sim}} \cup {\bigcup_{i=0}^n {\LRb{\xE_i}}}}
\]
\end{defi}

In contrast to \defref{fairness}, the fairness condition is the same
for all prime critical pairs since the inference rule \ded of
\ib never produces rewrite rules. In that sense, \ib is closer to
known completion procedures, but as shown before, this comes at the expense
of not being allowed to apply \sip, \col and \com with $\RmB$.
If we want to allow this more general kind of simplification as it is the
case for \ia, local cliffs have to be deduced as rules. Note that this
possibly leads to an increase in critical pairs which one has to consider
in practice, but it can also reduce the number of \ori steps one has to
perform. Since \com and \col (without the encompassment condition) can
emulate the behavior of \sip, it is possible to get the best of both
worlds by using \ia and deferring the computation of
critical pairs with rules which stem from local cliffs until it is
needed to proceed in the completion process. We will now show that
for finite runs, \ia is at least as powerful as \ib which is the main
motivation for the focus on \ia instead of \ib in this article. Moreover,
the simulation result actually allows us to reduce correctness of \ib to
correctness of \ia, so we get this property without any additional effort.

In order to prove that fair and non-failing runs in \ib can be simulated
in \ia, we start with the following technical lemma which intuitively
states that a step in \ib can be simulated by at most one step in
\ia in such a way that the results only differ in the number of oriented
equations.
We denote an application of the rule \ori in an inference system \infsys
by \smash{$\seqao{\infsys}{\phantom{*}}{\textsf{o}}$}.

\begin{lem}
\label{lem:simaux1}
If $(\xE_1,\xR_1) \iB (\xE_2, \xR_2)$ and
$(\xE_1,\xR_1) \seqao{\ib}{*}{\textsf{o}} (\xE_1',\xR_1')$ then
there exists a pair $(\xE_2',\xR_2')$ such that
$(\xE_1',\xR_1') \iA^= (\xE_2',\xR_2')$ and
\smash{$(\xE_2,\xR_2) \seqao{\ib}{*}{\textsf{o}} (\xE_2',\xR_2')$}.
In a picture:
\begin{center}
\begin{tikzpicture}[pair/.style={execute at begin node={\strut},
   execute at end node={\strut},},node distance=1cm]
\node[pair] (1) {$\xE_1,\xR_1$};
\node[pair] (1p) [right=of 1] {$\xE_2,\xR_2$};
\node[pair] (2) [below=of 1] {$\xE_1',\xR_1'$};
\node[pair] (2p) [right=of 2] {$\xE_2',\xR_2'$};
\path (1) -- (1p) node [pos=.5] {$\iB$};
\path (2) -- (2p) node [pos=.5] {$\iA^=$};
\path (1) -- (2) node [pos=.5]
  {\rotatebox{-90}{$\seqao{\ib}{*}{\textsf{o}}$}};
\path (1p) -- (2p) node [pos=.5]
  {\rotatebox{-90}{$\seqao{\ib}{*}{\textsf{o}}$}};
\end{tikzpicture}
\end{center}
\end{lem}

\begin{proof}
Let $>$ be a fixed $\xB$-compatible reduction order which is used in
both \ia and \ib. From
\smash{$(\xE_1,\xR_1) \seqao{\ib}{*}{\textsf{o}} (\xE_1',\xR_1')$} we
obtain $\xE_1' \subseteq \xE_1$, $\xR_1 \subseteq \xR_1'$ and
$\xE_1 \setminus \xE_1' \subseteq \xR_1' \cup \xR_1'^{-1}$. In order
to simplify the formulation, we will refer to the sequence of
\ori steps between a pair and its primed variant as the invariant.
We proceed by a case analysis on the rule applied in the inference
step $(\xE_1,\xR_1) \iB (\xE_2,\xR_2)$.
\begin{itemize}[label=$\triangleright$]
\item
In the case of \ded, we apply the same rule in \ia. For
local peaks $s \Lb{\xR_1} \cdot \Rb{\xR_1} t$ we have
$\xE_2 = \xE_1 \cup \SET{s \approx t}$, $\xR_2 = \xR_1$,
$\xE_2' = \xE_1' \cup \SET{s \approx t}$ and $\xR_2' = \xR_1'$. For
local cliffs $s \Lb{\xR_1} \cdot \LRb{\xB} t$ we have
$\xE_2 = \xE_1 \cup \SET{s \approx t}$, $\xR_2 = \xR_1$ as well as
$\xR_2' = \xR_1' \cup \SET{t \R s}$, $\xE_2' = \xE_1'$. In both
cases, the invariant is preserved. \smallskip
\item
Suppose the inference step in \ib orients an equation $s \approx t$. If
${s \approx t} \in \xE_1'$ we perform the same
step in \ia which preserves the invariant. Otherwise,
${s \approx t} \in \xE_1' \setminus \xE_1$ and hence
${v \R w} \in \xR_1'$ where $\SET{v,w} = \SET{s,t}$. In this case, an
empty step in \ia preserves the invariant. \smallskip
\item
If the inference step in \ib deletes an equation
$s \approx t$, it has to be in $\xE_1'$ which enables us to
perform the same step in \ia while preserving the invariant:
Suppose that ${s \approx t} \in \xE_1 \setminus \xE_1'$. Since
the equation is deleted, we have $s \bsim t$. Neither $s > t$ nor
$t > s$ can hold as $>$ is $\xB$-compatible and irreflexive. Hence, the
equation cannot be oriented which contradicts the assumption
\smash{$(\xE_1,\xR_1) \seqao{\ib}{*}{\textsf{o}} (\xE_1',\xR_1')$}.
\smallskip
\item
If the inference step in \ib is \com or \col, the same
step can be performed in \ia while preserving the invariant as
$\xR_1 \subseteq \xR_1'$. \smallskip
\item
Finally, suppose that the inference step in
\ib simplifies an equation $s \approx t$. Since the
orientation of equations does not matter in completion, we may
assume that the simplification transforms the equation to
$s' \approx t$ without loss of generality. If
$s \approx t \in \xE_1'$, we perform the same step in
\ia which preserves the invariant. Otherwise, there is a rule
${v \R w} \in \xR_1'$ such that $\SET{v,w} = \SET{s,t}$. If $v = s$
and $w = t$, we use \col for the inference step in \ia which
produces the same equation ${s' \approx t} \in \xE_2'$ which means
that the invariant is preserved. If $v = t$ and $w = s$, we use
\com for the inference step in \ia in order to obtain the rule
$t \R s'$. From $t > s$ and $s > s'$ we obtain $t > s'$ and
therefore the equation $s' \approx t$ can be oriented into the
rule $t \R s'$. Thus, the invariant is preserved. \qedhere
\end{itemize}
\end{proof}
\noindent 
For the proof of the simulation result, we need a slightly different
form of the previous lemma. Analogous to the notation for rewrite
relations, the relation \smash{$\seqao{\infsys}{!}{\textsf{o}}$} denotes
the exhaustive application of the inference rule \ori.

\begin{cor}
\label{cor:simaux2}
If $(\xE_1,\xR_1) \iB (\xE_2, \xR_2)$ and
$(\xE_1,\xR_1) \seqao{\ib}{!}{\textsf{o}} (\xE_1',\xR_1')$ then
$(\xE_1',\xR_1') \iA^* (\xE_2',\xR_2')$ where
\smash{$(\xE_2,\xR_2) \seqao{\ib}{!}{\textsf{o}} (\xE_2',\xR_2')$}.
\end{cor}

\begin{proof}
Let $(\xE_1,\xR_1) \iB (\xE_2,\xR_2)$ and
$(\xE_1,\xR_1) \seqao{\ib}{!}{\textsf{o}} (\xE_1',\xR_1')$.
\lemref{simaux1} yields $(\xE_1',\xR_1') \iA^= (\xE_2',\xR_2')$ and
\smash{$(\xE_2,\xR_2) \seqao{\ib}{*}{\textsf{o}} (\xE_2',\xR_2')$}. It
follows that the orientable equations in $\xE_2'$ are also included
in $\xE_2$. Hence, we can compute
\smash{$(\xE_2',\xR_2') \seqao{\ia}{!}{\textsf{o}} (\xE_3',\xR_3')$},
satisfying $(\xE_1',\xR_1') \iA^* (\xE_3',\xR_3')$ and
\smash{$(\xE_2,\xR_2) \seqao{\ib}{!}{\textsf{o}} (\xE_3',\xR_3')$} as
desired.
\end{proof}

\begin{thm}
\label{thm:simres}
For every fair run $(\xE,\varnothing) \iB^* (\varnothing,\xR)$
there exists a fair run $(\xE,\varnothing) \iA^* (\varnothing,\xR)$.
\end{thm}

\begin{proof}
Assume $(\xE_0,\xR_0) \iB^n (\xE_n,\xR_n)$ where
$\xR_0 = \xE_n = \varnothing$. By $n$ applications of
\corref{simaux2} we arrive at the following situation:
\begin{center}
\begin{tikzpicture}[pair/.style={execute at begin node={\strut},
execute at end node={\strut},},node distance=1cm]
\node[pair] (0) {$\xE_0,\xR_0$};
\node[pair] (0p) [below=of 0] {$\xE_0',\xR_0'$};
\node[pair] (0m) [left=of 0p] {$\xE_0, \xR_0$};
\node[pair] (1) [right=of 0] {$\xE_1,\xR_1$};
\node[pair] (1p) [below=of 1] {$\xE_1',\xR_1'$};
\node[pair] (dots) [right=of 1] {$\cdots$};
\node[pair] (dotsp) [below=of dots] {$\cdots$};
\node[pair] (n) [right=of dots] {$\xE_n,\xR_n$};
\node[pair] (np) [below=of n] {$\xE_n',\xR_n'$};
\path (0) -- (1) node [pos=.5] {$\seq{\ib}$};
\path (0m) -- (0p) node [pos=.5] {$\seqao{\ia}{!}{\textsf{o}}$};
\path (0p) -- (1p) node [pos=.5] {$\seqa{\ia}{*}$};
\path (1) -- (dots) node [pos=.5] {$\seq{\ib}$};
\path (1p) -- (dotsp) node [pos=.5] {$\seqa{\ia}{*}$};
\path (dots) -- (n) node [pos=.5] {$\seq{\ib}$};
\path (dotsp) -- (np) node [pos=.5] {$\seqa{\ia}{*}$};
\path (0) -- (0p) node [pos=.5]
 {\rotatebox{-90}{$\seqao{\ib}{!}{\textsf{o}}$}};
\path (1) -- (1p) node [pos=.5]
 {\rotatebox{-90}{$\seqao{\ib}{!}{\textsf{o}}$}};
\path (n) -- (np) node [pos=.5]
 {\rotatebox{-90}{$\seqao{\ib}{!}{\textsf{o}}$}};
\end{tikzpicture}
\end{center}
The following two statements hold:
\begin{enumerate}
\item
For $0 \leq i \leq n$, all orientable equations in $\xE_i$ are
in $\xR_i'$ (possibly reversed) and the other equations are in
$\xE_i'$. \smallskip
\item
$\pcppm{\xR_n',\xB^\pm}$ is a set of orientable equations.
\end{enumerate}
Statement (1) is immediate from the simulation relation
\smash{$\seqao{\ib}{!}{\textsf{o}}$} and statement (2) follows from
$\xB$-compatibility of the used reduction order together with the
fact that every (prime) critical pair is connected by one
$\xR_n$-step and one $\xB$-step. Furthermore, $\xE_n = \varnothing$
implies $\xE_n' = \varnothing$ as well as $\xR_n = \xR_n'$. Hence,
we obtain fairness of the run in \ia by showing the following
inclusions:
\begin{align*}
\pcp{\xR_n'} &\,\subseteq\, {{\Dab{\xR_n'}{\sim}} \cup
{\bigcup_{i=0}^n {\LRb{\xE_i' \cUp \xR_i'}}}} &
\pcppm{\xR_n',\xB^\pm} &\,\subseteq\,
{{\Dab{\xR_n'}{\sim}} \cup {\bigcup_{i=0}^n {\LRb{\xR_i'}}}}
\end{align*}
Let ${s \approx t} \in \pcp{\xR_n'}$. By fairness of the run in
\ib we obtain $s \Dab{\xR_n'}{\sim} t$ or $s \LRb{\xE_k} t$ for
some $k \leq n$. In the former case, we are immediately done. In
the latter case we obtain \smash{$s \LRb{\xE_k' \cUp \xR_k'} t$} from (1)
as desired. Now, let ${s \approx t} \in \pcppm{\xR_n',\xB^\pm}$. By
fairness of the run in \ib we obtain \smash{$s \Dab{\xR_n'}{\sim} t$}
or $s \LRb{\xE_k} t$ for some $k \leq n$. Again, we are immediately
done in the former case. In the latter case we have
$s \LRb{\xR_k'} t$ because of (1) and (2). Therefore, the run in
\ia is fair.
\end{proof}
\noindent 
The previous theorem is an important simulation result which justifies
the emphasis on \ia in this article. Moreover, together with
\thmref{correctness} the correctness of the inference system
\ib is an easy consequence.

\begin{cor}
\label{cor:bcorrectness}
Every fair and non-failing run for $\xE$ in \ib produces a
$\xB$-complete presentation of $\xE$.
\end{cor}

\section{Canonicity}
\label{sec:canonicity}

Complete representations resulting from completion may have redundant
rules which do not contribute to the computation of normal forms. The
notion of \emph{canonicity} addresses this issue by defining a minimal
and unique representation of a complete TRS for a given reduction
order. In this section, canonicity results for $\xB$-complete TRSs are
presented. After establishing results for abstract rewriting, means
to compute $\xB$-canonical TRSs and the uniqueness of $\xB$-canonical
TRSs are discussed. The results and proofs in this section closely
follow the presentation of canonicity results for standard rewriting
in \cite{HMSW19} by carefully lifting definitions and results to
rewriting modulo $\xB$. To the best of our knowledge, this section
presents the first account of canonicity results for $\xB$-complete TRSs.

\subsection{Results for Abstract Rewriting Modulo}

In the following we assume that
$\xA_1 = \langle A, \Rb{\xA_1} \rangle$ and
$\xA_2 = \langle A, \Rb{\xA_2} \rangle$ are ARSs 
with the same underlying set $A$
and $\sim$ is an
equivalence relation on
$A$.
The upcoming
definition introduces two different notions of equivalence
between ARSs. Both
are needed in order to relate
different representations of the same equational theory,
in particular the relation between complete presentations
and their canonical forms.

\begin{defi}
The ARSs $\xA_1$ and $\xA_2$ are
\emph{(conversion) equivalent modulo $\sim$} if
${\bigLRab{\xA_1}{*}} = {\bigLRab{\xA_2}{*}}$
and \emph{normalization equivalent modulo $\sim$} if
\smash{${\Rab{\xA_1}{!} \cdot \sim} = {\Rab{\xA_2}{!} \cdot \sim}$}.
\end{defi}

\begin{exa}
Let $A = \SET{\m{a}, \m{b}, \m{c}}$, $\m{a} \sim \m{b}$,
$\Rb{\xA_1} = \SET{(\m{a},\m{c})}$ and $\Rb{\xA_2} = \SET{(\m{b},\m{c})}$.
Then, $\xA_1$ and $\xA_2$ are conversion equivalent modulo
$\sim$ but not normalization equivalent modulo $\sim$:
We have $\m{a} \Rab{\xA_1}{!} \m{c}$ but $\m{a} \in \nf{\xA_2}$
and $\m{a} \not\sim \m{c}$.
\end{exa}

Next we show two fundamental properties of normalization
equivalence modulo $\sim$ which are needed in \thmref{ddot}.

\begin{lem}
\label{lem:normeqmodsim_eqmodsim}
If $\xA_1$ and $\xA_2$
are normalization equivalent modulo $\sim$ and terminating, they are
equivalent modulo $\sim$.
\end{lem}

\begin{proof}
We prove ${\bigLRab{\xA_1}{*}} \subseteq {\bigLRab{\xA_2}{*}}$ by
induction on the length of conversions in $\xA_1$. The claim follows
by symmetry. If $a \bigLRab{\xA_1}{0} b$ then $a = b$ and therefore
also $a \bigLRab{\xA_2}{0} b$. For
$a \bigLRab{\xA_1}{n} a' \bigLRb{\xA_1} b$ the induction hypothesis
yields $a \bigLRab{\xA_2}{*} a'$. If $a' \sim b$, the result follows
immediately. Otherwise, either $a' \Rb{\xA_1} b$ or
$b \Rb{\xA_1} a'$. We only consider the first case by again
exploiting symmetry. Since $\xA_1$ and $\xA_2$ are normalization
equivalent and terminating, there is an element $c$ such that
$b \Rab{\xA_1}{!} \cdot \sim c$ and $b \Rab{\xA_2}{!} \cdot \sim c$
and therefore also $a' \Rab{\xA_1}{!} \cdot \sim c$ and
$a' \Rab{\xA_2}{!} \cdot \sim c$. Hence, we can connect $a'$ and $b$
by a conversion in $\xA_2$ as desired.
\end{proof}

\begin{lem}
\label{lem:canonicityabstractres}
Let $\nf{\xA_2} \subseteq \nf{\xA_1}$ and
${\Rb{\xA_2}} \subseteq {({\Rb{\xA_1}}/{\sim})^{+}}$. If $\xA_1$ is
complete modulo $\sim$ then $\xA_2$ is complete modulo $\sim$ and
normalization equivalent modulo $\sim$ to $\xA_1$.
\end{lem}

\begin{proof}
From the inclusion
${\Rb{\xA_2}} \subseteq {({\Rb{\xA_1}}/{\sim})^{+}}$ as well as the
termination of $\xA_1$ modulo $\sim$ we obtain that $\xA_2$ is
terminating modulo $\sim$. Next, we show
\smash{${\Rab{\xA_2}{!} \cdot \sim} \subseteq
{\Rab{\xA_1}{!} \cdot \sim}$}. The inclusion
${\Rb{\xA_2}} \subseteq {({\Rb{\xA_1}}/{\sim})^{+}}$ gives rise to
\smash{${\Rab{\xA_2}{!}} \subseteq {\bigLRab{\xA_1}{*}}$}. Since $\xA_1$
is Church--Rosser modulo $\sim$ and $\nf{\xA_2} \subseteq \nf{\xA_1}$
we have \smash{${\Rab{\xA_2}{!}} \subseteq {\Rab{\xA_1}{!} \cdot \sim}$}
and therefore
\smash{${\Rab{\xA_2}{!} \cdot \sim} \subseteq {\Rab{\xA_1}{!} \cdot \sim}$}
as desired. For the reverse inclusion
consider \smash{$a \Rab{\xA_1}{!} b' \sim b$}. The termination of $\xA_2$
modulo $\sim$ implies the termination of $\xA_2$, so there is some
$c$ such that $a \Rab{\xA_2}{!} \cdot \sim c$ and thus
$a \Rab{\xA_1}{!} c' \sim c$. Now, the Church--Rosser modulo $\sim$
property of $\xA_1$ yields $b' \sim c'$. We obtain $b \sim c$ and
therefore $a \Rab{\xA_2}{!} \cdot \sim b$. Hence, $\xA_1$ and
$\xA_2$ are normalization equivalent modulo $\sim$. Finally, the
Church--Rosser modulo $\sim$ property of $\xA_2$ follows from the
sequence of inclusions
\[
{\bigLRab{\xA_2}{*}} \,\subseteq\, {\bigLRab{\xA_1}{*}}
\,\subseteq\, {\Rab{\xA_1}{!} \cdot \sim \cdot \Lab{\xA_1}{!}}
\,\subseteq\, {\Rab{\xA_2}{!} \cdot \sim \cdot \Lab{\xA_2}{!}}
\]
which are justified by
${\Rb{\xA_2}} \subseteq {({\Rb{\xA_1}}/{\sim})^{+}}$, the fact that
$\xA_1$ is complete modulo $\sim$ and normalization equivalence
modulo $\sim$ of $\xA_1$ and $\xA_2$, respectively.
\end{proof}

\subsection{Computing \texorpdfstring{$\xB$}{B}-Canonical Term Rewrite
Systems}

Intuitively, $\xB$-complete TRSs need to have a variety of left-hand
sides of rules in order to match every possible $\xB$-equivalent term
of a reducible term. However, for the right-hand sides of rules only
one representative of the $\xB$-equivalence class suffices. This
rationale is reflected in the following series of definitions.

\begin{defi}
Two rules $\ell \R r$ and $\ell' \R r'$ are
\emph{right-$\xB$-equivalent variants} if there exists a renaming
$\sigma$ such that $\ell\sigma = \ell'$ and $r\sigma \bsim r'$. We
write $\xR_1 \llrbsim \xR_2$ if every rule of $\xR_1$ has a
right-$\xB$-equivalent variant in $\xR_2$ and vice versa. For any
TRS $\xR$, the TRS $\xR_{\llrbsim}$ denotes a right-$\xB$-equivalent
variant-free version, i.e., only one representative of every
equivalence class of right-$\xB$-equivalent variants is present.
\end{defi}

\begin{exa}
From the TRS $\xR$ consisting of the rules
\begin{align*}
\m{s}(x) \times y &\cR (x \times y) + y &
y \times \m{s}(x) &\cR (x \times y) + y &
\m{s}(x) \times y &\cR y + (x \times y)
\end{align*}
where $+$ and $\times$ are AC function symbols we obtain
$\xR_{\llrbsim}$ by removing the third rule:
\begin{align*}
\m{s}(x) \times y &\cR (x \times y) + y &
y \times \m{s}(x) &\cR (x \times y) + y
\end{align*}
\end{exa}

The relation $\llrbsim$ weakens the notion of literal similarity of
TRSs ($\doteq$) to the setting of the inference system \ia: While
composing with $\RmB$ is allowed, left-hand sides may only be rewritten
with $\RbR$.

\begin{defi}
\label{def:bcanonical}
A TRS $\xR$ is \emph{left-reduced} if for every rewrite rule
$\ell \to r \in \xR$, $\ell$ is a normal form of
$\xR \setminus \SET{\ell \to r}$ and
\emph{right-$\xB$-reduced} if for every rewrite rule
${\ell \R r} \in \xR$, $r$ is a normal form with respect to
$\RmB$. We say that $\xR$ is \emph{canonical modulo $\xB$}
if it is complete modulo $\xB$, left-reduced and right-$\xB$-reduced.
\end{defi}

For standard rewriting, canonical systems are defined in \cite{M83}
as TRSs which are complete, left-reduced and \emph{right-reduced}. Here,
right-reduced is defined just as right-$\xB$-reduced but with
$\RbR$ instead of $\RmB$. Intuitively, canonical systems are minimal
complete systems because the rules cannot be used to simplify (other) rules
anymore. In our setting, it is important to note that while right-hand
sides can be simplified with $\RmB$, for left-hand sides
only simplification with $\RbR$ is considered. This ensures
that all terms which are reducible with respect to $\RmB$ can also be
rewritten with $\RbR$. Note that this also reflects the definition of
\col and \com in \ia.

The next definition is an extension of M\'etivier's proposed procedure
\cite[Theorem 7]{M83} to transform a complete system into a canonical
system. As in the original definition, a $\xB$-complete TRS is transformed
into a $\xB$-canonical TRSs in two stages: First, a right-$\xB$-reduced
version without right-$\xB$-equivalent variants is constructed. After
that, superfluous rules are removed in order to obtain a left-reduced
system. Hence, the overall result is canonical modulo $\xB$. Needless to
say, this is only computable if $\xB$ has a decidable equational theory. A
presentation of M\'etivier's concept for standard rewriting in the style
of our upcoming definition can be found in \cite{HMSW19}.

\begin{defi}
\label{def:ddot}
Given a $\xB$-terminating TRS $\xR$, the TRSs $\dot{\xR}$ and
$\ddot{\xR}$ are defined as follows:
\begin{align*}
\dot{\xR} &\,=\,
\SET{\ell \to r{\Db{\xR/\xB}} \mid {\ell \to r} \in \xR}_{\llrbsim} \\
\ddot{\xR} &\,=\,
\SET{\ell \to r \in \dot{\xR} \mid \ell \in
\nf{\dot{\xR} \setminus \SET{\ell \to r}}}
\end{align*}
Here, $r{\Db{\xR/\xB}}$ denotes an arbitrary normal form of $r$ with
respect to $\RmB$.
\end{defi}

\begin{exa}
Consider again the ES $\xE$ from \exaref{a95_ex_4_2_15b}.
If neither \col nor \com is used, AC completion results in the
complete presentation $\xR = \SET{1',2',3',4\,\text{--}\,10,13,14,15}$.
Because the right-hand sides of $14$ and $15$ are reducible by
$\xR/\xB$ and also because $13$ is a right-$\xB$-equivalent variant of 
$1'$,
the TRS $\dot{\xR}$ consists of $1',2',3',4\,\text{--}\,10$ and the
modified rules:
\begin{align*}
\m{f}(x + (y + z)) &\nR{14''} (\m{f}(x) + \m{f}(y)) + \m{f}(z)
&
\m{f}((x + y) + z) &\nR{15''} \m{f}(x) + (\m{f}(y) + \m{f}(z))
\end{align*}
Since the left-hand sides of $5\,\text{--}\,10$, $14''$, and $15''$ are
reducible by $1'$, $3'$ or $4$, they are excluded from $\ddot{\xR}$. Thus,
$\ddot{\xR} = \SET{1',2',3',4}$ is obtained.
As expected, it is the result of performing completion with \col and \com.  
\end{exa}

The following example shows that for preserving the equational theory,
it does not suffice to make $\dot{\xR}$ variant-free
in the sense of $\doteq$.

\begin{exa}
Suppose we change the definition to
\[
\dot{\xR} \,=\, \SET{\ell \to r{\Db{\xR/\xB}} \mid {\ell \to r} \in \xR}
\]
where variant-freeness is implicit. Consider the TRS
$\xR$ consisting of the rules
\begin{align*}
\m{f}(x + y) &\cR \m{f}(x) + \m{f}(y) &
\m{f}(x + y) &\cR \m{f}(y) + \m{f}(x)
\end{align*}
where $+$ is an AC function symbol. We have $\dot{\xR} = \xR$
since the two rules are not variants and therefore
\smash{$\ddot{\xR} = \varnothing$} which obviously induces a
different equational theory than $\xR$. However, by using
\defref{ddot} we obtain e.g.\ the following $\xB$-canonical TRS:
\[
\ddot{\xR} \,=\, \dot{\xR} \,=\,
\SET{\m{f}(x + y) \R \m{f}(x) + \m{f}(y)}
\]
\end{exa}

Before we can prove the main result of this section, we need the
following technical lemma.

\begin{lem}
\label{lem:ddotaux}
If $\xR$ is a $\xB$-complete \textup{TRS} and $\dot{\xR}$ is
Church--Rosser modulo $\xB$ then
$\nf{\ddot{\xR}} \subseteq \nf{\dot{\xR}}$.
\end{lem}

\begin{proof}
Suppose that $\xR$ is $\xB$-complete and $\dot{\xR}$ is
Church--Rosser modulo $\xB$. We prove the contrapositive
\smash{$s \notin \nf{\dot{\xR}} \implies
s \notin \nf{\ddot{\xR}}$}
for which it suffices to show that
\smash{$\ell \notin \nf{\ddot{\xR}}$} whenever
${\ell \R r} \in \dot{\xR}$ due to closure of rewriting under
contexts and substitutions. We prove this by induction on $\ell$
with respect to the well-founded order $\enc$. If
\smash{${\ell \R r} \in \ddot{\xR}$}, the claim follows
immediately. Otherwise,
$\ell \notin \nf{\dot{\xR} \setminus \SET{\ell \R r}}$, i.e., there
is a rule ${\ell' \R r'} \in \dot{\xR}$ with $\ell \enceq
\ell'$. Furthermore, the rules $\ell \R r$ and $\ell' \R r$ are not
right-$\xB$-equivalent variants by the definition of
$\dot{\xR}$. By definition,
${\enceq} = {\enc} \cup {\doteq}$, so we continue by a case analysis
on $\ell \enceq \ell'$.
\begin{itemize}[label=$\triangleright$]
\item
If $\ell \enc \ell'$, the induction hypothesis yields
$\ell' \notin \nf{\ddot{\xR}}$ and therefore
$\ell \notin \nf{\ddot{\xR}}$ as desired. \smallskip
\item
If $\ell \doteq \ell'$ then by definition there exists a renaming
$\sigma$ such that $\ell = \ell'\sigma$. From the right-$\xB$-reducedness
of $\dot{\xR}$ we conclude that $r$ and $r'$ are normal forms with
respect to \smash{$\Rb{\dot{\xR}/\xB}$}. Since normal forms are closed
under renamings, this also holds for $r'\sigma$. Together with
\smash{$r \Lb{\dot{\xR}} \ell = \ell'\sigma \Rb{\dot{\xR}} r'\sigma$}
and the Church--Rosser modulo $\xB$ property of $\dot{\xR}$ we
obtain $r \bsim r'\sigma$. Thus, $\ell \R r$ and $\ell' \R r'$ are
right-$\xB$-equivalent variants. This contradicts the definition
of $\dot{\xR}$. Therefore, this case cannot happen. \qedhere
\end{itemize}
\end{proof}

\begin{thm}
\label{thm:ddot}
If $\xR$ is a $\xB$-complete \textup{TRS} then $\ddot{\xR}$ is a
\textup{TRS} which is normalization equivalent modulo $\xB$ to $\xR$,
conversion equivalent modulo $\xB$ to $\xR$ and canonical modulo
$\xB$.
\end{thm}

\begin{proof}
Let $\xR$ be a $\xB$-complete TRS. By definition,
\smash{$\ddot{\xR} \subseteq \dot{\xR} \subseteq
{\RmB[+]}$}. Since $\xR$ and $\dot{\xR}$ have the same
left-hand sides, their normal forms coincide. An application of
\lemref{canonicityabstractres} yields that $\dot{\xR}$ is
$\xB$-complete and normalization equivalent modulo $\xB$ to $\xR$.
Furthermore, $\nf{\ddot{\xR}} = \nf{\dot{\xR}}$: The
inclusion $\nf{\ddot{\xR}} \subseteq \nf{\dot{\xR}}$
follows from \lemref{ddotaux} and the inclusion
$\nf{\dot{\xR}} \subseteq \nf{\ddot{\xR}}$ is a consequence
of $\ddot{\xR} \subseteq \dot{\xR}$. Now, another
application of \lemref{canonicityabstractres} yields that
$\ddot{\xR}$ is also $\xB$-complete and normalization
equivalent modulo $\xB$ to $\xR$. The TRS $\dot{\xR}$ is
right-$\xB$-reduced by definition. Since
$\ddot{\xR} \subseteq \dot{\xR}$, $\ddot{\xR}$ is
also right-$\xB$-reduced. Together with the already established fact
that $\nf{\dot{\xR}} = \nf{\ddot{\xR}}$, left-reducedness
of $\ddot{\xR}$ follows from its definition. Hence,
$\ddot{\xR}$ is canonical modulo $\xB$. Finally,
\lemref{normeqmodsim_eqmodsim} establishes that $\ddot{\xR}$
is not only normalization equivalent modulo $\xB$ but also
(conversion) equivalent modulo $\xB$ to $\xR$.
\end{proof}
\noindent 
Due to the constructive nature of \defref{ddot}, the
previous theorem states that given a $\xB$-complete TRS $\xR$ it is
possible to compute (if the equational theory of $\xB$ is decidable)
a $\xB$-canonical TRS which is equivalent modulo
$\xB$ to $\xR$. An inspection of \defref{bcanonical}
further reveals that the inference system \ia can produce
$\xB$-canonical TRSs: If
\smash{$(\xE,\varnothing) \seqa{\ia}{*} (\varnothing,\xR)$} is a fair run
and neither \com nor \col are applicable, then $\xR$ is not only
complete but also canonical modulo $\xB$. The following lemma shows that
this also holds for the inference system \ib.

\begin{lem}
\label{lem:ibcanonical}
If $(\xE,\varnothing) \seqa{\ib}{*} (\varnothing,\xR)$ is a fair
run and neither \com nor \col are applicable, then $\xR$ is
$\xB$-canonical.
\end{lem}

\begin{proof}
According to \corref{bcorrectness}, the TRS $\xR$
is a $\xB$-complete presentation of $\xE$. Furthermore, $\xR$ is
left-reduced by definition as \col is not applicable. For a proof
by contradiction, suppose that $\xR$ is not
right-$\xB$-reduced. Hence, there exists a rule $\ell \R r$ in $\xR$
which can be modified to $\ell \R r'$ where
\smash{$r \RmB[!] r'$}.
Note that $r$ and $r'$ are not equivalent modulo $\xB$ as $\xR$
is terminating modulo $\xB$. From the $\xB$-completeness of $\xR$
we obtain \smash{$r \Rab{\xR}{!} \cdot \bsim^{} \cdot \Lab{\xR}{!} r'$}.
Since $r'$ is a normal form with respect to $\RmB$, we
have $r \Rab{\xR}{+} \cdot \bsim^{} r'$ which contradicts our
assumption that \com is not applicable.
\end{proof}
\noindent 
Therefore, both \ia and \ib can produce canonical systems due to
the availability of \com and \col which is also referred to as
\emph{inter-reduction}. If a given completion procedure lacks
inter-reduction, it is an instance of \emph{elementary completion}.
Note that the original completion procedure by Knuth and
Bendix~\cite{KB70} performs elementary completion. The next example
shows that there is a big difference between \ia and \ib with
respect to elementary completion. In particular, \ia cannot
simulate \ib if both are restricted to elementary completion: 
In order to simulate \ib's \sip in \ia, \com and \col
are needed as can be seen in the proof of \lemref{simaux1}.

\begin{exa}
\label{exa:ia_elementary_completion}
Consider the ESs $\xE = \SET{x \cdot \m{1} \approx x}$ and
$\xB = \SET{(x \cdot y) \cdot z \approx x \cdot (y \cdot z)}$
as well as a corresponding run in \ia. There is only one option to
orient the only equation into a terminating rewrite rule, namely
$x \cdot \m{1} \R x$. From the critical peak
\smash{$x \cdot y \Lo{}{\epsilon} (x \cdot y) \cdot \m{1}
\LRo{}{\epsilon} x \cdot (y \cdot 1)$} we can deduce the rule
$x \cdot (y \cdot 1) \R x \cdot y$. In \ib restricted to elementary
completion, the result of \ded would be the corresponding equation
which can be simplified to a trivial equation and therefore
deleted. In \ia restricted to elementary completion, however,
there are no means of removing or even altering the deduced
rule. Hence, the rule will be used to deduce more critical pairs
with the associativity axiom which are again kept as
rules. Clearly, this process cannot terminate as overlaps with the
associativity axioms can increase the size of the left-hand side
of the rules without bound.
\end{exa}

The previous example shows that for many ESs, there will not
be a finite run in \ia restricted to elementary completion. Hence,
the usage of inter-reduction and the resulting definition of
canonicity is crucial for the success of \ia in generating finite
solutions to validity problems.
We close the section by stating the fact that even for reduced TRSs
which are not complete, there may be fewer prime critical pairs
than ordinary critical pairs. In Examples \ref{exa:accp} and \ref{exa:acpcp}
we already observed that for a complete TRS.

\begin{exa}
\label{exa:fewerpcp}
Consider the left-linear reduced TRS $\xR$:
\begin{align*}
\m{f}(\m{g}(x,\m{a})) &\R x
&
\m{f}(\m{g}(\m{a}, x)) &\R \m{f}(\m{f}(\m{a}))
&
\m{g}(\m{a},\m{a}) &\R \m{a}
\end{align*}
The TRS admits the three critical pairs stemming from the following critical peaks:
\begin{align*}
&  \peak{\underline{\m{f}(\m{g}(\m{a}, \m{a}))}}{\m{f}(\m{f}(\m{a}))}{\m{a}}
&& \peak{\underline{\m{f}(\m{g}(\m{a}, \m{a}))}}{\m{a}}{\m{f}(\m{f}(\m{a}))}
&& \peak{\m{f}(\underline{\m{g}(\m{a}, \m{a})})}{\m{f}(\m{a})}{\m{a}}
&& \peak{\m{f}(\underline{\m{g}(\m{a}, \m{a})})}{\m{f}(\m{a})}{\m{f}(\m{f}(\m{a}))}
\end{align*}
As the first two are not prime, $\pcp{\xR}$ only consists of the two
equations $\m{f}(\m{a}) \approx \m{a}$ and
$\m{f}(\m{f}(\m{a})) \approx \m{a}$.  Note that the TRS is completed by
adding the rule $\m{f}(\m{a}) \to \m{a}$ corresponding to the former
equation.
\end{exa}

\subsection{Uniqueness of \texorpdfstring{$\xB$}{B}-Canonical Term
Rewrite Systems}

The main result of this section states that $\xB$-canonical TRSs which
are compatible with the same $\xB$-compatible reduction order are
unique up to right-$\xB$-equivalent variants. This property justifies
the usage of the term \emph{canonical} and motivates the computation
of TRSs which are canonical modulo $\xB$ as discussed in the last
section. In order to prove this result, we start with an auxiliary
lemma.

\begin{lem}
\label{lem:uniqueaux}
Let $\xR$ be a \textup{TRS} which is right-$\xB$-reduced and let $s$
be a reducible term which is minimal with respect to $\enc$. If
$s \Rab{\xR}{+} \cdot \bsim^{} t$ then $s \R t$ is a
right-$\xB$-equivalent variant of a rule in $\xR$.
\end{lem}

\begin{proof}
Let $\ell \R r$ be the first rule which is applied in the sequence
$s \Rab{\xR}{+} \cdot \bsim^{} t$, so $s \enceq \ell$. Since $s$ is
a minimal reducible term with respect to $\enc$, we have
$s \doteq \ell$. By definition, there is a renaming $\sigma$ such
that $s = \ell\sigma$. Since $\xR$ is right-$\xB$-reduced, $r$ is a
normal form with respect to $\RmB$ and therefore also
$\RbR$. Closure of normal forms under renamings yields
$r\sigma \in \nf{\xR}$. Hence, $t \bsim r\sigma$ and we conclude
that $s \R t$ and $\ell \R r$ are right-$\xB$-equivalent variants.
\end{proof}
\noindent 
In contrast to its counterpart for standard rewriting in
\cite[Theorem~4.9]{HMSW19}, the following lemma needs termination modulo
$\xB$ as an
additional precondition. Note that it still treats a more general case
than the main result of this subsection (\thmref{canonical_unique}) as
the Church--Rosser modulo $\xB$ property is not required.

\begin{lem}
\label{lem:normeq_unique}
\textup{TRS}s which are
normalization equivalent modulo $\xB$, terminating modulo $\xB$,
left-reduced and right-$\xB$-reduced are unique up to
right-$\xB$-equivalent variants.
\end{lem}

\begin{proof}
Let $\xR$ and $\xS$ be TRSs which are normalization equivalent
modulo $\xB$, terminating modulo $\xB$, left-reduced and
right-$\xB$-reduced. As the argument is symmetric, we only show that
every rule of $\xR$ has a right-$\xB$-equivalent variant in
$\xS$. Consider ${\ell \R r} \in \xR$. Note that $\ell$ and $r$ are
not $\xB$-equivalent since otherwise the cycle
$r \bsim \ell \RbR r$ contradicts the fact that $\xR$ is
terminating modulo $\xB$. Furthermore, $r \in \nf{\xR}$ as $\xR$ is
right-$\xB$-reduced. Hence, normalization equivalence modulo $\xB$
of $\xR$ and $\xS$ yields $\ell \Rab{\xS}{+} \cdot \bsim^{}
r$. Moreover, the left-reducedness of $\xR$ yields that $\ell$ is a
minimal $\xR$-reducible term with respect to $\enc$. Now suppose
there exists a rule ${\ell' \R r'} \in \xS$ such that
$\ell \enc \ell'$. Right-$\xB$-reducedness of $\xS$ yields
$r' \in \nf{\xS}$ and termination modulo $\xB$ of $\xS$ implies that
$\ell'$ and $r'$ are not $\xB$-equivalent. Together with
normalization equivalence modulo $\xB$ of $\xR$ and $\xS$ we obtain
$\ell' \Rab{\xR}{+} \cdot \bsim^{} r'$ which contradicts the fact
that $\ell$ is a minimal $\xR$-reducible term with respect to
$\enc$. Therefore, $\ell$ is a minimal $\xS$-reducible term with
respect to $\enc$ and from \lemref{uniqueaux} we obtain that
$\ell \R r$ is a right-$\xB$-equivalent variant of a rule in $\xS$
which concludes the proof.
\end{proof}
\noindent 
In the proof of the following theorem we now just need to establish
the preconditions of the previous lemma.

\begin{thm}
\label{thm:canonical_unique}
Let $\xR$ and $\xS$ be \textup{TRS}s which are equivalent modulo $\xB$ and
canonical modulo $\xB$. If $\xR$ and $\xS$ are compatible with the same
$\xB$-compatible reduction order then $\xR \llrbsim \xS$.
\end{thm}

\begin{proof}
Let $\xR$ and $\xS$ be compatible with the $\xB$-compatible
reduction order $>$. We show that $\xR$ and $\xS$ are normalization
equivalent modulo $\xB$ which allows us to conclude the proof by
\lemref{normeq_unique}. As the argument is symmetric,
we only show
${\Rab{\xR}{!} \cdot \bsim^{}} \subseteq {\Rab{\xS}{!} \cdot \bsim^{}}$.
Consider $s \Rab{\xR}{!} t' \bsim^{} t$. Since $\xS$
is terminating, there exists a term $u$ such that
$t' \Rab{\xS}{!} u$. From the equivalence modulo $\xB$ of $\xR$ and
$\xS$ we obtain \smash{$t' \LRo{\xR \cUp \xB}{*} u$}. As $\xR$ is
Church--Rosser modulo $\xB$ and $t' \in \nf{\xR}$ we have
$u \Rab{\xR}{!} \cdot \bsim^{} t'$. If $t'$ and $u$ are not
$\xB$-equivalent then both $u > t'$ (as
$u \Rab{\xR}{!} \cdot \bsim^{} t'$) and $t' > u$ (as $t' \Rab{\xS}{!} u$),
which contradicts the irreflexivity of $>$.
If $t' \neq u$, we also obtain a contradiction to the irreflexivity
of $>$ from $u \bsim^{} t' \Rab{\xS}{!} u$ together with termination
modulo $\xB$ of $\xS$. Hence, $t' = u$ which means that
$t' \in \nf{\xS}$. Equivalence of $\xR$ and $\xS$ modulo $\xB$
yields \smash{$s \LRo{\xS \cUp \xB}{*} t'$} from which we obtain
$s \Rab{\xS}{!} \cdot \bsim^{} t'$ by $\xB$-completeness of $\xS$
and $t' \in \nf{\xS}$. Finally, $t \bsim t'$ establishes
$s \Rab{\xS}{!} \cdot \bsim^{} t$ as desired.
\end{proof}
\noindent 
The previous result cannot be strengthened to literal similarity as
the following counterexample shows.

\begin{exa}
Consider the ES $\xE$ consisting of the single equation
\[
\m{f}(x + y) \cE \m{f}(x) + \m{f}(y)
\]
where $+$ is an AC function symbol. There are two $\xB$-complete
presentations of $\xE$ consisting of one rule each:
\begin{align*}
\xR &\,=\, \SET{\m{f}(x + y) \R \m{f}(x) + \m{f}(y)} &
\xR' &\,=\, \SET{\m{f}(x + y) \R \m{f}(y) + \m{f}(x)}
\end{align*}
While the rules of the two TRSs are right-$\xB$-equivalent variants,
they are not variants. Note that both systems are compatible with
the same $\xB$-compatible reduction order.
\end{exa}

Note that $\xR$ and $\xR'$ in the previous example can also be
obtained with the inference system \ib. This shows that while the
relation $\llrbsim$ is motivated by the definition of the inference
system \ia, the notion of right-$\xB$-equivalent variants
naturally arises in completion modulo $\xB$ for left-linear TRSs.

\section{AC Completion}
\label{sec:ac-completion}

So far, the theoretical results have been generalized by using
an arbitrary equational theory $\xB$. In practice, however, this
article is concerned with the particular theory AC. The results of
this section allow us to assess the effectiveness of the inference system
\ia in the setting of AC completion.

\subsection{Limitations of Left-Linear AC Completion}
\label{sec:limitations}

In addition to the restriction to left-linear rewrite rules, the
following example demonstrates another severe limitation of the
inference system \ia previously unmentioned in the literature.

\begin{exa}
\label{exa:ia_diverge}
Consider the ES $\xE$ consisting of the equations
\begin{align*}
\andop(\m{0},\m{0}) &\cE \m{0} &
\andop(\m{1},\m{1}) &\cE \m{1} &
\andop(\m{0},\m{1}) &\cE \m{0}
\end{align*}
where $\andop$ is an AC function symbol. There is only one way to
orient each equation. Furthermore, there are no critical pairs
between the resulting rewrite rules. Hence, using the inference
system \ia we arrive at the intermediate TRS
\begin{align*}
\andop(\m{0},\m{0}) &\cR \m{0} &
\andop(\m{1},\m{1}) &\cR \m{1} &
\andop(\m{0},\m{1}) &\cR \m{0}
\end{align*}
where the only possible next step is to deduce local cliffs. We will
now show that this has to be done infinitely many times. Note that
an AC-complete presentation $\xR$ of $\xE$ has to be able to rewrite
any term that is AC-equivalent to a reducible term. Consider the
infinite family of terms
\begin{align*}
s_0 &\,=\, \andop(\m{0},\m{1}) &
s_1 &\,=\, \andop(\andop(\m{0},x_1),\m{1}) &
s_2 &\,=\, \andop(\andop(\andop(\m{0},x_1),x_2),\m{1}) & \cdots \\
\intertext{as well as}
t_0 &\,=\, \m{0} &
t_1 &\,=\, \andop(\m{0},x_1) &
t_2 &\,=\, \andop(\andop(\m{0},x_1),x_2) & \cdots
\end{align*}
Clearly, $s_n \LRab{\xE \cUp \AC}{*} t_n$ for all $n \in \mathbb{N}$
and therefore also $s_n \Dab{\xR}{\sim} t_n$ for all
$n \in \mathbb{N}$, but this demands infinitely many rules in $\xR$:
For each $s_n$ there is an AC-equivalent term such that the
constants $\m{0}$ and $\m{1}$ are next to each other which allows
us to rewrite it using the rule $\andop(\m{0},\m{1}) \R \m{0}$. However,
with $n$ also the number of variables between
these constants increases which requires $\xR$ to have infinitely
many rules since rewrite rules can only be applied before the
representation modulo AC is changed.
\end{exa}

Note that there is nothing special about this example except the fact
that it contains at least one equation which can only be oriented such
that the left-hand side contains an AC function symbol where both
arguments have ``structure'', i.e., both arguments contain a
function symbol which is different from the original AC function
symbol. As a consequence, the necessity of infinitely many rules
applies
to all equational systems which have this property. Needless to say,
this means that for a large class of equational systems the
corresponding AC-canonical presentation (in the left-linear sense) is
infinite if it exists. This observation is in stark contrast to the
properties of general AC completion as presented in the next section
which can complete the ES $\xE$ from \exaref{ia_diverge} into a finite
AC-canonical TRS by simply orienting all rules from left to
right. The following example shows that given a nonempty context,
the same effect can also be seen even when only one argument
contains a different function symbol.

\begin{exa}
\label{exa:ia_diverge2}
Consider the ES $\xE$ consisting of the equation
$\m{s}(\m{p}(x) + y) = x + y$ where $+$ is an AC function symbol.
Note that orienting the equation from left to right is the only
way to convert $\xE$ into a terminating TRS $\xR$. The rule has
no critical pairs with itself, but as in the previous example,
adding critical pairs with $\AC$ does not terminate. This can
be seen by considering the infinite family of terms
\begin{align*}
s_0 &\,=\, \m{s}(\m{p}(x) + y_1) &
s_1 &\,=\, \m{s}((\m{p}(x) + y_1) + y_2) &
s_2 &\,=\, \m{s}(((\m{p}(x) + y_1) + y_2) + y_3) & \cdots \\
\intertext{as well as}
t_0 &\,=\, x + y_1 &
t_1 &\,=\, (x + y_1) + y_2 &
t_2 &\,=\, ((x + y_1) + y_2) + y_3 & \cdots
\end{align*}
Again, $s_n \LRab{\xE \cUp \AC}{*} t_n$ for all $n \in \mathbb{N}$.
For a complete presentation $\xR$ of $\xE \cup \AC$ we have
$s_n \Dab{\xR}{\sim} t_n$ for all $n \in \mathbb{N}$, but this
demands infinitely many rules in $\xR$ as before.
\end{exa}

\subsection{General AC Completion}
\label{sec:gaccomp}

Inference systems for completion modulo an equational theory which are
not restricted to the left-linear case usually need more inference
rules than the ones already covered in this article. For general AC
completion, however, there exists a particularly simple inference
system which constitutes a special case of normalized completion
\cite{M96} and can be found in Sarah Winkler's PhD thesis
\cite[p.~109]{W13}.

\begin{defi}
The inference system \kbac is the same as \ia for the fixed
theory AC but with a modified collapse rule which allows us to
rewrite with $\RmAC$ and the following rule which replaces
the two deduction rules of \ia:
\[
\ded \quad \ds \frac{\xE,\xR}{\xE \cup \SET{s \approx t},\xR} \quad
\text{if $s \LbR \cdot \acsim \cdot \RbR t$}
\]
\end{defi}

The purpose of this section is to show how \ia can be simulated by
\kbac in the case of $\xB = \AC$. In addition to its theoretical
significance, this simulation result is also of practical importance
as it facilitates switching from \ia to \kbac in the middle of a
completion process instead of starting from scratch. We have more
to say about this at the end of this section.

Since local cliffs cannot be deduced
in \kbac, the simulation has to work with a potentially smaller set of
rewrite rules. Furthermore, during a run, the variants of rules
stemming from local cliffs may be in different states with respect to
inter-reduction (\col and \com). Given an intermediate TRS $\xR$ of a
run in \ia as well as an intermediate TRS $\xR'$ of a run in
\kbac, the invariant $\xR \subseteq {\Rab{\xR'/\AC}{+}}$ resolves both
of the aforementioned problems. The main motivation behind this
invariant is the avoidance of \com and \col in the \kbac run.

\begin{lem}
\label{lem:gacsim}
If $(\xE_1,\xR_1) \iA (\xE_2, \xR_2)$ and
\smash{$\xR_1 \subseteq {\Rab{\xR'_1/\AC}{+}}$} then there exists a
\textup{TRS}
$\xR_2'$ such that \smash{$(\xE_1,\xR'_1) \seqa{\kbac}{*} (\xE_2,\xR'_2)$}
and $\xR_2 \subseteq {\Rab{\xR'_2/\AC}{+}}$.
\end{lem}

\begin{proof}
Let $>$ be a fixed AC-compatible reduction order which is used in
both \ia and \kbac. Suppose $(\xE_1,\xR_1) \iA (\xE_2,\xR_2)$
and \smash{$\xR_1 \subseteq {\Rab{\xR'_1/\AC}{+}}$}. We perform a case
analysis on the inference rule applied in 
$(\xE_1,\xR_1) \iA (\xE_2,\xR_2)$. The only interesting cases are
when \ded, \sip, \com or \col are applied.
\begin{itemize}[label=$\triangleright$]
\item
If \ded is applied, we further distinguish whether it was
applied to a local peak or cliff. In the case of a local cliff, we
have $\xE_1 = \xE_2$ and $\xR_2 = \xR_1 \cup \SET{\ell \to r}$ with
$\ell \Rb{\xR_1/\AC} r$. From $\ell \Rb{\xR_1/\AC} r$ and
$\xR_1 \subseteq {\Rab{\xR'_1/\AC}{+}}$ we obtain
$\ell \Rab{\xR'_1/\AC}{+} r$. Thus,
$\xR_2 \subseteq {\Rab{\xR'_1/\AC}{+}}$ holds. As
\smash{$(\xE_1,\xR'_1) \seqa{\kbac}{0} (\xE_2,\xR'_1)$} is trivial, the
claim follows. In the case of a local peak, we have
$\xR_1 = \xR_2$ and $\xE_2 = \xE_1 \cup \SET{t \approx u}$ with
$t \Lb{\xR_1} s \Rb{\xR_1} u$. Since
$\xR_1 \subseteq {\Rab{\xR'_1/\AC}{+}}$ holds, we have
\[
t \,\Lab{\xR_1'/\AC}{*}\, v\, \Lb{\xR_1'} \cdot \acsim^{}\, s
\,\acsim^{} \cdot \Rb{\xR_1'}\, w \,\Rab{\xR_1'/\AC}{*} u
\]
for some $v$ and $w$. By performing \ded and \sip steps
\[
(\xE_1,\xR'_1) \,\seq{\kbac}\, (\xE_1 \cup \SET{v \approx w},\xR'_1)
\,\seqa{\kbac}{*}\, (\xE_1 \cup \SET{t \approx u},\xR'_1) \,=\,
(\xE_2,\xR'_1)
\]
is obtained. As $\xR_1 = \xR_2$, the inclusion
$\xR_2 \subseteq {\Rab{\xR'_1/\AC}{+}}$ is trivial. Hence, the
claim holds. \smallskip
\item
If \sip is applied, we have $\xR_1 = \xR_2$,
$\xE_1 = \xE_0 \cup \SET{s \approx t}$ and
$\xE_2 = \xE_0 \cup \SET{s' \approx t'}$ with $s \Rab{\xR_1}{=} s'$
and $t \Rab{\xR_1}{=} t'$. By
$\xR_1 \subseteq {\Rab{\xR'_1/\AC}{+}}$ we have
$s \Rab{\xR'_1/\AC}{*} s'$ and $t \Rab{\xR'_1/\AC}{*} t'$.
Therefore, performing \sip, we obtain
\smash{$(\xE_1,\xR'_1) \seqa{\kbac}{*} (\xE_2,\xR'_1)$}. As
$\xR_1 = \xR_2$, the inclusion
$\xR_2 \subseteq {\Rab{\xR'_1/\AC}{+}}$ is trivial. \smallskip
\item
If \com is applied, we have $\xE_1 = \xE_2$,
$\xR_1 = \xR_0 \cup \SET{\ell^{\vphantom{1}} \to r}$ and
$\xR_2 = \xR_0 \cup \SET{\ell \to r'}$ with $r \Rb{\xR_0/\AC} r'$.
We have $(\xE_1,\xR'_1) \seqa{\kbac}{0} (\xE_2,\xR'_1)$. Since
the inclusions
\smash{$\xR_0 \subseteq \xR_1 \subseteq {\Rab{\xR'_1/\AC}{+}}$} yield
\smash{$\ell \Rab{\xR'_1/\AC}{+} r \Rab{\xR'_1/\AC}{+} r'$}, we obtain
$\xR_2 \subseteq {\Rab{\xR'_1/\AC}{+}}$. \smallskip
\item
If \col is applied, we have
$\xE_2 = \xE_1 \cup \SET{\ell' \approx r}$ and
$\xR_1 = \xR_2 \uplus \SET{\ell \to r}$ with
$\ell \Rb{\xR_2} \ell'$. By
$\xR_2 \subseteq \xR_1 \subseteq {\Rab{\xR'_1/\AC}{+}}$ we have
\[
\ell' \,\Lab{\xR'_1/\AC}{*}\, t \,\Lab{\xR'_1}{} \cdot \acsim^{}\, \ell
\,\acsim^{} \cdot \Rab{\xR'_1}{}\, u \,\Rab{\xR'_1/\AC}{*}\, r
\]
for some $t$ and $u$. Performing \ded and \sip, we obtain:
\[
(\xE_1,\xR'_1) \,\seq{\kbac}\, (\xE_1 \cup \SET{t \approx u},\xR'_1)
\,\seqa{\kbac}{*}\, (\xE_1 \cup \SET{\ell' \approx r},\xR'_1) \,=\,
(\xE_2,\xR'_1)
\]
By $\xR_2 \subseteq \xR_1 \subseteq {\Rab{\xR'_1/\AC}{+}}$ the
claim is concluded. \qedhere
\end{itemize}
\end{proof}

\begin{thm}
\label{thm:gacsim}
For every fair run $(\xE,\varnothing) \iA^* (\varnothing,\xR)$ there
exists a run $(\xE,\varnothing) \seqa{\kbac}{*} (\varnothing,\xR')$ such
that $\xR'/\AC$ is an \textup{AC}-complete presentation of $\xE$.
\end{thm}

\begin{proof}
With a straightforward induction argument, we obtain the run
\smash{$(\xE,\varnothing) \seqa{\kbac}{*} (\varnothing,\xR')$} as well as
\smash{${\xR} \subseteq {\Rab{\xR'/\AC}{+}}$} ($\ast$) from
\lemref{gacsim}.
Furthermore, AC termination of $\xR'$ and
${\LRab{\xE \cUp \AC}{*}} = {\LRab{\xR' \cUp \AC}{*}}$ ($\ast\ast$)
are easy consequences of the definition of \kbac. AC-completeness
of $\xR$ follows from fairness of the run in \ia and
\thmref{correctness}. For the Church--Rosser modulo AC property of
$\xR'/\AC$, consider a conversion $s \LRab{\xR' \cUp \AC}{*} t$.
From ($\ast\ast$) we obtain $s \LRab{\xE \cUp \AC}{*}$ and
therefore $s \Rab{\xR}{*} \cdot \acsim^{} \cdot \Lab{\xR}{*} t$ by
the fact that $\xR$ is an AC-complete presentation of
$\xE$. Finally, ($\ast$) yields
$s \Rab{\xR'/\AC}{*} \cdot \sim^{}_{\AC\hphantom{/}} \cdot
\Lab{\xR'/\AC}{*} t$ as
desired. Thus, $\xR'/\AC$ is an AC-complete presentation of $\xE$.
\end{proof}

In addition to the result of the previous theorem, the proof of
\lemref{gacsim} provides a procedure to construct a \kbac run which
``corresponds'' to a given \ia run. In particular, this means that
it is possible to switch from \ia to \kbac at any point while
performing AC completion. This is of practical relevance: Assume that
AC completion is started with \ia in order to avoid AC
unification. If \ia gets stuck due to simplified equations which
are not orientable into a left-linear rule or it seems to be the case
that the procedure diverges due to the problem described in
\exaref{ia_diverge}, starting from scratch with \kbac is not
necessary. We conclude the section by illustrating the practical
relevance of the simulation result with an example.

\begin{exa}
\label{exa:agroups}
Consider the ES $\xE$ for abelian groups consisting of the equations
\begin{align*}
\m{e} \cdot x &\cE x & x^{-} \cdot x &\cE \m{e}
\end{align*}
where $\cdot$ is an AC symbol. Note that the well-known completion
run for non-abelian group theory is also a run in \ia: Critical
pairs with respect to the associativity axiom are deducible via
local cliffs, non-left-linear intermediate rules are allowed and all
(intermediate) rules are orientable with e.g.\ AC-KBO. Hence, we obtain
the TRS $\xR'$ consisting of the rules
\begin{align*}
\m{e} \cdot x &\nR{1} x &
x^{-} \cdot (x \cdot y) &\nR{4} y &
x \cdot x^{-} &\nR{7} \m{e} \\
x^{-} \cdot x &\snR{2} \m{e} &
(x \cdot y)^{-} &\snR{5} y^{-} \cdot x^{-} &
\m{e}^{-} &\snR{8} \m{e} \\
x^{--} &\snR{3} x &
x \cdot \m{e} &\snR{6} x &
x \cdot (x^{-} \cdot y) &\snR{9} y
\end{align*}
and switch to \kbac where we can collapse the redundant rules
$4$, $6$, $7$ and $9$. A final joinability check of all AC critical
pairs reveals that the resulting TRS $\xR$ is an AC-complete presentation
of abelian groups. Hence, the simulation result allows us to make progress
with \ia even when it is doomed to fail. In particular, critical
pairs between rules whose left-hand sides do not contain AC symbols do
not need to be recomputed.
\end{exa}

\section{Implementation}
\label{sec:implementation}

The command-line tool \accompll implements AC completion for
left-linear TRSs based on the inference system \ia~(\defref{ia}). It
uses external termination tools instead of a fixed AC-compatible
reduction order and is written in the programming language Haskell.
The source code of the tool \accompll is available on
GitHub\footnote{\url{https://github.com/niedjoh/accompll}}. As input,
the tool expects a file in the
WST\footnote{\url{https://www.lri.fr/~marche/tpdb/format.html}} format
describing the equational theory on which left-linear AC completion
should be performed. The user can choose whether $\RbR$,
$\RcAC$ or $\RmAC$ is used for rewriting in the inference
rules \sip and \com. Furthermore, the generation of critical pairs can
be restricted to the primality criterion. Note that most of
these options are facilitated by the theoretical results in the previous
sections. A discussion of $\RcAC$ will follow in \secref{impldet}.

Another feature is the validity problem solving mode which solves a
given instance of the validity problem for an equational theory $\xE$
upon successful completion of $\xE$. This mode can be triggered by
supplying a concrete equation $s \approx t$ as a command line argument
in addition to the file describing $\xE$.

In the tool \accompll, external termination tools do much of the heavy
lifting. In particular, the user can supply the executable of an
arbitrary termination tool as long as the output starts with
\texttt{YES}, \texttt{MAYBE}, \texttt{NO} or \texttt{TIMEOUT} (all
other cases are treated as an error). The input format for the
termination tool can be set by a command line argument. The available
options are the WST format as well as the XML format of the Nagoya
Termination Tool \cite{YKS14}.\footnote{%
\url{https://www.trs.cm.is.nagoya-u.ac.jp/NaTT/natt-xml.html}}

Since starting a new process for every call of the termination tool
causes a lot of operating system overhead, the tool supports an
interactive mode which allows it to communicate with a single process
of the termination tool in a dialogue style. Here, the only constraint
for the termination tool is that it accepts a sequence of termination
problems separated by the keyword \texttt{(RUN)}. This is currently
only implemented in an experimental version of Tyrolean Termination
Tool 2 (\TTTT) \cite{KSZM09}, but we hope that more termination tools
will follow as this approach has a positive effect on the runtime of
completion with termination tools while demanding comparatively little
implementation effort. The remainder of this section discusses
important aspects and properties of the implementation.

\subsection{Termination Tools}
\label{sec:tt}

The reduction order is a critical input parameter of a completion
procedure. Finding an appropriate order can be very challenging as the
equations which are generated during a completion run are usually not
known in advance. Hence, the nature of this problem is different from
the standard termination problem where the input is one specific TRS.

While it is well-known that termination is an undecidable property of
TRSs, a number of termination tools have been developed which can
solve the termination problem automatically in many practical cases.
The usage of termination tools in completion has been pioneered by
Wehrman et al.~\cite{WSW06}. Intuitively, the difficulty of finding an
appropriate reduction order is replaced by a sequence of calls to a
termination tool. This allows us to define a completion procedure
which does not depend on a reduction order as input. An important
ingredient for completion with termination tools is a
\emph{constraint system} which keeps track of all rules which have been
produced during the run. Checking termination for the constraint system as
opposed to the current TRS guarantees that there exists a single
reduction order which can be used for the whole run. This is important
as correctness is lost when the reduction order is changed during a
completion run \cite{S94}. The next definition extends the inference
system \ia by a constraint system.

\begin{defi}
\label{def:iatt}
The inference system \iatt transforms triples consisting of an ES
$\xE$ and TRSs $\xR$, $\xC$ over the common signature $\xF$. Except
for \ori, the inference rules are trivial extensions of the rules in
\iatt where the constraint system $\xC$ is not changed.
\[
\begin{array}{@{}r@{\quad}c@{\quad}l@{}}
& \ds \frac{\xE \uplus \SET{s \approx t}, \xR, \xC}
{\xE, \xR \cup \SET{s \R t}, \xC \cup \SET{s \R t}}
& \text{if $\xC \cup \SET{s \R t}$ is $\xB$-terminating} \\
\ori & & \\
& \ds \frac{\xE \uplus \SET{s \approx t}, \xR, \xC}
{\xE, \xR \cup \SET{t \R s}, \xC \cup \SET{t \R s}}
& \text{if $\xC \cup \SET{t \R s}$ is $\xB$-terminating}
\end{array}
\]
\end{defi}

Note that $\xR$ is subject to insertion
as well as removal of rules while the constraint
system $\xC$ just acts as an accumulator of all orientations for the
termination check.
As usual, \smash{$(\xE,\xR,\xC) \seq{\iatt} (\xE',\xR',\xC')$} denotes a
step in the inference system \iatt and a sequence of steps starting
with $(\xE,\varnothing,\varnothing)$ constitutes a run for $\xE$. We
will now prove that \iatt is sound and complete with respect to \ia.

\begin{lem}
\label{lem:ttsoundness}
For every run
\[
\xE_0,\xR_0,\xC_0 \,\seq{\iatt}\, \xE_1,\xR_1,\xC_1
\,\seq{\iatt}\, \cdots \,\seq{\iatt}\, \xE_n,\xR_n,\xC_n
\]
there exists an equivalent run
\[
\xE_0,\xR_0 \,\seq{\ia}\, \xE_1,\xR_1 \,\seq{\ia}\, \cdots \,\seq{\ia}\,
\xE_n,\xR_n
\]
which uses the reduction order $\Rab{\xC_n/\xB}{+}$.
\end{lem}

\begin{proof}
First of all, note that \smash{$\Rab{\xC_n/\xB}{+}$} is a $\xB$-compatible
reduction order as it is transitive, closed under contexts and
substitutions and the $\xB$-termination of $\xC_n$ is established by
induction on the length of the run in \iatt. Since the constraint
system is only altered and used in the \ori rule, all other steps
are valid steps in \ia by just removing the constraint system.
Since $\xC_n$ includes all rules which are oriented during the run,
the corresponding \ori steps in \ia also succeed.
\end{proof}

\begin{lem}
\label{lem:ttcompleteness}
For every run
\[
\xE_0,\xR_0 \,\seq{\ia}\, \xE_1,\xR_1 \,\seq{\ia}\, \cdots \,\seq{\ia}\,
\xE_n,\xR_n
\]
using the $\xB$-compatible reduction order $>$ there exists an equivalent
run
\[
\xE_0,\xR_0,\xC_0 \,\seq{\iatt}\, \xE_1,\xR_1,\xC_1 \,\seq{\iatt}\,
\cdots \,\seq{\iatt}\, \xE_n,\xR_n,\xC_n
\]
where ${\Rab{\xC_n/\xB}{+}} \subseteq {>}$.
\end{lem}

\begin{proof}
We perform induction on $n$. For $n = 0$ we just translate
$(\xE_0,\varnothing)$ to $(\xE_0,\varnothing,\varnothing)$. For
$n > 0$ we obtain \smash{$(\xE_0,\xR_0,\xC_0) \seqa{\iatt}{*}
(\xE_{n-1},\xR_{n-1},\xC_{n-1})$} and
\smash{${\Rab{\xC_{n-1}/\xB}{+}} \subseteq {>}$} from the induction
hypothesis. If the step
$(\xE_{n-1},\xR_{n-1}) \seq{\ia} (\xE_n,\xR_n)$ is not an
application of \ori we have
\smash{$(\xE_{n-1},\xR_{n-1},\xC_{n-1}) \seq{\iatt}
(\xE_n,\xR_n,\xC_n)$} where $\xC_{n-1} = \xC_n$ and therefore
\smash{${\Rab{\xC_n/\xB}{+}} \subseteq {>}$} as the constraint system is
just carried along without any further restrictions in all rules
except \ori. Now assume that the step is an application of \ori on
the equation $s \approx t$. We have
$\xE_n = \xE_{n-1} \setminus \SET{s \approx t}$ and
$\xR_n = \xR_{n-1} \cup \SET{v \R w}$ where $v > w$ and
$\SET{v,w} = \SET{s,t}$. Let $\xC_n = \xC_{n-1} \cup \SET{v \R w}$. Thus,
\smash{$(\xE_{n-1},\xR_{n-1},\xC_{n-1}) \seq{\iatt}
(\xE_n,\xR_n,\xC_n)$} by an application of \ori and we obtain
\smash{$(\xE_0,\xR_0,\xC_0) \seqa{\iatt}{*} (\xE_n,\xR_n,\xC_n)$} as
well as ${\Rab{\xC_n/\xB}{+}} \subseteq {>}$ from the induction
hypothesis together with $v > w$ and the definition of $\xC_n$.
\end{proof}
\noindent 
Note that the presented translation between runs of \ia and \iatt allows
us to speak about the fairness of runs in \iatt since this simply
means that the corresponding run in \ia is fair.

\lemref{ttsoundness} shows that the usage of termination tools is a
sound extension which allows us to perform completion without
constructing an appropriate reduction order beforehand. Furthermore,
the usage of termination tools does not affect the applicability of
our completion procedure due to the previous completeness result
(\lemref{ttcompleteness}). However, implementing completion with
termination tools such that it does not affect the applicability in
practice is a highly nontrivial task. The reason for this is that
without a concrete reduction order, both versions of the \ori rule may
be applicable. This yields a potentially huge search space as
orienting a rule in the ``wrong'' direction may cause completion to
fail or diverge. In \cite{WSW06} this problem is solved by traversing
the search space which is a binary tree with some best-first strategy
based on a cost function which takes the number of equations, rules
and critical pairs into account.

The state of the art for solving this problem efficiently is
\emph{multi-completion with termination tools} due to
Winkler et al.~\cite{WSMK13}. This method traverses the whole search
space while being as efficient as possible by sharing computations between
different nodes in the binary tree which represents the search
space. The method is implemented in the tool \mkbtt. As the
implementation of this approach is a major effort, we adopt the
strategy used by the automatic mode of the Knuth--Bendix Completion
Visualizer (\kbcv) \cite{SZ12}. Instead of traversing the whole search
space, \kbcv runs two threads in parallel where one thread prefers to
orient equations from left to right while the other one prefers to
orient from right to left. If one of the threads finishes
successfully, the corresponding result is reported. Completion fails
if both threads fail. Needless to say, this compromises completeness,
but it is a trade-off which works well in many practical cases. In
particular, \kbcv can also complete systems which \mkbtt cannot within
a given time constraint \cite{SZ12}.

\subsection{Strategy}
\label{sec:strategy}

The presentation of \iatt as an inference system allows
implementations to use the given rules in any order as long as the
produced run is fair. This section is about the strategy for applying
rules of \iatt which is employed by the tool \accompll. From now on,
we specialize to the equational theory AC as \accompll was developed
for this specific case.

The used strategy is based on Huet's completion procedure
\cite[Section 7.4]{BN98}. An important property of this procedure is
that the rules \sip, \del, \com and \col are applied eagerly in order
to keep the intermediate ESs and TRSs as small as possible. Usually,
this has a positive effect on the runtime of the completion
process. Unlike Huet's completion procedure, in the employed strategy
the orientation of a rule is always directly followed by the
computation of all possible critical pairs which involve this
rule. Furthermore, \ori is only applied once per iteration to the
smallest equation with respect to the term size. This modified version
of Huet's completion procedure is implemented e.g.\ in \kbcv and the
flow chart depicted in \figref{flowchart} is based on
\cite{SZ12}. However, in contrast to standard completion, we have to
add some of the critical pairs as rules in \iatt. Hence, we apply
\ded on the selected equation before \com and \col are applied
exhaustively. Moreover, we need to keep track of a separate set of
pending rules ($\xP$) as the eager recursive computation of critical
pairs between rules and AC axioms might unnecessarily lead to
non-termination of the completion procedure.

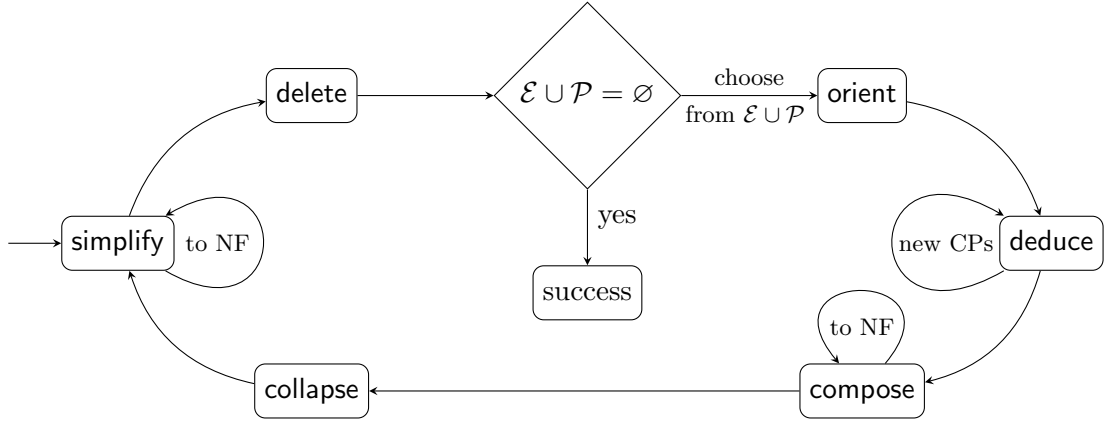
\begin{figure}[t]
\centering
\begin{tikzpicture}[node distance=1cm,>=stealth]
\node (1) [block] {\sip};
\node (9) [right=of 1] {\phantom{\sip}};
\node (2) [block,above=of 9,yshift=8] {\del};
\node (7) [block,below=of 9,yshift=-8] {\col};
\node (3) [decision,right=18mm of 2]
 {\MC[15]{$\xE \cup \xP = \varnothing$}};
\node (8) [finish,below=of 3] {success};
\node (4) [block,right=18mm of 3] {\ori};
\node (10) [below=of 4,yshift=-8] {\phantom{\sip}};
\node (5) [block,right=of 10] {\ded};
\node (6) [block,below=of 10,yshift=-8] {\com};
\node (11) [left=7mm of 1] {};
\draw [->] (1) to [bend left] (2);
\draw [->] (2) -- (3);
\draw [->] (3) -- (4) node [pos=.5,above] {\footnotesize choose}
node [pos=.47,below] {\footnotesize from $\xE \cup \xP$};
\draw [->] (4) to [bend left] (5);
\draw [->] (5) to [bend left] (6);
\draw [->] (6) -- (7);
\draw [->] (7) to [bend left] (1);
\draw [->] (3) -- (8) node [pos=.4,right] {yes};
\draw [->] (11) -- (1);
\draw [->] (1) to [reflexive right] (1)
node [right=7.6mm] {\footnotesize to NF};
\draw [->] (5) to [reflexive left] (5)
node [left=6.3mm] {\footnotesize new CPs};
\draw [->] (6) to [reflexive above,in=130,out=50] (6)
node [above=6mm] {\footnotesize to NF};
\end{tikzpicture}
\caption{Flow chart for \accompll's completion procedure.}
\label{fig:flowchart}
\end{figure}

We will now describe the flow chart depicted in \figref{flowchart} in
detail. As already mentioned, \accompll also keeps track of a list of
pending rules $\xP$. Intuitively, a pending rule is like a normal
rule with the exception that critical pairs involving it have not yet
been computed and it is not used to collapse other rules. Hence,
\accompll works on a quadruple of an ES and three TRSs
$(\xE,\xP,\xR,\xC)$ which are processed by one main loop. First of
all, \sip is applied exhaustively such that both sides of every
equation are normal forms with respect to $\xE$ and $\xP$. This
preliminary step may be already enough to join some equations, so
after that every AC equivalent equation is removed with \del. If no
equations or pending rules are left, we are done. Otherwise, an
equation $\ell \approx r$ or pending rule $\ell \R r$ which is minimal
with respect to $|\ell| + |r|$ where $|\cdot|$ denotes the size of a
term is selected. If it is a rule, \ori is just the identity
function. Otherwise, \ori is applied with the orientation preference
of the given thread, i.e., it first orients the equation in the
preferred direction and only tries the other option in case of
failure. An important feature of the implementation of \ori is that it
can be postponed, i.e., if one equation is not orientable in either
direction, the next one is tried. If \ori never succeeds, the thread
terminates in a failure state.

In any case, a successful application of \ori yields some rule
$\rho$. Next, \ded is used to produce the sets $\cp{\SET{\rho}}$ and
$\cppm{\xR,\SET{\rho}}$ which are added to $\xE$ as well as
$\cppm{\AC^\pm,\SET{\rho}}$ which is added to $\xP$.
The computation of critical pairs is of course restricted to the
primality criterion if the corresponding option is set.

After that, all
rules in $\xP \cup \xR$ are exhaustively composed such that their
right-hand sides are normal forms with respect to $\xP \cup \xR$.
Finally, \col is applied exhaustively in $\xR$ with respect to $\rho$
and in $\xP$ with respect to $\SET{\rho} \cup \xR$. Here, we
distinguish between $\xR$ and $\xP$ because at this point, the
left-hand sides of $\xR$ are known to be normal forms of the remaining
rules of $\xR$ while this can not be assured for the deduced rules
$\xP$. Note that we do not collapse $\xP$ with respect to $\xP$ as it
would require checking rules for equality.

In order to produce a left-linear system, \ori is only applied if the
left-hand side of the resulting rule is linear which makes all
intermediate TRSs left-linear. Fairness only demands left-linearity
for the resulting TRS, but improving this aspect was not considered as
it is unclear based on which criteria non-left-linear rules should be
turned back into equations again. Moreover, keeping each intermediate
TRS left-linear makes intermediate rules stemming from critical pairs
left-linear by definition: Since the AC axioms are linear, overlaps
between the AC axioms and a left-linear TRS always produce linear
terms. Thus, the left-linearity of pending rules does not have to be
checked as rewriting with linear rules preserves linearity. We are now
ready to prove the main result of this section.

\begin{thm}
\label{theorem:accompll}
The \textup{TRS}s produced by the tool \accompll are \textup{AC}-canonical
presentations of the \textup{ES} provided as input.
\end{thm}

\begin{proof}
Suppose that given an ES $\xE$, \accompll outputs a TRS $\xR_n$ and
consider the thread which produced this result in $n$ iterations of
the main loop as depicted in the flow chart (\figref{flowchart}).
In particular, let $\xE_i$, $\xP_i$, $\xR_i$ and $\xC_i$ denote the
respective values of $\xE$, $\xP$, $\xR$ and $\xC$ at the decision
node in the flow chart in the $i$-th iteration. Note that
$\xE_0 = \xE$, $\xP_0 = \xR_0 = \xC_0 = \varnothing$ and
$\xE_n = \xP_n = \varnothing$. It is immediate from the flow chart
and its textual description that
\[
(\xE_0,\xR_0 \cup \xP_0,\xC_0)
\,\seqa{\iatt}{*}\, (\xE_1,\xR_1 \cup \xP_1,\xC_1)
\,\seqa{\iatt}{*}\, \cdots \,\seqa{\iatt}{*}\,
(\xE_n,\xR_n \cup \xP_n,\xC_n)
\]
is a run for $\xE$. Furthermore, all prime critical pairs of $\xR_n$
have been considered as an intermediate equation and all prime
critical pairs between $\xR_n$ and $\AC^\pm$ have been considered
as an intermediate rule. The fact that $\xR_n$ is an AC-complete
presentation of $\xE$ now follows from fairness (\defref{fairness})
together with \lemref{ttsoundness} and \thmref{correctness}.
Finally, a straightforward induction argument shows that the
statements
\begin{enumerate}
\item
$\ell \in \nf{\xR_i}$ for every ${\ell \R r} \in \xR_i \cup \xP_i$ and
\smallskip
\item
$r \in \nf{\xR_i \cup \xP_i}$ for every ${\ell \R r} \in \xR_i \cup \xP_i$
\end{enumerate}
hold for $0 \leq i \leq n$. Together, (1) and (2) show that $\xR_n$
is also AC-canonical: Just like in the proof of \lemref{ibcanonical}
we conclude that if each rule of an AC-complete system is
right-reduced, it is also right-AC-reduced.
\end{proof}

\subsection{Implementation Details}
\label{sec:impldet}

The implementation of the tool \accompll is based on the Haskell
\texttt{term-rewriting} library \cite{FAS13} which takes care of most
of our needs regarding term rewriting except rewriting modulo AC,
prime critical pairs and the computation of normal forms. The
following paragraphs describe selected implementation details.

\paragraph{Rewriting to normal forms} A naive implementation of the
computation of normal forms is highly inefficient as in general, the
whole term has to be traversed for every rewrite step. As a trade-off
between the time-consuming implementation of sophisticated term
indexing techniques and the naive implementation we employ a bottom-up
construction of normal forms by innermost rewriting of marked terms
where normal form positions are remembered.

\paragraph{AC equivalence of terms} Checking whether two terms are
equivalent modulo AC is needed for the implementation of the inference
rule \ded. To this end, an equation $s \approx t$ is transformed into
a canonical form $s' \approx t'$ where nested applications of AC
symbols are flattened out to just one $n$-ary application for
arbitrary $n$ and their arguments are ordered with respect some total
order on terms. Then, $s \acsim t$ if and only if $s' = t'$. As an
example, the terms $\m{f}(x,\m{f}(y,\m{g}(\m{f}(z,\m{a}))))$ and
$\m{f}(\m{f}(y,x),\m{g}(\m{f}(\m{a},z)))$ with $\m{f} \in \AC$
are AC equivalent because they have the same canonical form
$\m{f}(\m{g}(\m{f}(\m{a},z)),x,y)$ with respect to
the lexicographic order on the representation of terms as strings of
ASCII symbols.

\paragraph{Normal forms in rewriting modulo AC}
As already mentioned in \secref{preliminaries}, a direct implementation of
$\RmAC$ cannot be efficient. In the following, we will prove that it
suffices to implement the relation $\RcAC$ from Peterson and
Stickel~\cite{PS81} which is easier to compute but relies on AC
matching. Although there are more efficient implementations, we used
the AC matching algorithm due to Contejean \cite{C04} as it has been
certified. The result we are about to prove
(\corref{slash_and_comma_star}) is a special case of more general
results for conditional rewriting modulo an equational theory by
Meseguer~\cite[Corollary 3]{M17}. However, to the best of our
knowledge, the literature does not contain a direct proof of the
required result despite its importance for effective implementations
also in the unconditional case. Therefore, in the following, we give a
detailed and direct proof of the result.

We start with recalling the definition of the concept of
extended rules due to~\cite{PS81}.
Let $\xR$ be a TRS and let $f(u,v) \R r$ be a rule in $\xR$ where
$f$ is an AC function symbol. The rule $f(f(u,v),x) \R f(r,x)$ with
$x \in {\xV \setminus \var{f(u,v)}}$ is an \emph{extension} of
$f(u,v) \R r$. The TRS consisting of $\xR$ together with all
extensions of rules in $\xR$ is denoted by $\xR^e$. We will now
prove that the relations $\RmAC[*] \cdot \acsim^{}$ and
$\RecAC[*] \cdot \acsim^{}$ coincide. To this end, we show 
${{\acsim \cdot \RbR}} \subseteq {\RecAC \cdot \acsim}$, 
which is called \emph{strict coherence} in~\cite{M17}.

\begin{lem}
\label{lem:local_coherence}
${\Lab{\AC}{} \cdot \RecAC} \subseteq {\RecAC \cdot \Lab{\AC}{=}}$
\end{lem}

\begin{proof}
Suppose $t \Lab{\AC}{p} s \RecAC[q] u$. We show 
$t \RecAC \cdot \Lab{\AC}{=} u$ by induction on $s$. Let $\ell_1 \to r_1$
and $\ell_2 \to r_2$ be the rules employed in the left and right
steps.
We distinguish four cases.
We use the facts that 
$\acsim \cdot \Rab{\xR,\h\AC}{\epsilon}$ and
$\Rab{\xR,\h\AC}{\epsilon}$ coincide 
and
$\Lb{\AC}$ is contained in $\acsim$.
\begin{enumerate}
\item
If $q = \epsilon$ then $t \RecAC u$. Thus, the claim holds. \smallskip
\item
If $p = \epsilon$ and $q \notin \posf{\ell_1}$ then $s = \ell_1\sigma$
and
$t = r_1\sigma$ for some substitution $\sigma$ and there exists a variable
position $p'$ in $\ell_1$ with $p' \leqslant q$.
We have $\ell_1\sigma|_{p'} \RecAC u|_{p'}$.
Define the substitution $\tau$ as follows:
\[
\tau(x) \,=\, \begin{cases}
u|_{p'} & \text{if $x = \ell_1|_{p'}$} \\
\sigma(x) & \text{otherwise}
\end{cases}
\]
As $r_1$ is linear and $\var{\ell_1} = \var{r_1}$,
the variable $\ell_1|_{p'}$ has one occurrence in $r_1$.
Hence, we obtain
$t = r_1\sigma \Rb{\xR^e,\AC} r_1\tau \Lb{\AC} \ell_1\tau = u$.
\smallskip
\item
If $p = ip'$ and $q = jq'$ with $i, j \in \mathbb{N}$ then we further
distinguish two subcases. If $i = j$ then 
$t|_i \Lb{\AC} s|_i \RecAC u|_i$ and we obtain
$t \RecAC \cdot \Lab{\AC}{=} u$ with help of the induction hypothesis.
If $i \neq j$ then we clearly have
$t \RecAC[q] \cdot \Lab{\AC}{p} u$. 
\smallskip
\item
In the final case, $p = \epsilon$ and $q \in \posf{\ell_1}$ with
$q \neq \epsilon$. Since $\posf{\ell_1} \subseteq \SET{\epsilon,1}$, the
rule $\ell_1 \to r_1$ is an associativity rule
$f(f(x,y),z) \R f(x,f(y,z))$ and $q = 1$. So $s|_1 \acsim \ell_2\sigma$
and $u = f(r_2\sigma,s|_2)$. If $\ell_2 \to r_2 \in \xR$ then the
corresponding extended rule results in
$f(\ell_2\sigma,s|_2) \Rab{\xR^e}{\epsilon} f(r_2\sigma,s|_2)$
and thus
$t \acsim s \acsim f(\ell_2\sigma,s|_2) \Rab{\xR^e}{\epsilon} u$,
which entails $t \RecAC u$. Otherwise, $\ell_2 \to r_2$ is the extension
$f(\ell,x') \to f(r,x')$ of some rule $\ell \to r \in \xR$ with
$x' \notin \var{\ell}$. Define the substitution $\tau$ as follows:
\[
\tau(x) \,=\, \begin{cases}
f(x'\sigma,s|_2) & \text{if $x = x'$} \\
\sigma(x) & \text{otherwise}
\end{cases}
\]
Since $t \acsim s \acsim f(\ell_2\sigma,s|_2) \acsim
f(\ell\sigma,f(x'\sigma,s|_2)) = \ell_2\tau$, we obtain
$t \RecAC[\epsilon] r_2\tau$. Moreover, 
$r_2\tau = f(r\sigma,f(x'\sigma,s|_2)) \Lb{\AC}
f(f(r\sigma,x'\sigma),s|_2)) = u$.
\qedhere
\end{enumerate}
\end{proof}

\begin{lem}
\label{lem:strict_coherence}
${{\acsim \cdot \RecAC}} \subseteq {\RecAC \cdot \acsim}$
\end{lem}

\begin{proof}
\lemref{local_coherence} generalizes to
${{\Lab{\AC}{*} \cdot \RecAC}} \subseteq {\RecAC \cdot \Lab{\AC}{*}}$
by a straightforward induction argument. From
$\AC \subseteq {\Lab{\AC}{*}}$ we obtain
${\acsim} \subseteq {\Lab{\AC}{*}}$ and hence the claim follows.
\end{proof}

\begin{thm}
\label{thm:slash_and_comma_equiv}
The relations $\RmAC$ and $\RecAC \cdot \acsim$ coincide.
\end{thm}

\begin{proof}
Using \lemref{strict_coherence} we obtain
\[
{\RmAC} \,=\, {\acsim \cdot \RbR \cdot \acsim} \,\subseteq\,
{\RecAC \cdot \acsim \cdot \acsim} \,=\, {\RecAC \cdot \acsim}
\]
The other direction follows from the definition of $\RcAC$ as well as the
fact that ${\RbR} = {\Rb{\xR^e}}$:
\[
{\RecAC} \,\subseteq\, {\acsim \cdot \Rb{\xR^e}} \,=\,
{\acsim \cdot \RbR} \,\subseteq\, {\RmAC} \qedhere
\]
\end{proof}

\begin{cor}
\label{cor:slash_and_comma_star}
The relations $\RmAC[*] \cdot \acsim^{}$ and $\RecAC[*] \cdot \acsim^{}$
coincide.
\end{cor}

\begin{proof}
It is sufficient to show 
${\RmAC[n] \cdot \acsim^{}} = {\RecAC[n] \cdot \acsim^{}}$
for all $n \geqslant 0$. We use induction on $n$. If $n = 0$ then the
claim is trivial. If $n > 0$ then the claim is verified as follows:
\begin{align*}
{\RmAC[n] \cdot \acsim^{}}
&\,=\, {\RmAC[n-1] \cdot \acsim^{} \cdot \RmAC}
\tag{${\RmAC \cdot \acsim^{}} \,=\, {\acsim^{} \cdot \RmAC}$} \\
&\,=\, {\RecAC[n-1] \cdot \acsim^{} \cdot \RmAC}
\tag{induction hypothesis} \\
&\,=\, {\RecAC[n-1] \cdot \RmAC}
\tag{${\acsim^{} \cdot \RmAC} \,=\, {\RmAC}$} \\
&\,=\, {\RecAC[n-1] \cdot \RecAC \cdot \acsim^{}}
\tag{\thmref{slash_and_comma_equiv}} \\
&\,=\, {\RecAC[n] \cdot \acsim^{}}
\tag*{\qedhere}
\end{align*}
\end{proof}

\section{Experimental Results}
\label{sec:experimental-results}

We evaluated the performance of our tool \accompll on a problem
set containing 52 ESs. The problem set is based on the one used in
\cite{WM11} and has been extended by further examples from the
literature as well as handcrafted examples. From the 52 ESs, 10
contain equations which cannot be oriented into a left-linear
rule. Furthermore, 6 problems are ground which means that \accompll
cannot find a finite solution by \exaref{ia_diverge}. This leaves
us with 36 problems for which \accompll may find an AC complete
presentation. However, for some of these ESs, simplified equations
where both terms are non-linear may be deduced which causes
\accompll to get stuck. Furthermore, many interesting problems
exhibit the properties described in Examples~\ref{exa:ia_diverge}
and \ref{exa:ia_diverge2} and are therefore out of reach for our
method. Nevertheless, there are 7 examples in the problem set where
left-linear completion with AC axioms is preferable to general AC
completion due to significantly better performance.

The experiments were performed on an Intel Core
i7-7500U running at a clock rate of 2.7 GHz with 15.5 GiB of main
memory. Our tool \accompll was used with the termination tool \TTTT as
well as an experimental version (denoted by \TTTTe) which allows our
tool to communicate a sequence of termination problems without having
to start a new process all the time, as described in the preceding
section.

\begin{table}[t]
\setlength{\tabcolsep}{2pt}
\renewcommand{\arraystretch}{1.2} \centering
\caption{Experimental results on 52 problems (excerpt)}
\medskip
\label{tab:results}
\begin{tabular}{l*{11}{c}}
\toprule &
\multicolumn{2}{c}{\accompll (\TTTT)} & &
\multicolumn{2}{c}{\accompll (\TTTTe)} & &
\multicolumn{2}{c}{\maedmax} & &
\multicolumn{2}{c}{\mkbtt} \\
& \phantom{~~}(1)\phantom{~~} & \phantom{~~}(2)\phantom{~~} & \phantom{~~}
& \phantom{~~}(1)\phantom{~~} & \phantom{~~}(2)\phantom{~~} & \phantom{~~}
& \phantom{~~}(1)\phantom{~~} & \phantom{~~}(2)\phantom{~~} & \phantom{~~}
& \phantom{~~}(1)\phantom{~~} & \phantom{~~}(2)\phantom{~~}
\\ \midrule
\exaref{a95_ex_4_2_15b} &
0.54 & \phantom{1}4 & & 0.31 & \phantom{1}4 & & \phantom{1}0.02 & 3 & &
\phantom{1}0.23 & \phantom{1}2 \\
\exaref{ia_diverge} &
$\infty$ & & & $\infty$ & & & \phantom{1}0.01 & 3 & & \phantom{1}0.16 &
\phantom{1}3 \\
\exaref{agroups} &
$\bot$ & & & $\bot$ & & & \phantom{1}0.16 & 5 & & \phantom{1}0.16 &
\phantom{1}0 \\
\exaref{eh} &
1.60 & \phantom{1}4 & & 0.48 & \phantom{1}4 & & $\infty$ & & & 13.66 &
\phantom{1}3 \\
\exaref{list_add} &
2.62 & 14 & & 0.54 & 14 & & $\bot$ & & & 60.03 & 11 \\
\exaref{nat_amsd} &
1.99 & 15 & & 0.48 & 15 & & $\infty$ & & & \phantom{1}9.14 & 10 \\
\midrule
\textbf{problems solved} &
\multicolumn{2}{c}{\textbf{17}} & &
\multicolumn{2}{c}{\textbf{17}} & &
\multicolumn{2}{c}{\textbf{23}} & &
\multicolumn{2}{c}{\textbf{38}} \\
\bottomrule
\end{tabular}
\end{table}

\tabref{results} compares the two configurations of \accompll with the
normalized completion \cite{M96} mode of \mkbtt \cite{WM13} and the AC
completion mode of \maedmax \cite{W19} on selected examples. The tool
\mkbtt is the original implementation of multi-completion with
termination tools \cite{WSMK13}. \maedmax, on the other hand,
implements \emph{maximal completion} \cite{KH11} which makes use of
MaxSAT/MaxSMT solvers instead of termination tools in order to avoid
using concrete reduction orders as input. To the best of our
knowledge, there is no other comparable completion tool which
supports AC axioms. Since normalized completion subsumes general AC
completion, a comparison with the aforementioned modes of both systems
allows us to assess the effectiveness of \accompll with respect to the
state of the art in AC completion. Note that both normalized completion
and general AC completion use AC unification.

In \tabref{results}, columns (1) show the execution time in seconds
where $\infty$ denotes that the timeout of 60 seconds has been reached
and $\bot$ denotes failure of completion. Columns (2) state the number
of rules of the completed TRS. In \exaref{a95_ex_4_2_15b}, the
equations just have to be oriented and one additional rule has to be
added in the case of left-linear AC completion. Hence, all three
systems can handle this problem easily. Examples~\ref{exa:ia_diverge}
and \ref{exa:agroups} show the two main limitations of left-linear AC
completion: it diverges on problems which contain an AC symbol where
both arguments have ``structure'' and non-left-linear presentations
are out of reach. For general AC completion and normalized
completion, respectively, these examples pose no problem. The
remaining examples show that the absence of AC unification can make
left-linear completion more practical than general AC completion.

\begin{exa}
\label{exa:eh}
The Eckmann--Hilton argument \cite{EH62} considers the
following equational theory $\xE$
\begin{align*}
\m{0} \oplus a &\cE a & \m{1} \otimes a &\cE a &
(a \oplus b) \otimes (c \oplus d) &\cE (a \otimes c) \oplus (b \otimes d)
\\
a \oplus \m{0} &\cE a & a \otimes \m{1} &\cE a
\end{align*}
and proves $\m{0} \approx \m{1}$.\footnote{This example was
brought to our attention by Vincent van Oostrom.} Moreover, it
establishes that $\oplus$ and $\otimes$ coincide and are 
associative as well as commutative. If we assume that $\oplus$ and
$\otimes$ are AC symbols, \accompll produces the following complete
presentation $\xR$ of $\xE$ in less than a second:
\begin{align*}
b \oplus a &\cR a \otimes b & \m{1} \otimes a &\cR a \\
\m{0} &\cR \m{1} & a \otimes \m{1} &\cR a
\end{align*}
In \tabref{results}, we can see that
\mkbtt needs considerably more time to complete $\xE$ and \maedmax
even times out after 60 seconds. By inspecting $\xR$, it is
easily seen that the operations as well as the unit
elements are equivalent. Furthermore, the concise presentation of
$\xE$ as a complete TRS is facilitated by our canonicity results as
well as the implementation of inter-reduction (\col and \com). Using
the inference system \ib without inter-reduction and with the same AC
compatible reduction order, we obtain a much larger complete
but non-canonical presentation of $\xE$ which extends $\xR$:
\begin{align*}
b \oplus a &\cR a \otimes b &
\m{0} \oplus a &\cR a &
(a \oplus b) \otimes (c \oplus d) &\cR (a \otimes c) \oplus (b \otimes d)
\\
\m{0} &\cR \m{1} &
a \oplus \m{0} &\cR a &
a \otimes (b \oplus c) &\cR (\m{0} \otimes b) \oplus (a \otimes c) \\
\m{1} \otimes a &\cR a &
\m{0} \otimes a &\cR a &
(\m{0} \otimes a) \oplus b &\cR a \oplus b \\
a \otimes \m{1} &\cR a
\end{align*}
\end{exa}

\begin{exa}
\label{exa:list_add}
Consider the ES $\xE$
\begin{align*}
x + \m{0} &\cE x &
\m{append}(\m{nil},l) &\cE l \\
x + \m{s}(y) &\cE \m{s}(x + y) &
\m{append}(\m{c}(x,l_1),l_2) &\cE \m{c}(x,\m{append}(l_1,l_2)) \\
\m{add}(\m{nil},\m{nil}) &\cE \m{0} &
\m{rev}(\m{nil}) &\cE \m{nil} \\
\m{add}(\m{c}(x,l),\m{nil}) &\cE x + \m{add}(l,\m{nil}) &
\m{rev}(\m{c}(x,l)) &\cE \m{append}(\m{rev}(l),\m{c}(x,\m{nil})) \\
\m{add}(\m{nil},\m{c}(x,l)) &\cE x + \m{add}(\m{nil},l) &
\m{rev}(\m{rev}(l)) &\cE l \\
\m{add}(\m{c}(x,l_1),\m{c}(y,l_2)) &\cE (x + y) + \m{add}(l_1,l_2) &
\end{align*}
together with the AC axioms for $+$. This is an extension of
\cite[Exercise 4.2.4(b)]{A95} (the addition operation on lists) with
the standard append and reverse functions on lists.
We added the involution axiom for list reversal
in order to
generate critical pairs. Our tool \accompll produces a
complete presentation in less than a second by orienting all
equations in $\xE$ from left to right and adding the following rules:
\begin{align*}
\m{0} + x &\R x &
\m{rev}(\m{append}(l,\m{c}(x,\m{nil}))) &\R \m{c}(x,\m{rev}(l)) \\
\m{s}(x) + y &\R \m{s}(y + x) &
\end{align*}
It takes \mkbtt more than 60 seconds to solve this problem and
\maedmax terminates with an error message.
\end{exa}

\begin{exa}
\label{exa:nat_amsd}
Consider the ES $\xE$ 
\begin{align*}
\m{0} + x &\cE x &
\m{0} \times x &\cE \m{0} \\
\m{s}(x) + y &\cE \m{s}(x + y) &
\m{s}(x) \times y &\cE (x \times y) + y \\
\m{0} - x &\cE \m{0} &
(x + y) \times z &\cE (x \times z) + (y \times z) \\
x - \m{0} &\cE x &
\m{div}(\m{0},\m{s}(y)) &\cE \m{0} \\
\m{s}(x) - \m{s}(y) &\cE x - y &
\m{div}(\m{s}(x),\m{s}(y)) &\cE \m{s}(\m{div}(x - y,\m{s}(y)))
\end{align*}
together with the AC axioms for $+$ and $\times$, defining
addition, multiplication, cutoff subtraction and round-up division
on the natural numbers. Our tool \accompll produced a complete
presentation in less than a second by orienting all equations in
$\xE$ from left to right and adding the following rules:
\begin{align*}
x + \m{0} &\cR x &
x \times \m{0} &\cR \m{0} \\
x + \m{s}(y) &\cR \m{s}(y + x) &
x \times \m{s}(y) &\cR (y \times x) + x
\end{align*}
Since round-up division cannot be handled by simplification orders
\cite{D79}, this example also shows the merits of using termination
tools in completion. Note that it takes \mkbtt much longer to
complete $\xE$ (more than 9 seconds) and \maedmax times out on this
problem after 60 seconds.
\end{exa}

Despite the successes in solving the previous three examples
quickly, the severity of the limitations of left-linear AC
completion is reflected in the total number of solved problems as
shown in \tabref{results}. In particular, the problem set does not
contain an ES which is only completed by \accompll. However, given
\thmref{gacsim}, this is not unexpected. Another noteworthy but
unsurprising fact is that complete systems produced by \accompll tend
to have more rules since every rule needs different versions of
left-hand sides to facilitate rewriting without AC matching. The
complete results are available on the tool's
website.\footnote{\label{website}%
\url{http://cl-informatik.uibk.ac.at/software/accompll/}}
In addition to full details for the experiments with \TTTT as
shown in \tabref{results}, the website also contains additional
experiments with the termination tool \muterm \cite{GL20}. For
\accompll, using \muterm instead of \TTTT does not change the set
of solved problems but it usually takes more time to complete ESs.
The situation is different for \mkbtt: Using \muterm instead of
\TTTT internally, it can complete fewer systems overall (26 instead
of 38 out of 52). However, five of these systems are not completed
by \mkbtt with \TTTT and three of them are not completed by any
other tool configuration. 
We conclude with some additional comments on the results.

\begin{itemize}[label=$\triangleright$]
\item
The results are not cluttered with detailed results for the
available options regarding prime critical pairs and the concrete
rewrite relation used for \sip and \com since they did not lead to
significant runtime differences. Instead, the default options (no
prime critical pairs and the rewrite relation $\RbR$) were used
for the experiments.
\smallskip
\item
The restriction to prime critical pairs did not pay off in the
experimental results. However, it may reduce the number of critical
pairs which have to be considered even if inter-reduction is applied
eagerly (see \exaref{fewerpcp}).
Therefore, apart from the theoretical relevance of its feasibility,
we consider it an important improvement over the status quo which could
be beneficial for users of \accompll in the future.
\smallskip
\item
The possibility to choose the concrete rewrite relation
used for \sip and \com also did not lead to significant improvements
in our experiments. However, simplifying with $\RmAC$ can boost performance
since the problem can be simplified even before a rule with a fitting
left-hand side is derived.
\smallskip
\item
Due to the incompleteness of the used approach for completion
with termination tools, some equations in the problems
\texttt{A95\_ex4\_2\_4a.trs} as well as \texttt{sp.trs} had to be
reversed in order to get appropriate results. Note that this does
not distort the experimental results for left-linear AC completion
in general as the problem lies in the particular implementation of
completion with termination tools.
\end{itemize}

\section{Conclusion}
\label{sec:conclusion}

This article consolidates and extends existing work on
left-linear $\xB$-completion \cite{H80,B91,A95} by using and
adapting elegant proof techniques which have been put forward for
standard completion in \cite{HMSW19}. This approach allowed for
improvements of existing results: Huet's result could be
strengthened to prime critical pairs in \thmref{acpcp} which in turn
improved the definition of fair runs in the inference system \ia
(\defref{fairness}). Furthermore, the usage of peak-and-cliff
decreasingness instead of proof orderings simplified the proof of
the correctness result for \ia (\thmref{correctness}). The
limitation to finite runs also facilitated the removal of the
encompassment condition which allows the inference system to produce
smaller $\xB$-complete systems as we showed in
\exaref{encompassment}.
 
In addition, the relationship between the two existing inference
systems from the literature has been investigated thoroughly, its
core part being a simulation result (\thmref{simres}) which states
that any fair run in \ib can be simulated by a fair run in \ia. With
the novel simulation result, completion with \ib can be reduced to
completion with \ia in the case of finite runs.
Furthermore, the presentation of novel canonicity results
facilitates a concrete definition of minimal complete systems in our
setting. Since the inference system \ia adds critical pairs
stemming from overlaps between (intermediate) rules and equations in
$\xB$, its termination relies heavily on inter-reduction and
therefore canonicity (\exaref{ia_elementary_completion}).
The final theoretical contribution of this article is a formal
presentation of the correspondence between left-linear AC completion
and general AC completion through another simulation result
(\thmref{gacsim}).

The tool \accompll is the first implementation of left-linear AC
completion. Our novel results on canonicity define in which way the
systems produced by \accompll are minimal and unique.
Unfortunately, the experimental results show that despite the
practical advantages of avoiding AC unification and deciding
validity problems with the normal rewrite relation, left-linear AC
completion often needs infinitely many rules and therefore
diverges. This problem does not seem to be mentioned in the
literature. However, the possible speedup due to the
avoidance of AC unification reported in the experimental results as
well as the already mentioned simulation result for general AC
completion show that left-linear AC completion also has merits in
practice.

We conclude by giving some pointers for future work. First of all, the
merits of our novel simulation result for general AC completion could
be evaluated experimentally by providing an implementation. While
\accompll adopts the two-thread version of
multi-completion~\cite{SZ12}, we anticipate that left-linear AC
completion can also be effectively implemented by a variant of
maximal completion that aims to find a canonical system~\cite{SW15,H21}.
Another interesting research
direction is normalized completion for the left-linear case. If
successful, this would facilitate the treatment of important cases
such as abelian groups despite the restriction to left-linear
TRSs. Furthermore, a formalization of the established theoretical
results is desirable. To that end, the existing Isabelle/HOL
formalization from \cite{HMSW19} is a perfect starting point as some
results of this article are extensions of the results for standard
rewriting presented there.

\section*{Acknowledgment}
\noindent
We thank Jonas Sch\"opf and Fabian Mitterwallner for providing the
experimental version of \TTTT as well as the anonymous reviewers
for their suggestions and comments.

\bibliographystyle{alphaurl}
\bibliography{references}

\begin{thebibliography}{WSMK13}

\bibitem[Ave95]{A95}
J{\"u}rgen Avenhaus.
\newblock {\em Reduktionssysteme}.
\newblock Springer Berlin Heidelberg, 1995.
\newblock In German.
\newblock \href {https://doi.org/10.1007/978-3-642-79351-6}
  {\path{doi:10.1007/978-3-642-79351-6}}.

\bibitem[Bac91]{B91}
Leo Bachmair.
\newblock {\em Canonical Equational Proofs}.
\newblock Progress in Theoretical Computer Science. Birkh{\"a}user Boston,
  1991.
\newblock \href {https://doi.org/10.1007/978-1-4684-7118-2}
  {\path{doi:10.1007/978-1-4684-7118-2}}.

\bibitem[BD94]{BD94}
Leo Bachmair and Nachum Dershowitz.
\newblock Equational inference, canonical proofs, and proof orderings.
\newblock {\em Journal of the {ACM}}, 41(2):236--276, 1994.
\newblock \href {https://doi.org/10.1145/174652.174655}
  {\path{doi:10.1145/174652.174655}}.

\bibitem[BDP89]{BDP89}
Leo Bachmair, Nachum Dershowitz, and David~A. Plaisted.
\newblock Completion without failure.
\newblock In Hassan A{\"{\i}t}-Kaci and Maurice Nivat, editors, {\em Resolution
  of Equations in Algebraic Structures}, volume 2: Rewriting Techniques, pages
  1--30. Academic Press, 1989.
\newblock \href {https://doi.org/10.1016/B978-0-12-046371-8.50007-9}
  {\path{doi:10.1016/B978-0-12-046371-8.50007-9}}.

\bibitem[BL87]{BL87}
Ahlem {Ben Cherifa} and Pierre Lescanne.
\newblock Termination of rewriting systems by polynomial interpretations and
  its implementation.
\newblock {\em Science of Computer Programming}, 9(2):137--159, 1987.
\newblock \href {https://doi.org/10.1016/0167-6423(87)90030-X}
  {\path{doi:10.1016/0167-6423(87)90030-X}}.

\bibitem[BN98]{BN98}
Franz Baader and Tobias Nipkow.
\newblock {\em Term Rewriting and All That}.
\newblock Cambridge University Press, 1998.
\newblock \href {https://doi.org/10.1017/CBO9781139172752}
  {\path{doi:10.1017/CBO9781139172752}}.

\bibitem[Con04]{C04}
Evelyne Contejean.
\newblock A certified {AC} matching algorithm.
\newblock In Vincent van Oostrom, editor, {\em Proc.\ 15th International
  Conference on Rewriting Techniques and Applications}, volume 3091 of {\em
  Lecture Notes in Computer Science}, pages 70--84, 2004.
\newblock \href {https://doi.org/10.1007/978-3-540-25979-4_5}
  {\path{doi:10.1007/978-3-540-25979-4_5}}.

\bibitem[Der79]{D79}
Nachum Dershowitz.
\newblock A note on simplification orderings.
\newblock {\em Information Processing Letters}, 9(5):212--215, 1979.
\newblock \href {https://doi.org/10.1016/0020-0190(79)90071-1}
  {\path{doi:10.1016/0020-0190(79)90071-1}}.

\bibitem[Dev91]{D91}
Herv{\'e} Devie.
\newblock Linear completion.
\newblock In St{\'e}phane Kaplan and Mitsuhiro Okada, editors, {\em Proc.\ 2nd
  Workshop on Conditional and Typed Rewriting Systems}, volume 516 of {\em
  Lecture Notes in Computer Science}, pages 233--245, 1991.
\newblock \href {https://doi.org/10.1007/3-540-54317-1_94}
  {\path{doi:10.1007/3-540-54317-1_94}}.

\bibitem[DM79]{DM79}
Nachum Dershowitz and Zohar Manna.
\newblock Proving termination with multiset orderings.
\newblock {\em Communications of the {ACM}}, 22(8):465--476, 1979.
\newblock \href {https://doi.org/10.1145/359138.359142}
  {\path{doi:10.1145/359138.359142}}.

\bibitem[EH62]{EH62}
Beno Eckmann and Peter~John Hilton.
\newblock Group-like structures in general categories {I} muliplications and
  comultiplications.
\newblock {\em Mathematische Annalen}, 145(3):227--255, 1962.
\newblock \href {https://doi.org/10.1007/BF01451367}
  {\path{doi:10.1007/BF01451367}}.

\bibitem[FAS13]{FAS13}
Bertram Felgenhauer, Martin Avanzini, and Christian Sternagel.
\newblock A {Haskell} library for term rewriting.
\newblock In {\em Proc.\ 1st International Workshop on Haskell and Rewriting
  Techniques}, 2013.
\newblock \href {https://doi.org/10.48550/ARXIV.1307.2328}
  {\path{doi:10.48550/ARXIV.1307.2328}}.

\bibitem[FvO13]{FvO13}
Bertram Felgenhauer and Vincent van Oostrom.
\newblock Proof orders for decreasing diagrams.
\newblock In {Femke van} Raamsdonk, editor, {\em Proc.\ 24th International
  Conference on Rewriting Techniques and Applications}, volume~21 of {\em
  Leibniz International Proceedings in Informatics}, pages 174--189, 2013.
\newblock \href {https://doi.org/10.4230/LIPIcs.RTA.2013.174}
  {\path{doi:10.4230/LIPIcs.RTA.2013.174}}.

\bibitem[GL20]{GL20}
Ra{\'u}l Guti{\'e}rrez and Salvador Lucas.
\newblock {MU-TERM}: Verify termination properties automatically (system
  description).
\newblock In Nicolas Peltier and Viorica Sofronie-Stokkermans, editors, {\em
  Proc.\ 10th International Joint Conference on Automated Reasoning}, volume
  12167 of {\em Lecture Notes in Artificial Intelligence}, pages 436--447,
  2020.
\newblock \href {https://doi.org/10.1007/978-3-030-51054-1_28}
  {\path{doi:10.1007/978-3-030-51054-1_28}}.

\bibitem[Hir21]{H21}
Nao Hirokawa.
\newblock Completion and reduction orders.
\newblock In Naoki Kobayashi, editor, {\em Proc.\ 6th International Conference
  on Formal Structures for Computation and Deduction}, volume 195 of {\em
  Leibniz International Proceedings in Informatics}, pages 2:1--2:9, 2021.
\newblock \href {https://doi.org/10.4230/LIPIcs.FSCD.2021.2}
  {\path{doi:10.4230/LIPIcs.FSCD.2021.2}}.

\bibitem[HMSW19]{HMSW19}
Nao Hirokawa, Aart Middeldorp, Christian Sternagel, and Sarah Winkler.
\newblock Abstract completion, formalized.
\newblock {\em Logical Methods in Computer Science}, 15(3):19:1--19:42, 2019.
\newblock \href {https://doi.org/10.23638/LMCS-15(3:19)2019}
  {\path{doi:10.23638/LMCS-15(3:19)2019}}.

\bibitem[Hue80]{H80}
G{\'e}rard Huet.
\newblock Confluent reductions: Abstract properties and applications to term
  rewriting systems.
\newblock {\em Journal of the {ACM}}, 27(4):797--821, 1980.
\newblock \href {https://doi.org/10.1145/322217.322230}
  {\path{doi:10.1145/322217.322230}}.

\bibitem[JK86]{JK86}
Jean-Pierre Jouannaud and H{\'e}l{\`e}ne Kirchner.
\newblock Completion of a set of rules modulo a set of equations.
\newblock {\em SIAM Journal on Computing}, 15(4):1155--1194, 1986.
\newblock \href {https://doi.org/10.1137/0215084} {\path{doi:10.1137/0215084}}.

\bibitem[KB70]{KB70}
Donald~E. Knuth and Peter~B. Bendix.
\newblock Simple word problems in universal algebras.
\newblock In John Leech, editor, {\em Computational Problems in Abstract
  Algebra}, pages 263--297. Pergamon Press, 1970.
\newblock \href {https://doi.org/10.1016/B978-0-08-012975-4.50028-X}
  {\path{doi:10.1016/B978-0-08-012975-4.50028-X}}.

\bibitem[KH11]{KH11}
Dominik Klein and Nao Hirokawa.
\newblock Maximal completion.
\newblock In Manfred Schmidt-Schau{\ss}, editor, {\em Proc.\ 22nd International
  Conference on Rewriting Techniques and Applications}, volume~10 of {\em
  Leibniz International Proceedings in Informatics}, pages 71--80, 2011.
\newblock \href {https://doi.org/10.4230/LIPIcs.RTA.2011.71}
  {\path{doi:10.4230/LIPIcs.RTA.2011.71}}.

\bibitem[KMN88]{KMN88}
Deepak Kapur, David~R. Musser, and Paliath Narendran.
\newblock Only prime superpositions need be considered in the {K}nuth--{B}endix
  completion procedure.
\newblock {\em Journal of Symbolic Computation}, 6(1):19--36, 1988.
\newblock \href {https://doi.org/10.1016/S0747-7171(88)80019-1}
  {\path{doi:10.1016/S0747-7171(88)80019-1}}.

\bibitem[KSZM09]{KSZM09}
Martin Korp, Christian Sternagel, Harald Zankl, and Aart Middeldorp.
\newblock Tyrolean {T}ermination {T}ool 2.
\newblock In Ralf Treinen, editor, {\em Proc.\ 20th International Conference on
  Rewriting Techniques and Applications}, volume 5595 of {\em Lecture Notes in
  Computer Science}, pages 295--304, 2009.
\newblock \href {https://doi.org/10.1007/978-3-642-02348-4_21}
  {\path{doi:10.1007/978-3-642-02348-4_21}}.

\bibitem[Mar96]{M96}
Claude March{\'e}.
\newblock Normalized rewriting: An alternative to rewriting modulo a set of
  equations.
\newblock {\em Journal of Symbolic Computation}, 21(3):253--288, 1996.
\newblock \href {https://doi.org/10.1006/jsco.1996.0011}
  {\path{doi:10.1006/jsco.1996.0011}}.

\bibitem[Mes17]{M17}
Jos{\'e} Meseguer.
\newblock Strict coherence of conditional rewriting modulo axioms.
\newblock {\em Theoretical Computer Science}, 672:1--35, 2017.
\newblock \href {https://doi.org/10.1016/j.tcs.2016.12.026}
  {\path{doi:10.1016/j.tcs.2016.12.026}}.

\bibitem[M{\'e}t83]{M83}
Yves M{\'e}tivier.
\newblock About the rewriting systems produced by the {Knuth--Bendix}
  completion algorithm.
\newblock {\em Information Processing Letters}, 16(1):31--34, 1983.
\newblock \href {https://doi.org/10.1016/0020-0190(83)90009-1}
  {\path{doi:10.1016/0020-0190(83)90009-1}}.

\bibitem[NHM23a]{NHM23b}
Johannes Niederhauser, Nao Hirokawa, and Aart Middeldorp.
\newblock {Church--Rosser} modulo for left-linear {TRSs} revisited.
\newblock In Cyrille Chenavier and Sarah Winkler, editors, {\em Proc.\ 12th
  International Workshop on Confluence}, pages 14--19, 2023.

\bibitem[NHM23b]{NHM23a}
Johannes Niederhauser, Nao Hirokawa, and Aart Middeldorp.
\newblock Left-linear completion with {AC} axioms.
\newblock In Brigitte Pientka and Cesare Tinelli, editors, {\em Proc.\ 29th
  International Conference on Automated Deduction}, volume 14132 of {\em
  Lecture Notes in Artificial Intelligence}, pages 401--418, 2023.
\newblock \href {https://doi.org/10.1007/978-3-031-38499-8_23}
  {\path{doi:10.1007/978-3-031-38499-8_23}}.

\bibitem[Ohl98]{O98}
Enno Ohlebusch.
\newblock {Church--Rosser} theorems for abstract reduction modulo an
  equivalence relation.
\newblock In Tobias Nipkow, editor, {\em Proc.\ 9th International Conference on
  Rewriting Techniques and Applications}, volume 1379 of {\em Lecture Notes in
  Computer Science}, pages 17--31, 1998.
\newblock \href {https://doi.org/10.1007/BFb0052358}
  {\path{doi:10.1007/BFb0052358}}.

\bibitem[PS81]{PS81}
Gerald~E. Peterson and Mark~E. Stickel.
\newblock Complete sets of reductions for some equational theories.
\newblock {\em Journal of the {ACM}}, 28(2):233--264, 1981.
\newblock \href {https://doi.org/10.1145/322248.322251}
  {\path{doi:10.1145/322248.322251}}.

\bibitem[SK94]{S94}
Andrea Sattler-Klein.
\newblock About changing the ordering during {Knuth--Bendix} completion.
\newblock In Patrice Enjalbert, Ernst~W. Mayr, and Klaus~W. Wagner, editors,
  {\em Proc.\ 11th Annual Symposium on Theoretical Aspects of Computer
  Science}, volume 775 of {\em Lecture Notes in Computer Science}, pages
  175--186, 1994.
\newblock \href {https://doi.org/10.1007/3-540-57785-8_140}
  {\path{doi:10.1007/3-540-57785-8_140}}.

\bibitem[ST13]{ST13}
Christian Sternagel and Ren{\'e} Thiemann.
\newblock Formalizing {K}nuth--{B}endix orders and {K}nuth--{B}endix
  completion.
\newblock In {Femke van} Raamsdonk, editor, {\em Proc.\ 24th International
  Conference on Rewriting Techniques and Applications}, volume~21 of {\em
  Leibniz International Proceedings in Informatics}, pages 286--301, 2013.
\newblock \href {https://doi.org/10.4230/LIPIcs.RTA.2013.287}
  {\path{doi:10.4230/LIPIcs.RTA.2013.287}}.

\bibitem[SW15]{SW15}
Haruhiko Sato and Sarah Winkler.
\newblock Encoding dependency pair techniques and control strategies for
  maximal completion.
\newblock In Amy~P. Felty and Aart Middeldorp, editors, {\em Proc.\ 25th
  International Conference on Automated Deduction}, volume 9195 of {\em Lecture
  Notes in Computer Science}, pages 152--162, 2015.
\newblock \href {https://doi.org/10.1007/978-3-319-21401-6_10}
  {\path{doi:10.1007/978-3-319-21401-6_10}}.

\bibitem[SZ12]{SZ12}
Thomas Sternagel and Harald Zankl.
\newblock {KBCV} -- {K}nuth--{B}endix completion visualizer.
\newblock In Bernhard Gramlich, Dale Miller, and Uli Sattler, editors, {\em
  Proc.\ 6th International Joint Conference on Automated Reasoning}, volume
  7364 of {\em Lecture Notes in Artificial Intelligence}, pages 530--536, 2012.
\newblock \href {https://doi.org/10.1007/978-3-642-31365-3_41}
  {\path{doi:10.1007/978-3-642-31365-3_41}}.

\bibitem[Ter03]{T03}
Terese, editor.
\newblock {\em Term Rewriting Systems}, volume~55 of {\em Cambridge Tracts in
  Theoretical Computer Science}.
\newblock Cambridge University Press, 2003.

\bibitem[vO94]{vO94}
Vincent van Oostrom.
\newblock Confluence by decreasing diagrams.
\newblock {\em Theoretical Computer Science}, 126(2):259--280, 1994.
\newblock \href {https://doi.org/10.1016/0304-3975(92)00023-K}
  {\path{doi:10.1016/0304-3975(92)00023-K}}.

\bibitem[Win13]{W13}
Sarah Winkler.
\newblock {\em Termination Tools in Automated Reasoning}.
\newblock PhD thesis, University of Innsbruck, 2013.

\bibitem[Win19]{W19}
Sarah Winkler.
\newblock Extending maximal completion.
\newblock In Herman Geuvers, editor, {\em Proc.\ 4th International Conference
  on Formal Structures for Computation and Deduction}, volume 131 of {\em
  Leibniz International Proceedings in Informatics}, pages 3:1--3:15, 2019.
\newblock \href {https://doi.org/10.4230/LIPIcs.FSCD.2019.3}
  {\path{doi:10.4230/LIPIcs.FSCD.2019.3}}.

\bibitem[WM11]{WM11}
Sarah Winkler and Aart Middeldorp.
\newblock {AC} completion with termination tools.
\newblock In Nikolaj Bj{\o}rner and Viorica Sofronie-Stokkermans, editors, {\em
  Proc.\ 23rd International Conference on Automated Deduction}, volume 6803 of
  {\em Lecture Notes in Artificial Intelligence}, pages 492--498, 2011.
\newblock \href {https://doi.org/10.1007/978-3-642-22438-6_37}
  {\path{doi:10.1007/978-3-642-22438-6_37}}.

\bibitem[WM13]{WM13}
Sarah Winkler and Aart Middeldorp.
\newblock Normalized completion revisited.
\newblock In {Femke van} Raamsdonk, editor, {\em Proc.\ 24th International
  Conference on Rewriting Techniques and Applications}, volume~21 of {\em
  Leibniz International Proceedings in Informatics}, pages 318--333, 2013.
\newblock \href {https://doi.org/10.4230/LIPIcs.RTA.2013.319}
  {\path{doi:10.4230/LIPIcs.RTA.2013.319}}.

\bibitem[WSMK13]{WSMK13}
Sarah Winkler, Haruhiko Sato, Aart Middeldorp, and Masahito Kurihara.
\newblock Multi-completion with termination tools.
\newblock {\em Journal of Automated Reasoning}, 50(3):317--354, 2013.
\newblock \href {https://doi.org/10.1007/s10817-012-9249-2}
  {\path{doi:10.1007/s10817-012-9249-2}}.

\bibitem[WSW06]{WSW06}
Ian Wehrman, Aaron Stump, and Edwin~M. Westbrook.
\newblock Slothrop: {K}nuth--{B}endix completion with a modern termination
  checker.
\newblock In Frank Pfenning, editor, {\em Proc.\ 17th International Conference
  on Rewriting Techniques and Applications}, volume 4098 of {\em Lecture Notes
  in Computer Science}, pages 287--296, 2006.
\newblock \href {https://doi.org/10.1007/11805618_22}
  {\path{doi:10.1007/11805618_22}}.

\bibitem[YKS14]{YKS14}
Akihisa Yamada, Keiichirou Kusakari, and Toshiki Sakabe.
\newblock Nagoya {T}ermination {T}ool.
\newblock In Gilles Dowek, editor, {\em Proc.\ 25th International Conference on
  Rewriting Techniques and Applications and 12th International Conference on
  Typed Lambda Calculi and Applications}, volume 8560 of {\em Lecture Notes in
  Computer Science}, pages 466--475, 2014.
\newblock \href {https://doi.org/10.1007/978-3-319-08918-8_32}
  {\path{doi:10.1007/978-3-319-08918-8_32}}.

\end{thebibliography}

\end{document}